\documentclass[12pt]{article}
\usepackage{amsthm,amsfonts,amsmath,amssymb,amscd,mathrsfs,url,hyperref}
\usepackage{graphicx,epsfig,float}
\usepackage[usenames,dvipsnames]{color}

\title{Born's Rule for Arbitrary Cauchy Surfaces}

\author{
Matthias Lienert\footnote{Fachbereich Mathematik, Eberhard-Karls-Universit\"at, Auf der Morgenstelle 10, 72076 T\"ubingen, Germany}~\footnote{E-mail: matthias.lienert@uni-tuebingen.de}\ \ and
Roderich Tumulka$^*$\footnote{E-mail: roderich.tumulka@uni-tuebingen.de}
}
\date{September 17, 2019}

\newcommand{\be}{\begin{equation}}
\newcommand{\ee}{\end{equation}}
\newcommand{\Hilbert}{\mathscr{H}}
\newcommand{\Kilbert}{\mathscr{K}}

\newcommand{\PPP}{\mathbb{P}}
\newcommand{\RRR}{\mathbb{R}}

\newcommand{\CCC}{\mathbb{C}}
\newcommand{\M}{\mathbb{M}}

\newcommand{\QQQ}{\mathbb{Q}}
\newcommand{\ZZZ}{\mathbb{Z}}
\newcommand{\NNN}{\mathbb{N}}

\newcommand{\patch}{P}
\newcommand{\scp}[2]{\langle #1|#2 \rangle}

\newcommand{\vX}{\boldsymbol{X}}

\newcommand{\vx}{\boldsymbol{x}}
\newcommand{\vy}{\boldsymbol{y}}
\newcommand{\vz}{\boldsymbol{z}}

\newcommand{\tr}{\mathrm{tr}}

\newcommand{\sA}{\mathscr{A}}
\newcommand{\sB}{\mathscr{B}}
\newcommand{\sC}{\mathscr{C}}
\newcommand{\sD}{\mathscr{D}}
\newcommand{\sE}{\mathscr{E}}
\newcommand{\sF}{\mathscr{F}}

\newcommand{\sN}{\mathscr{N}}

\newcommand{\sS}{\mathscr{S}}

\newcommand{\sW}{\mathscr{W}}

\newcommand{\free}{\mathrm{free}}
\newcommand{\foliation}{\mathscr{F}}

\newcommand{\Prob}{\mathbb{P}}

\newcommand{\sP}{\mathscr{P}}
\newcommand{\past}{{\rm past}}
\newcommand{\future}{{\rm future}}

\newcommand{\Gr}{{\rm Gr}}
\newcommand{\Sr}{{\rm Sr}}
\newcommand{\mbskl}{{M_B(s_{k\ell})}}
\newcommand{\mbsk}{{M_B(s_k)}}

\newcommand{\nbsk}{{N_B(s_k)}}

\newcommand{\nbskl}{{N_B(s_{k\ell})}}
\newcommand{\mcskl}{{M_C(s_{k\ell})}}

\newcommand{\msll}{{M_\patch(L_\ell)}}
\newcommand{\msl}{{M_\patch(L)}}
\newcommand{\vacket}[1]{{| \emptyset(#1) \rangle}}

\newcommand{\cklcheck}{\check{C}_{k\ell}}
\newcommand{\cklhat}{\widehat{C}_{k\ell}}
\newcommand{\mcsklcheck}{\check{M}_C(s_{k\ell})}
\newcommand{\mcsklhat}{\widehat{M}_C(s_{k\ell})}
\newcommand{\mcscheck}{\check{M}_C(s)}
\newcommand{\mcshat}{\widehat{M}_C(s)}
\newcommand{\mslcheck}{\check{M}_C^\varepsilon(L)}
\newcommand{\mslhat}{\widehat{M}_C^\varepsilon(L)}

\newcommand{\ncsklcheck}{\check{N}_C(s_{k\ell})}
\newcommand{\ncsklhat}{\widehat{N}_C(s_{k\ell})}
\newcommand{\ncshat}{\widehat{N}_C(s)}
\newcommand{\mcdsklcheck}{\check{M}_{C\cup D}(s_{k\ell})}

\newcommand{\mcs}{{M_C(s)}}
\newcommand{\mslepsilon}{{M_C^\varepsilon(L)}}
\newcommand{\ckl}{{C_{k\ell}}}
\newcommand{\plcheck}{{\check{\patch}_\ell^\varepsilon}}
\newcommand{\plhat}{{\widehat{\patch}_\ell^\varepsilon}}
\newcommand{\dpepsilon}{{\partial \patch^\varepsilon}}
\newcommand{\mpllcheck}{{\check{M}_{\patch}^\varepsilon(L_\ell)}}
\newcommand{\mpllhat}{{\widehat{M}_{\patch}^\varepsilon(L_\ell)}}
\newcommand{\mplcheck}{{\check{M}_{\patch}^\varepsilon(L)}}
\newcommand{\mplhat}{{\widehat{M}_{\patch}^\varepsilon(L)}}
\DeclareMathOperator{\range}{range}

\theoremstyle{plain}
\newtheorem{thm}{Theorem}
\newtheorem{prop}{Proposition}
\theoremstyle{definition}
\newtheorem{defn}{Definition}

\newcounter{remarks}
\begin{document}
\maketitle

\begin{abstract}
Suppose that particle detectors are placed along a Cauchy surface $\Sigma$ in Minkowski space-time, and consider a quantum theory with fixed or variable number of particles (i.e., using Fock space or a subspace thereof). It is straightforward to guess what Born's rule should look like for this setting: The probability distribution of the detected configuration on $\Sigma$ has density $|\psi_\Sigma|^2$, where $\psi_\Sigma$ is a suitable wave function on $\Sigma$, and the operation $|\cdot|^2$ is suitably interpreted. We call this statement the ``curved Born rule.'' Since in any one Lorentz frame, the appropriate measurement postulates referring to constant-$t$ hyperplanes should determine the probabilities of the outcomes of any conceivable experiment, they should also imply the curved Born rule. This is what we are concerned with here: deriving Born's rule for $\Sigma$ from Born's rule in one Lorentz frame (along with a collapse rule). We describe two ways of defining an idealized detection process, and prove for one of them that the probability distribution coincides with $|\psi_\Sigma|^2$. For this result, we need two hypotheses on the time evolution: 
that there is no interaction faster than light, and that there is no propagation faster than light. The wave function $\psi_\Sigma$ can be obtained from the Tomonaga--Schwinger equation, or from a multi-time wave function by inserting configurations on $\Sigma$. Thus, our result establishes in particular how multi-time wave functions are related to detection probabilities.

\medskip

\end{abstract}

\section{Introduction}
\label{sec:intro}

\subsection{The Curved Born Rule}
\label{sec:rule1}

The usual Born rule of non-relativistic quantum mechanics for a system of $N$ particles states that, if we detect the particles at time $t$, then 
the probability to find them in infinitesimal regions $d^3\vx_i$ around $\vx_i\in \RRR^3,~i=1,\ldots,N$,  is given by
\be
	\rho_t(\vx_1,\ldots,\vx_N)\, d^3 \vx_1 \cdots d^3 \vx_N = |\psi_t(\vx_1,\ldots,\vx_N)|^2 \, d^3 \vx_1 \cdots d^3 \vx_N,
	\label{nonrelbornrule}
\ee
where $\psi_t$ is the wave function determined by Schr\"odinger's equation. 
In relativistic space-time, this formulation of Born's rule evidently refers to a particular Lorentz frame. Thus, relativity demands a generalization of the statement, a generalization in which the role of constant-time hypersurfaces is played by spacelike hypersurfaces or, more precisely, Cauchy surfaces.\footnote{A \textit{Cauchy surface} \cite{Wiki:Cauchy,Wald} is a subset $\Sigma$ of Minkowski space-time $\M$ which is intersected by every inextensible causal (i.e., non-spacelike) curve exactly once. Some authors \cite{ON:1983} use here ``timelike curve''; the difference is that we do not allow Cauchy surfaces to contain any lightlike line segment. Still, some tangent vectors to $\Sigma$ can be lightlike.}
 
Hints for how to formulate such a generalization can be obtained from the familiar Dirac equation for a single particle in Minkowski space-time $\M$. By restricting a solution $\psi(x)$, $x\in\M$, to arguments $x \in \Sigma$ on a Cauchy surface $\Sigma$, we obtain a wave function $\psi_\Sigma$ associated with $\Sigma$, and the Dirac equation can be regarded as defining a unitary evolution between Hilbert spaces associated with different Cauchy surfaces $\Sigma$ and $\Sigma'$, i.e, a unitary isomorphism \cite{Dim82}
\be
	U_{\Sigma}^{\Sigma'} : \Hilbert_\Sigma \rightarrow \Hilbert_{\Sigma'},~~~ \psi_\Sigma \mapsto \psi_{\Sigma'}.
\ee
Explicitly, the Hilbert spaces $\Hilbert_{\Sigma}$ are given by suitable $L^2$ spaces of functions $\Sigma\to\CCC^4$ with inner product
\be\label{scp1def}
	\langle \psi_\Sigma | \chi_\Sigma \rangle 
	= \int_\Sigma d^3x \, \overline{\psi}_\Sigma(x) \, \gamma^\mu  n_\mu(x)\, \chi_\Sigma(x),
\ee
where $d^3x$ is the volume relative to the 3-metric on $\Sigma$, and $n_\mu(x)$ is the future-pointing unit normal vector field on $\Sigma$. The previous wave function $\psi_t(\vx)$ is contained in this scheme as the wave function $\psi_{\Sigma_t}(x)$ associated with the horizontal surface
\be
\Sigma_t = \{ x \in \M : x^0 = t\}\,.
\ee
In this case, $n_\mu(x) = (1,0,0,0)$; using $\overline{\psi} \gamma^0 = \psi^\dagger$, we obtain that
\be
  \langle \psi_{\Sigma_t} | \psi_{\Sigma_t}\rangle 
  = \int_{\RRR^3}d^3\vx\, |\psi_t(\vx)|^2 = \int_{\RRR^3} d^3\vx\, \rho_t(\vx)
\ee
with $|\psi|^2 := \sum_{s=1}^4 |\psi_s|^2$, summing over the spin components of $\psi$. 
This suggests that the natural analog of $\rho_t(\vx)$ for any Cauchy surface $\Sigma$, in fact the probability density on $\Sigma$, is given by
\be\label{rhoSigma1}
	\rho_\Sigma(x) = \overline{\psi}_{\Sigma}(x) \, \gamma^\mu n_\mu(x) \,  \psi_{\Sigma}(x)\,.
\ee
The right-hand side is non-negative and, in fact, equal to $|\psi_\Sigma(x)|^2$ relative to the basis of spin space associated with the Lorentz frame tangent to $\Sigma$ at $x$. Put differently, the right-hand side of \eqref{rhoSigma1} is ``$|\psi|^2$ suitably interpreted.'' (For example, this distribution was taken for granted in the covariant flux-across-surfaces theorem \cite[Eq.~(10)]{DP03}.)

For several Dirac particles, the natural analog reads
\be\label{rhoSigmaN}
\rho_\Sigma(x_1,\ldots,x_N) = \overline{\psi}_{\Sigma}(x_1,\ldots,x_N) \, \bigl[\gamma^{\mu_1}\, n_{\mu_1}(x_1) \otimes \cdots \otimes \gamma^{\mu_N}\, n_{\mu_N}(x_N) \bigr] \, \psi_\Sigma(x_1,\ldots,x_N)
\ee
with $\psi_{\Sigma}:\Sigma^N \to (\CCC^4)^{\otimes N}$, and
\begin{multline}\label{scpNdef}
\langle \psi_\Sigma|\chi_\Sigma \rangle 
= \int_{\Sigma^N} d^3x_1\cdots d^3x_N \, \overline{\psi}_\Sigma(x_1,\ldots,x_N) \: \times\\
\times\: \bigl[\gamma^{\mu_1}\, n_{\mu_1}(x_1) \otimes \cdots \otimes \gamma^{\mu_N}\, n_{\mu_N}(x_N) \bigr] \, \chi_\Sigma(x_1,\ldots,x_N)\,.
\end{multline}
This suggests a generalized version of Born's rule in relativistic space-time that we can roughly formulate as follows.

\paragraph{The curved Born rule.} 
{\it If we place ideal detectors along an arbitrary (curved) Cauchy surface $\Sigma$, then the probability distribution of the detected particle configuration has density (relative to the volume defined by the 3-metric on $\Sigma$) given by
\be
\rho_\Sigma(x_1,\ldots,x_N)= |\psi_\Sigma(x_1,\ldots,x_N)|^2\,,
\ee
with $|\cdot|^2$ suitably interpreted as in \eqref{rhoSigmaN} and $\psi_\Sigma$ the wave function associated with $\Sigma$. We call this distribution the ``curved Born distribution.''}

\bigskip

This formulation is meant to include the possibility of a variable particle number $N$, as may arise from the creation and annihilation of particles, with nonzero probabilities for several values of $N$. To this end, we may take both $\rho_\Sigma$ and $\psi_\Sigma$ to be defined on a configuration space of a variable number of particles, as we will define in \eqref{Gammadef} below, corresponding, e.g., to a wave function from Fock space. 

Moreover, we expect the curved Born rule to be equally valid in curved space-time, although we limit our considerations here to Minkowski space-time.

Two questions arise.
1. Can we say with reasonable generality how $\psi_\Sigma$ should be defined if we are given a Hamiltonian and an initial wave function on the surface $\Sigma_0=\{x^0=0\}$?, and 2. Can we prove (or derive) the curved Born rule? In this paper, we will say a bit about the first question (more in subsequent work) and a lot about the second.

Since the curved Born rule is a natural statement and easy to guess, it is tempting to simply replace the usual, ``horizontal'' Born rule (i.e., for horizontal hyperplanes) by the curved Born rule. However, while the horizontal Born rule is usually introduced as a postulate, it is not possible to introduce the curved Born rule as a postulate. That is because as soon as we have a rule for probabilities on horizontal hypersurfaces (the horizontal Born rule), the statistics of outcomes of \emph{any} experiment is determined; for example, the theoretical analysis of an experiment could include the unitary interaction between the apparatus and the object, and the outcome could be read off through a quantum measurement of the pointer (or any other display) of the apparatus on a horizontal hypersurface $\Sigma_t$ at a sufficiently late time $t$. Thus, if we assume the horizontal Born rule, then the curved Born rule is \emph{either false or a theorem}. In particular, if the curved Born rule is correct, then it requires a proof. Here, we provide such a proof.

We note that results such as Malament's theorem \cite{Heg:1998,HC:2002} put limitations on the existence of 
ideal particle detectors under conditions related to the non-occurrence of negative energies. We regard our results as a first step for simple model quantum field theories (examples will be given below) where we do not worry about negative energies, so Malament's theorem does not apply. In the case of exclusively positive energies we expect that it implies fundamental limitations to the accuracy of detectors, which complicate the situation. Here, we leave such complications aside and assume the existence of ideal detectors on horizontal surfaces. We focus on the study of ideal detectors as a question of interest in its own right that needs to be covered first. Our results apply to a certain class of quantum theories (defined below), and since relevant quantum field theories lie outside of that class, it would certainly be of interest to investigate in future work how the curved Born rule could be proven beyond that class.

Our main goal in this paper is to justify the curved Born rule from \textit{equal-time} measurement rules. As we will explain shortly, several different such justifications are possible, corresponding to different concepts of ideal detectors and different ways of approaching a curved surface $\Sigma$ as a limit.

\subsection{The Curved Born Rule as a Theorem}

Let us outline the mathematical framework (to be detailed in Sections~\ref{sec:dynamics} and \ref{sec:detectors}) in which the curved Born rule can be formulated as a theorem. It comprises three ingredients: 
1. Hilbert spaces and their relation to space-time configurations via a PVM (projection-valued measure \cite{RS1}), 2. properties of the dynamics, and 3. measurement postulates.

\subsubsection{Hilbert Space and PVM}

As the first ingredient, we assume that a Hilbert space $\Hilbert_\Sigma$ is associated with every Cauchy surface $\Sigma$. These Hilbert spaces may contain states of one or several species of particles, and of a fixed or variable number of particles (such as for a Fock space). We construct a configuration space $\Gamma(\Sigma)$ for every $\Sigma$ as follows. For any set $R$, let
\be\label{Gammadef}
   \Gamma(R) := \{q \subset R: \# q< \infty\}
\ee
be the set of all finite subsets, and
\be\label{Gammandef}
  \Gamma_n(R) := \{q \subset R: \# q=n\} \subset \Gamma(R)
\ee
its $n$-particle sector for $n\in \NNN_0$. We regard $\Gamma(\Sigma)$ as the set of all unordered configurations of a variable number of particles in $\Sigma$. (For some purposes, it is convenient to use ordered configurations $(x_1,\ldots,x_n)$ and for others unordered ones $\{x_1,\ldots,x_n\}$; see Section~\ref{sec:config} for comments about switching between the two.) For what follows, we note that for any disjoint sets $A,B$,
\be\label{GammaAB}
\Gamma(A\cup B) \cong \Gamma(A) \times \Gamma(B)
\ee
in the sense that there is an (obvious) canonical identification mapping $q\mapsto (q\cap A, q\cap B)$ between the two spaces. In the following, we will make this identification whenever convenient and simply write $=$ instead of $\cong$.

We assume further that for every Cauchy surface $\Sigma$ we are given a PVM $P_\Sigma$ on the configuration space $\Gamma(\Sigma)$ acting on the Hilbert space $\Hilbert_\Sigma$; $P_\Sigma$ represents the ``configuration observable.''\footnote{Should 
$\Hilbert_\Sigma$ contain states of several (say, $m$) species of particles, it would also be possible to take $\Gamma(\Sigma)^m$ as the configuration space, but for simplicity we do not distinguish between different species in the configuration and take $P_\Sigma$ to be suitably projected so as to be defined on $\Gamma(\Sigma)$, as described in Remark~\ref{rem:species} in Section~\ref{sec:hypersurfaceevolution} below.}  Given any $\psi_\Sigma\in \Hilbert_\Sigma$ with $\|\psi_\Sigma\|=1$, the probability measure\footnote{All sets considered in this paper are measurable, and when we speak about ``any subset'' we mean ``any measurable subset.'' The relevant $\sigma$-algebras are specified in the beginning of Section~\ref{sec:dynamics}.}
\be
\PPP^{\psi_\Sigma} (S) = \|P_\Sigma(S)\, \psi_\Sigma\|^2~~~\forall S\subseteq \Gamma(\Sigma)
\ee
on configuration space $\Gamma(\Sigma)$ is the general version of the $|\psi|^2$ distribution as in
\be
\PPP^{\psi_\Sigma}(S) = \int_S dq\,  |\psi_\Sigma(q)|^2~~~\forall S\subseteq \Gamma(\Sigma)\,,
	\label{eq:pvmdensity}
\ee
as elucidated further in Section~\ref{sec:hypersurfaceevolution} (see particularly Remarks~\ref{rem:Fock} and \ref{rem:naturalization}). For example, if $\Hilbert_{\Sigma}^{(1)}$ is the 1-particle Hilbert space of Dirac wave functions $\Sigma\to\CCC^4$ with inner product \eqref{scp1def}, and if $\Hilbert_\Sigma=\Gamma_{-}(\Hilbert_{\Sigma}^{(1)})$ is the fermionic Fock space over $\Hilbert_{\Sigma}^{(1)}$, then $\Hilbert_\Sigma$ is naturally equipped with such a PVM $P_\Sigma$. In quantum field theories, $P_\Sigma$ could be regarded as arising from the simultaneous diagonalization of all particle number operators, which in turn could be defined in terms of the field operators \cite[Sec.~6.8]{DGTZ:2005b}.

\subsubsection{Unitary Dynamics}

As the second ingredient, the dynamics is given by unitary isomorphisms $U_{\Sigma}^{\Sigma'} : \Hilbert_\Sigma \rightarrow \Hilbert_{\Sigma'}$ for every pair of Cauchy surfaces $\Sigma, \Sigma'$; that is,
\be
 \psi_{\Sigma'} = U^{\Sigma'}_{\Sigma} \, \psi_\Sigma\,.
\ee
We require that
\be
U_{\Sigma'}^{\Sigma''} \, U_{\Sigma}^{\Sigma'} = U_{\Sigma}^{\Sigma''}\quad \text{and} \quad
U_{\Sigma}^{\Sigma} = I_\Sigma
\ee 
(with $I_\Sigma$ the identity on $\Hilbert_\Sigma$). 

The hypersurface evolution mappings $U_{\Sigma}^{\Sigma'}$  are, more or less, the same as the ones defined by the well-known Tomonaga-Schwinger equation, except that the latter is formulated in the interaction picture, whereas a hypersurface evolution can correspond as well to the Schr\"odinger or the Heisenberg picture; see Remark~\ref{rem:unitaryequivalence} in Section~\ref{sec:hypersurfaceevolution} and Section~\ref{sec:TS} for details.
Moreover, many relevant examples of mappings $U_{\Sigma}^{\Sigma'}$ can be extracted from the dynamics of \textit{multi-time wave functions} \cite{dfp:1932,bloch:1934,DV82b,pt:2013c,lienert:2015a}; an overview is given in \cite{LPT:2017}; see also Section~\ref{sec:MT} below. A multi-time wave function is, in the variable particle number case, a spinor-valued function $\phi$ on the set of spacelike configurations, i.e., (if we regard configurations as unordered) on
\be\label{sSdef}
	\sS = \Bigl\{ q \subset \M : \#q<\infty \text{ and }
	(x-y)^2<0\text{ for all } x, y \in q \text{ with }x\neq y\Bigr\},
\ee 
where $(x-y)^2=(x^{\mu}-y^{\mu})(x_{\mu}-y_{\mu})$ denotes the Minkowski square (using the metric signature ${+}{-}{-}{-}$). Denoting the restriction of a function $f$ to arguments in the set $X$ by $f|_X$, 
$\psi_\Sigma$ is given by
\be\label{psiSigmaphi}
\psi_\Sigma = \phi|_{\Gamma(\Sigma)}\,,
\ee
and $U_{\Sigma}^{\Sigma'}$ is given by the map $\phi|_{{\Gamma(\Sigma)}} \mapsto\phi|_{{\Gamma(\Sigma')}}$. Thus, the curved Born rule expresses the probabilities on $\Sigma$ in terms of the multi-time wave function $\phi$. The present paper was largely motivated by the question whether detection probabilities can be expressed in this way.

\bigskip

We say that a \emph{hypersurface evolution} is defined by specifying a family of Hilbert spaces $\Hilbert_{\Sigma}$, a family of PVMs $P_\Sigma$, and a family of time evolution operators $U_\Sigma^{\Sigma'}$ as described with the following three additional properties:
\begin{itemize}
\item[(i)] {\it $P_\Sigma$ is absolutely continuous}:
\be\label{nullset}
P_\Sigma(S)=0
\ee
for every set $S\subset \Gamma(\Sigma)$ of measure zero (see Section~\ref{sec:dynamics} for details). 

\item[(ii)] {\it Unique vacuum}: 
\be\label{UV}
\text{The vacuum subspace, $\range P_\Sigma(\{\emptyset\})$, is 1-dimensional.}
\ee
Note that $P_\Sigma(\emptyset)=0$ whereas $P_\Sigma(\{\emptyset\})$ is the projection associated with the vacuum configuration; that is because $\emptyset$ as a subset of $\Gamma(\Sigma)$ contains no configuration whereas $\emptyset\in\Gamma(\Sigma)$ denotes the 0-particle configuration. Note also that while the word ``vacuum'' is often used for ``ground state,'' we use it here for states with particle number 0, such as the Fock vacuum. 
\item[(iii)] \textit{Factorization of the PVM.} For any Cauchy surface $\Sigma$ and any subset $A\subseteq \Sigma$, there is an associated Hilbert space $\Hilbert_A$ and a PVM $P_A$ on $\Gamma(A)$ acting on $\Hilbert_A$ such that, for any $A,B\subseteq \Sigma$ with $A\cap B=\emptyset$,
	\be\label{HilbertAB}
	\Hilbert_{A\cup B} \cong \Hilbert_A \otimes \Hilbert_B
	\ee
	(where $\cong$ means unitarily equivalent; we will make this identification whenever convenient and simply write $=$ instead of $\cong$) and
	\be\label{PAB}
	P_{A\cup B}(S_A\times S_B) = P_A(S_A) \otimes P_{B}(S_B)
	\ee
	for all $S_A\subseteq \Gamma(A)$ and $S_B\subseteq \Gamma(B)$. 
\end{itemize}

All three conditions are familiar from Fock spaces (see Section~\ref{sec:hypersurfaceevolution} for details); in particular, they are fulfilled if $\Hilbert_{\Sigma}$ is the Fock space $\Gamma_{-}(\Hilbert^{(1)}_\Sigma)$ mentioned above with the natural PVM. 
From (ii) and (iii) it follows that also in every $\Hilbert_A$, the range of $P_A(\{\emptyset\})$ is 1-dimensional, as $P_{\Sigma}(\{\emptyset\})=P_A(\{\emptyset\}) \otimes P_{\Sigma\setminus A}(\{\emptyset\})$.

We abbreviate a given hypersurface evolution as $\sE=(\Hilbert_\circ,P_\circ,U_\circ^{\circ})$, with the $\circ$ symbol as a placeholder for Cauchy surfaces. In terms of a hypersurface evolution, the curved Born rule can abstractly be stated as saying that a configuration in the set $S\subseteq \Gamma(\Sigma)$ is found with probability
\be
\PPP^\psi(S) = \bigl\|P_\Sigma(S)\, U^\Sigma_{\Sigma_0} \psi_0 \bigr\|^2
\ee
with $\psi_0$ the initial datum on $\Sigma_0$.

\bigskip

For deriving the curved Born rule, we need the following two assumptions on the hypersurface evolution, (IL) and (PL).

\begin{figure}[th]
\centering
 \includegraphics[width=0.6\textwidth]{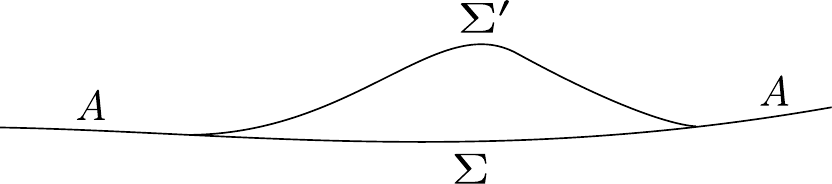}
 \caption{Example of the hypersurfaces appearing in the definition of ``interaction locality.''}
 \label{fig:interaction_locality}
\end{figure}

\begin{enumerate}
	\item[(IL)] \textit{Interaction locality.} Let $\Sigma, \Sigma'$ be two Cauchy surfaces and $A \subseteq \Sigma \cap \Sigma'$ (see Figure~\ref{fig:interaction_locality}). Then
	the evolution operator $U_{\Sigma}^{\Sigma'}$ acts as the identity on $\Hilbert_A$, i.e., there exists a unitary isomorphism $U_{\Sigma \backslash A }^{\Sigma' \backslash A}$ such that
	\be
		U_{\Sigma}^{\Sigma'} = I_{A} \otimes U_{\Sigma \backslash A }^{\Sigma' \backslash A}.
		\label{eq:il}
	\ee
	The operator $U_{\Sigma \backslash A }^{\Sigma' \backslash A}$ does not depend on $A$ except through $\Sigma\setminus A$ and $\Sigma'\setminus A$; that is, if $\tilde\Sigma$ and $\tilde\Sigma'$ are two further Cauchy surfaces, $\tilde{A} \subseteq \tilde\Sigma\cap \tilde\Sigma'$, $\tilde\Sigma\setminus \tilde A = \Sigma \setminus A$, and $\tilde\Sigma'\setminus \tilde A= \Sigma'\setminus A$, then $U^{\tilde\Sigma'\setminus \tilde A}_{\tilde\Sigma \setminus \tilde A} = U^{\Sigma'\setminus A}_{\Sigma \setminus A}$.
\end{enumerate}

Intuitively, (IL) expresses that there is no interaction between spacelike separated regions. Our second assumption, formulated as (PL) below, characterizes in terms of a given hypersurface evolution what it means to say that wave functions propagate no faster than light. To formulate (PL), we need to prepare with a couple of definitions.

\begin{figure}[tp]
\centering
 \includegraphics[width=0.6\textwidth]{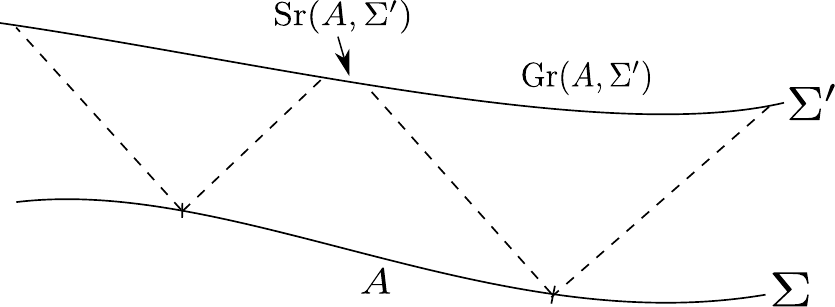}
 \caption{Illustration of grown and shrunk sets of $A \subset \Sigma$ with respect to $\Sigma'$.}
 \label{fig:grown_shrunk_sets}
\end{figure}

\begin{defn}
	Let $\Sigma, \Sigma'$ be Cauchy surfaces and $A \subseteq \Sigma$. We then define the \textit{grown set} of $A$ in $\Sigma'$ as (see Figure~\ref{fig:grown_shrunk_sets})
	\be
		\Gr(A,\Sigma') = [\future(A) \cup \past(A)] \cap \Sigma'.
	\ee
	Similarly, we define the \textit{shrunk set} of $A$ in $\Sigma'$ as:
	\be
		\Sr(A,\Sigma') = \{ x' \in \Sigma' : \Gr(\{ x'\}, \Sigma) \subseteq A\}.
	\ee
\end{defn}

	Since $\future(A)\cup \past(A)$ is also called the ``domain of influence'' of $A$, $\Gr(A,\Sigma')$ is the intersection of the domain of influence of $A$ with $\Sigma'$. Likewise,
	$\Sr(A,\Sigma')$ is the intersection of the domain of dependence of $A$ \cite{Wald} with $\Sigma'$. In particular, $\Sr(A,\Sigma')$ does not depend on which Cauchy surface $\Sigma$ we regard $A$ as a part of. Neither does $\Gr(A,\Sigma')$. Note also that 
	\be
	\Sr(A,\Sigma') = \Gr(A^c,\Sigma')^c = \Sigma' \setminus \Gr(\Sigma\setminus A, \Sigma')\,.
	\ee
Finally, we note that in curved space-time, where $\Gr(A,\Sigma')$ and $\Sr(A,\Sigma')$ can be defined in the same way, the name ``grown'' should not be taken to imply that $\Gr(A,\Sigma')$ had larger diameter in terms of the metric than $A$. 

For any $A\subseteq \Sigma$, let $\forall(A)$ denote the set of configurations on $\Sigma$ for which all particles lie in $A$,
\be
\forall(A) = \{ q \in \Gamma(\Sigma) : q \subseteq A \}\,.
\ee

\begin{defn}
	Given a hypersurface evolution $\sE$, we say that $\psi_\Sigma \in \Hilbert_\Sigma$ is \textit{concentrated in} $A \subseteq \Sigma$ iff
	\be
		P_\Sigma(\forall(A)) \psi_\Sigma = \psi_\Sigma\,,
	\ee
	that is, iff $\psi_\Sigma \in \range P_{\Sigma}(\forall(A))$.
\end{defn}

We are now ready to formulate our second assumption on the given hypersurface evolution $\sE$:

\begin{enumerate}
\item[(PL)] \textit{Propagation locality.} The following is true for all Cauchy surfaces $\Sigma, \Sigma'$ and all 
	subsets $A \subseteq \Sigma$: whenever $\psi_\Sigma$ is concentrated in $A$, then $\psi_{\Sigma'} = U_{\Sigma}^{\Sigma'} \psi_\Sigma$ is concentrated in $\Gr(A,\Sigma')$.
\end{enumerate}

We note that (PL) can equivalently be characterized by saying that the support of a wave function in configuration space grows at most at the speed of light. (This condition needs a careful formulation; in the case of conserved particle number, such a formulation is given in \cite[Section~7.1]{pt:2013a}.)

Non-trivial interacting models which satisfy these two postulates (when re-formulated as a hypersurface evolution) have, for example, been developed in the context of multi-time wave functions \cite{lienert:2015a,lienert:2015c,LN:2015,pt:2013c,pt:2013d}; see Section~\ref{sec:examples} for more detail.

\subsubsection{Detection Process}

As the third and final ingredient of the derivation of the curved Born rule, we need to give an appropriate definition of a \textit{detection process} along $\Sigma$. This term is used here to refer to the alternating use of unitary dynamics (between ``detections'') and the projection postulate (for the ``detections''). Our goal is to justify the curved Born rule using measurement postulates only for equal times, i.e., horizontal surfaces. We therefore need to approximate $\Sigma$ by flat pieces of hypersurfaces in a suitable way. Once a definition of a detection process is given, the curved Born rule can be obtained as a theorem.
It turns out that there are at least two different ways of defining a detection process, both leading to the curved Born rule:

\begin{enumerate}
	\item Approximating $\Sigma$ by flat pieces of hyperplanes as in Figure~\ref{fig:detectiondefs}(a), and considering detectors on the flat pieces.
	\item Approximating $\Sigma$ by horizontal pieces relative to a preferred Lorentz frame as in Figure~\ref{fig:detectiondefs}(b), and considering detectors on the horizontal pieces.
\end{enumerate}

\begin{figure}[ht]
\centering
 \includegraphics[width=0.8\textwidth]{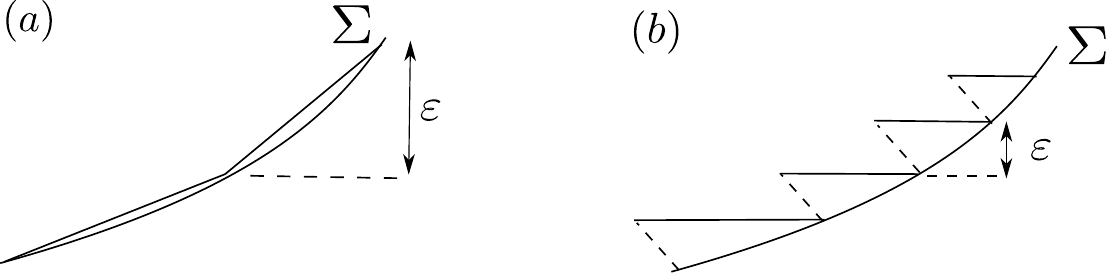}
 \caption{Possible discretization schemes of $\Sigma$. The time dimension is drawn vertically, space horizontally. (a) Approximation by arbitrary pieces of hyperplanes. (b) Approximation by pieces of equal-time surfaces. The dashed lines at $45^{\circ}$ correspond to light cones.}
 \label{fig:detectiondefs}
\end{figure}

The first approach shall be explored in subsequent work and only briefly outlined here. It corresponds more to the way detectors work physically. In the first approach we assume, in addition to (IL) and (PL), invariance under the Poincar\'e group (i.e., Lorentz transformations and space-time translations). Specifically, we suppose that the Born rule and its associated collapse rule hold, not only on the horizontal hyperplanes $\Sigma_t$, but in every Lorentz frame and thus on every (tilted) hyperplane. More precisely, we suppose that we can 
carry out a quantum measurement of whether or not the observed system contains any particles in a chosen region $R$ of any hyperplane.\footnote{As we will show in Proposition~\ref{prop:cdc}, a probability distribution on the configuration space $\Gamma(\Sigma)$ is already determined by the probabilities of the events that there is a particle in the region $R\subseteq \Sigma$.} 
Then (IL) and (PL) will allow us to compute also for any piecewise-flat Cauchy surface [as the approximating hypersurface in Figure~\ref{fig:detectiondefs}(a)] the joint probability distribution of all detection results, and thus to conclude a Born rule for such surfaces. Next, a limiting process will extend the Born rule to every Cauchy surface. 

The second approach, which we study in this paper, is based on approximating $\Sigma$ by \emph{horizontal} pieces $B_k$ relative to a preferred frame as in Figure~\ref{fig:detectiondefs}(b). Here, 
we assume that for a subset $B_{k\ell}$ of a horizontal piece $B_k$, we can carry out an ideal (projective) quantum measurement of whether there is a particle in $B_{k\ell}$, with the probabilities of outcomes given by the Born rule on $B_k$ and the corresponding collapse rule for the post-measurement wave function. As $\cup_k B_{k\ell}$ approaches a set $\patch_\ell\subseteq \Sigma$, we define the probability $p_\ell$ of finding a particle in $\patch_\ell$ to be the limiting probability of finding a particle in $B_{k\ell}$ for any $k$. This approach has the advantage of deriving the curved Born rule directly from the ``horizontal Born rule'' and the horizontal collapse rule. A difficulty here that requires some attention is that the same particle could be registered twice by different detectors, say on $B_{k\ell}$ and $B_{k+1,\ell+1}$, as depicted in Figure~\ref{fig:double}. 
We make it part of the definition of $p_\ell$ that double detections are excluded. We exclude them by tracing out any particles after they have passed one of the $B_k$; physically, this would more or less correspond to assuming that, when detected, a particle gets marked (changes its state in a certain way), that marked and unmarked particles do not interact, and that the future detectors ignore marked particles. That is a clean way of avoiding double detections.\footnote{We conjecture that these double detections do not actually change the probability of finding a particle in $\patch_\ell$ in the limit $\varepsilon\to 0$. But at present we do not have a proof for that.}

Theorem~\ref{thm:prob} asserts that with this definition of ``detection probabilities,'' they agree with the curved Born rule. A feature of the second approach not shared with the first is that the definition of the detection probabilities is given solely in terms of horizontal surfaces. Another feature of the second approach is that it does not assume Lorentz invariance of the dynamics, only (IL) and (PL). In particular, it does not assume translation invariance of the time evolution and can thus admit, e.g., time-dependent external fields.

\begin{figure}[ht]
\centering
 \includegraphics[width=0.35\textwidth]{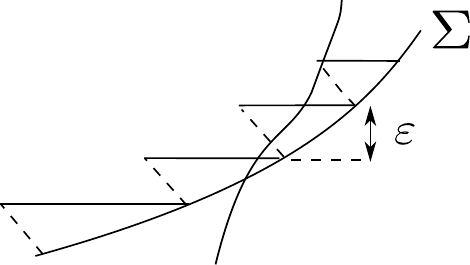}
 \caption{For the second discretization scheme, a time-like curve can cross several of the detection surfaces, i.e., a particle can get detected several times.}
 \label{fig:double}
\end{figure}

\subsubsection{Main Result}

We are now ready to give a first, informal statement of our main result, Theorem~\ref{thm:prob} (based on the second notion of detection process). Let $\Sigma_{k\varepsilon}$ denote the horizontal hyperplane at time $t=k\varepsilon$. Starting from a partition $\sP=(\patch_1,\ldots,\patch_r)$ of $\Sigma$ (with some technical assumptions and modifications described in Section~\ref{sec:detectors}) and defining $B_{k\ell}$ as the piece of $B_k\subseteq \Sigma_{k\varepsilon}$ vertically above $\patch_\ell$, we set $L_\ell=1$ if a particle gets detected in $B_{k\ell}$ for any $k$ and $L_\ell=0$ otherwise, and regard $L=(L_1,\ldots,L_r)$ as the outcome of the experiment. The joint probability distribution of $L$ is denoted $\Prob^{\psi_0,\varepsilon}_{\det, \sP}(L)$, and $\msl$ denotes the set of configurations in $\Gamma(\Sigma)$ such that, for each $\ell=1,\ldots,r$, there is no point in $\patch_\ell$ if $L_\ell=0$ and at least one point in $\patch_\ell$ if $L_\ell=1$.  
The theorem can be summarized as follows.\footnote{In this paper, the ``future'' of a set $R$ in space-time means the \emph{causal future}, often denoted $J^+(R)$ \cite{ON:1983}, as opposed to the timelike future $I^+(R)$; note that $R\subseteq J^+(R)$; likewise for the ``past.''}

\begin{thm}[informal statement]\label{thm:informal}
	Let $\Sigma$ be a Cauchy surface in the future of $\Sigma_0=\{x^0=0\}$ in Minkowski space-time, $\psi_0\in\Hilbert_{\Sigma_0}$ with $\|\psi_0\|=1$, and $(\Hilbert_{\circ},P_\circ,U_{\circ}^{\circ})$ a hypersurface evolution satisfying (IL) and (PL).  Then for any (admissible) partition $\sP$ of $\Sigma$, the detection probabilities $\Prob^{\psi_0,\varepsilon}_{\det, \sP}(L)$ converge in the limit $\varepsilon \rightarrow 0$ to those given by the curved Born distribution, i.e.,
\be
	\lim_{\varepsilon\to 0} \Prob^{\psi_0,\varepsilon}_{\det, \sP}(L) 
	=  \bigl\| P_\Sigma(\msl) \psi_\Sigma \bigr\|^2 
	=  \int\limits_{\msl} \!\!\!\! dq\, |\psi_\Sigma(q)|^2 \,.
	\label{eq:generalizedbornrule}
\ee
\end{thm}

The full technical statement of Theorem~\ref{thm:prob} is given in Section~\ref{sec:thm}. It turns out that the events $\msl$ for different choices of $\patch_\ell$ determine the distribution $\|P_\Sigma(\cdot)\psi_\Sigma\|^2$ uniquely. The proof of Theorem~\ref{thm:prob} makes no special use of dimension $3+1$ and applies equally in dimension $d+1$ for any $d\in\NNN$.

In Theorem~\ref{thm:prob}, we regard a hypersurface evolution as given, although one often starts out from just a Hamiltonian that defines the evolution only between horizontal surfaces $\Sigma_t$. In subsequent work, we address the question of existence and uniqueness of a hypersurface evolution satisfying (IL) and (PL) extending the evolution between horizontal surfaces defined by a given Hamiltonian. On the other hand, Theorem~\ref{thm:informal} alone already suggests the uniqueness of the hypersurface evolution because, if different ones existed, they could not be expected to lead to the same distribution over $\Gamma(\Sigma)$ for every $\Sigma$.

\bigskip

This paper is structured as follows. In Section~\ref{sec:rem}, we make a number of remarks that put our results in perspective. In Section~\ref{sec:dynamics}, we provide a mathematical discussion of the new concept of a hypersurface evolution $(\Hilbert_\circ,P_\circ,U_\circ^{\circ})$ and the properties (IL) and (PL). In Section~\ref{sec:examples}, we provide examples of such evolutions, starting from the free Dirac evolution, the Tomonaga-Schwinger equation, and multi-time equations. In Section~\ref{sec:detectors},  we lay out in detail the definition of ``ideal detector'' used here. 
In Section~\ref{sec:thm}, a precise technical formulation of Theorem~\ref{thm:prob} is given. Section~\ref{sec:proofs} is concerned with its proof. 
In Section~\ref{sec:conclusions}, we conclude.

\section{Remarks}\label{sec:rem}

\begin{enumerate}
\setcounter{enumi}{\theremarks}
	\item\label{rem:non-interacting} Here is a Corollary to Theorem~\ref{thm:prob}. Consider $N$ species of particles without interspecies interaction, a hypersurface evolution for each species, and $N$ Cauchy surfaces $\Sigma^{(i)}$. Suppose that detectors are placed along $\Sigma^{(i)}$ that detect only particles of species $i$  but do not interact with particles of other species. (Note that some points of $\Sigma^{(i)}$ may lie in the future of some points of $\Sigma^{(j)}$, $j\neq i$.) Then it can be shown (although we do not give details in this paper) that the curved Born rule holds in the following form: The joint probability distribution $\PPP$ of the detected configurations $q^{(i)}\in \Gamma(\Sigma^{(i)})$, $i=1,\ldots,N$, is given by
	\be
	\PPP(S) = \Bigl\|P_{\Sigma^{(1)} \ldots \Sigma^{(N)}}(S) \: U^{\Sigma^{(1)}}_{\Sigma_0} \!\otimes \cdots\otimes U^{\Sigma^{(N)}}_{\Sigma_0}\: \psi_0 \Bigr\|^2
	\ee
	for any $S\subseteq \Gamma(\Sigma^{(1)})\times \cdots \times \Gamma(\Sigma^{(N)})$ and $\psi_0\in\Hilbert_{\Sigma_0}^{(1)}\otimes \cdots \otimes \Hilbert_{\Sigma_0}^{(N)}$, where $P_{\Sigma^{(1)} \ldots\Sigma^{(N)}}$ is the product PVM \cite[Corollary~7]{DGTZ:2005} of $P_{\Sigma^{(1)}},\ldots,P_{\Sigma^{(N)}}$, i.e., defined by
	\be
	P_{\Sigma^{(1)} \ldots\Sigma^{(N)}}(S_1\times \cdots \times S_N) = P_{\Sigma^{(1)}}(S_1) \otimes \cdots \otimes P_{\Sigma^{(N)}}(S_N)\,.
	\ee
	\item Not all models of interest involve particle creation or annihilation, as some have a fixed particle number $N$. This remark is about how to treat such models. Since axiom (iii), factorization of the PVM, requires that $\Hilbert_\Sigma \cong \Hilbert_A \otimes \Hilbert_{\Sigma\setminus A}$, and since, depending on the positions of the particles, the number of particles in $A\subseteq \Sigma$ can vary, $\Hilbert_A$ needs to be a Hilbert space of a variable number of particles, and so does $\Hilbert_{\Sigma\setminus A}$. Thus, being the tensor product, also $\Hilbert_\Sigma$ needs to allow for a variable number of particles (in fact, as follows from considering smaller and smaller $A$, including arbitrarily large numbers of particles, as in a Fock space). That is, the hypersurface evolution must be defined for arbitrary numbers of particles, while the actual initial state $\psi_0$ may well be concentrated on one particular number $N$ with the consequence that, if the hypersurface evolution involves no creation or annihilation, $\psi_\Sigma$ is an $N$-particle state on every $\Sigma$.
	\item {\it Curved collapse rule.} Associated with the usual (``horizontal'') Born rule \eqref{nonrelbornrule} is a collapse rule that we use in our definition of detection probabilities on $\Sigma$ and that can be formulated as follows. \emph{If a detector is applied at time $t$ to a system with wave function $\psi_t$ and tests only whether the configuration lies in a certain set $S\subseteq \Gamma(\RRR^3)$ in configuration space, then the wave function immediately after $t$ is either $\sN_1\, P(S)\, \psi_t$ or $\sN_0\, (I-P(S))\, \psi_t$, depending on whether the outcome was yes or no.} (Here, $P(\cdot)$ is the PVM on $\Gamma(\RRR^3)$ representing the configuration observable, and the normalizing constants are $\sN_1 = \|P(S)\, \psi_t\|^{-1}$ and $\sN_0 = \|(I-P(S))\, \psi_t \|^{-1}$.)
	
	It is obvious to guess the statement of a similar collapse rule for arbitrary (curved) Cauchy surfaces $\Sigma$: \emph{If a detector is applied on $\Sigma$ to a system with wave function $\psi_\Sigma$ and tests only whether the configuration lies in a certain set $S\subseteq \Gamma(\Sigma)$, then the wave function immediately after $\Sigma$ is either $\sN_{1\Sigma}\, P_\Sigma(S)\, \psi_\Sigma$ or $\sN_{0\Sigma}\, (I_\Sigma-P_\Sigma(S))\, \psi_\Sigma$, depending on whether the outcome was yes or no.}
	
	The reason why we make no attempt in this paper to prove such a rule for $\Sigma$ is that our concepts of detection process are such that the detectors will obtain very detailed information about where on $\Sigma$ particles can be found---in some directions with small inaccuracy $\varepsilon$ tending to zero. Much of the information will be discarded when we define the outcome of the experiment to be merely whether there was a particle in each of the patches $P_\ell$; but the detectors actually obtain more precise information and correspondingly collapse the wave function more narrowly than indicated in the curved collapse rule just formulated. Put differently, the detection processes we are considering are \emph{not} processes that ``test only whether the configuration lies in $S$.''
	\item 	A prior result related to our theorem was obtained by Bloch \cite{bloch:1934} in the context of multi-time wave functions. Bloch considered a multi-time wave function $\phi(x_1,\ldots,x_N)$ of $N$ non-interacting particles; the position of particle $i$ gets measured at time $t_i$ (relative to some given Lorentz frame) with outcome $\vX_i\in\RRR^3$. Then, Bloch showed using the Born rule and projection postulate at equal times, the joint distribution of the outcomes has density $|\phi(X_1,\ldots,X_N)|^2$ with $X_i=(t_i,\vX_i)$. As a corollary, suppose the $N$ particles do interact in principle but particle $i$ is confined, during a certain time interval, to a region $V_i \subset \M$ such that any two regions are spacelike separated, and the detectors are set up along horizontal surfaces $R_i \subset V_i,~i=1,\ldots,N$ (see Figure~\ref{fig:bloch}). Then detection on any $\Sigma$ containing all $R_i$ yields a configuration $(X_1,\ldots,X_N)$ with $X_i\in R_i$ and distribution density $|\phi(X_1,\ldots,X_N)|^2=|\psi_\Sigma(X_1,\ldots,X_N)|^2$.
	
\begin{figure}[ht]
\centering
 \includegraphics[width=0.65\textwidth]{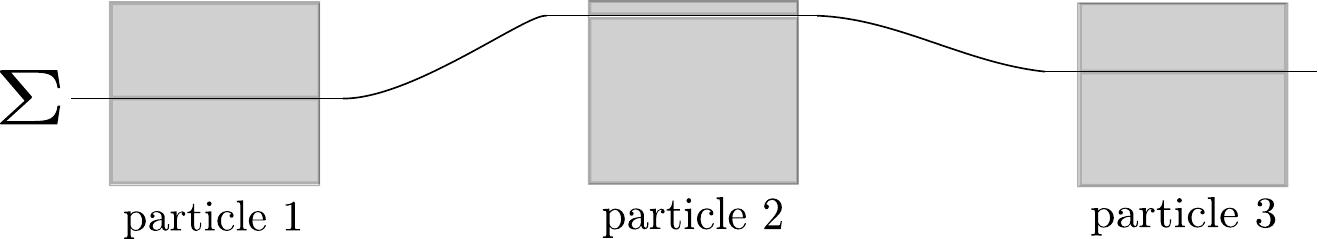}
 \caption{Space-time diagram of the situation to which Bloch's analysis in \cite{bloch:1934} applies: $N$ particles ($N=3$ shown) are confined to spacelike separated space-time regions (shaded); detectors are placed along a Cauchy surface $\Sigma$ that is horizontal in each region. Adapted from \cite{LPT:2017}.}
 \label{fig:bloch}
\end{figure}

	Bloch's result in the latter form is a special case of the curved Born rule and Theorem~\ref{thm:prob}; in the former form it is a special case of Remark~\ref{rem:non-interacting} above. Our result goes further in that (i)~we do not require that $\Sigma$ be horizontal in the regions considered; (ii)~we do not require the absence of interaction in the regions considered; and (iii)~we admit a variable number of particles.
	\item Some natural questions about multi-time wave functions $\phi$ are, how do we know what the ``right'' definition of, or the ``right'' equations for, $\phi$ are? After all, there are many functions on $\sS$ as in \eqref{sSdef} that agree with $\psi_t$ on all simultaneous configurations (i.e., on $\Gamma(\Sigma_t)$), and there may be several simple criteria selecting particular such functions. And, how do we know that $\sS$ is the ``right'' domain of definition for $\phi$? Our results provide some answers: Since $\psi_\Sigma$ is the restriction of $\phi$ to $\Gamma(\Sigma)$, the curved Born rule provides a direct connection between $\phi$ and detection probabilities---quantities that are measurable, at least in principle.  Even if many different functions could be called wave functions, the one that we know is directly connected to detection probabilities is the one whose evolution, when expressed as a hypersurface evolution, satisfies (IL) and (PL). Moreover, this connection holds only for spacelike configurations (except in the absence of interaction), so that $\sS$ is naturally selected as a domain. (Note that in most models it is selected as well by the condition that the multi-time evolution equations be consistent \cite{pt:2013a,pt:2013c}.) 
\item We now describe the implications of our result for Bohmian mechanics \cite{DT09}. But first we recall that, in contrast to other versions of quantum mechanics, the definition of Bohmian mechanics involves no postulates about detection results, and already the horizontal Born rule about detection results and the associated collapse rule are theorems derived from the equation of motion applied to all particles (including those forming the detector) and the quantum equilibrium distribution of their configuration \cite[Sec.~9.1--2]{DT09}.

Let us turn to a single Dirac particle: the Bohmian world line is an integral curve of the current vector field $j^\mu=\overline{\psi} \gamma^\mu \psi$ (which is everywhere timelike or lightlike) and is random with $|\psi|^2$ distributed initial condition. It is known that the random point where the trajectory intersects a given Cauchy surface $\Sigma$ in the absence of detectors has the curved Born distribution $\overline{\psi}(x)\, \gamma^\mu n_\mu(x)\, \psi(x)\, d^3x$. Our result shows that also in the presence of detectors along $\Sigma$, the intersection point, which coincides with the detection point, has the curved Born distribution on $\Sigma$ (whereas the intersection point with another surface $\Sigma'$ does not). 

For several particles (with or without interaction), Bohmian mechanics postulates that one foliation $\foliation$ of $\M$ into Cauchy surfaces (not necessarily horizontal) is singled out in nature, and the law of motion \cite{HBD} depends on it. It is known that, in the absence of detectors, the intersection points $X_1,\ldots,X_N$ of the $N$ world lines with a Cauchy surface $\Sigma$ jointly have curved Born distribution as in \eqref{rhoSigmaN} when $\Sigma$ belongs to $\foliation$ but generally not otherwise. Our result shows that 
\be\label{invisible}
\begin{minipage}{100mm}
in the presence of detectors on $\Sigma$, the joint distribution of the $N$ intersection points is a curved Born distribution regardless of whether or not $\Sigma$ belongs to $\foliation$.
\end{minipage}
\ee
This statement was made already in \cite[sec.~III.B]{HBD}. 
As a corollary, the empirical findings of the inhabitants of a Bohmian universe (governed by the law of motion of \cite{HBD}) contain no signature of $\foliation$, so that the inhabitants cannot determine $\foliation$ empirically. 

In \cite{HBD}, the following qualitative argument was given for \eqref{invisible}: The detection outcomes will agree with where the $N$ particles actually arrived on $\Sigma$, so it suffices to show that the detection outcomes have curved Born distribution. Treating the detectors as another quantum system, we may consider the wave function $\Psi$ of the object (the $N$ particles) and the detectors together.  Suppose $\Sigma'\in\foliation$ is so late that $\Sigma' \subseteq \future(\Sigma)$; then $\Psi_{\Sigma'}$ is a superposition of different measurement outcomes, and the Bohmian configuration on $\Sigma'$ has the corresponding $|\Psi|^2$ probability to display a particular outcome. Since recorded outcomes should be stable, the $|\Psi|^2$ weight of each outcome should be the same on all Cauchy surfaces $\Sigma''\subseteq \future(\Sigma)$. (In particular, since the evolution of wave functions does not depend on $\foliation$, the $|\Psi|^2$ weights will not depend on $\foliation$.) Thus, the $|\Psi|^2$ distribution of the outcomes on $\Sigma'$ should agree with that on $\Sigma$, which should be the curved Born distribution. Thus, the probability of the Bohmian configuration displaying a particular outcome should agree with the curved Born distribution, quod erat demonstrandum. This argument, while merely heuristic, seems very general and convincing; this argument and our result support and complement each other.
	\item\label{rem:oplus} We conjecture that axiom (iii), factorization of the PVM, is equivalent to the existence of a measurable field of Hilbert spaces $\Hilbert_x$, $x\in \M$, such that\footnote{Recall that we regard configurations $q$ on $\Sigma$ as unordered, i.e., as subsets of $\Sigma$.}
\be
\Hilbert_A=\int^\oplus_{\Gamma(A)} dq\, \bigotimes_{x\in q} \Hilbert_x\,,
\ee
and $P_A(S)$ the projection to
\be
\int^\oplus_S dq\, \bigotimes_{x\in q} \Hilbert_x\,,
\ee
where $dq$ is the natural measure on $\Gamma(A)$ (see Eq.~\eqref{dqdef} in Section~\ref{sec:config}).
	\item {\it Detection along a timelike hypersurface.} A question related to the one discussed here is to determine the probability distribution of detection events on a hypersurface $\Sigma$ that is not spacelike or Cauchy but \emph{timelike}. For example, suppose we prepare $N$ particles with initial wave function $\psi_0(\vx_1,\ldots,\vx_N)$ inside a region $\Omega \subset \RRR^3$, set up detectors on the 2-surface $\partial \Omega$, and wait for the detectors to register particles upon arrival at $\partial \Omega$, that is, on the timelike hypersurface $\Sigma = [0,\infty) \times \partial \Omega$. If, say, particle $k$ crosses $\Sigma$ at $x_k$, then we can ask for the joint distribution $\rho$ of $(x_1,\ldots,x_N)$. For some versions of this question, specific equations have been proposed and argued for as answers: In \cite{DP03,DT04} for the situation in scattering theory, where $\Omega$ is taken to be bounded but very large (say, a ball of radius $R\to\infty$ around the origin), and the particles do not interact (any more for $t\geq 0$). And in \cite{detect-rule,detect-dirac} for ideal detectors on a timelike $\Sigma$ at a finite distance, so that the presence of the detectors will have a back effect on the time evolution of the wave function. In all of these cases, the specific equation for the probability distribution is of the basic type $\rho=|\psi|^2$ as for the Born rule (more precisely, is given by the normal component $j^\mu \, n_\mu$ of the current density $j^\mu$ across the detector surface), but more needs to be said about how to obtain the appropriate $\psi$.
	
\end{enumerate}
\setcounter{remarks}{\theenumi}

\section{Definitions Used for Hypersurface Evolution}
\label{sec:dynamics}

Relative to a given Lorentz frame, let $\pi:\RRR^4\to \RRR^3$ be the projection
\be\label{pidef}
\pi(x^0,x^1,x^2,x^3) = (x^1,x^2,x^3)\,.
\ee
All $\sigma$-algebras we use are the appropriate Borel $\sigma$-algebras, denoted by $\sB(X)$ for the topological space $X$. On $\Sigma$ we consider the topology induced from that of $\mathbb{M}$. It is known \cite[p.~417]{ON:1983} that the restriction $\pi_\Sigma$ of the projection $\pi$ to $\Sigma$ is a homeomorphism $\Sigma\to\RRR^3$. As a consequence, $R\subseteq \Sigma$ lies in $\sB(\Sigma)$ iff $\pi(R)\in \sB(\RRR^3)$. Now Rademacher's theorem \cite{Rade} asserts that every Lipschitz function from $\RRR^n$ to $\RRR^m$ is differentiable almost everywhere; since $x^0 \circ\pi^{-1}_\Sigma$ is Lipschitz (with Lipschitz constant 1), $\Sigma$ possesses a tangent plane at almost every point (i.e., the exceptions project to a set in $\RRR^3$ of Lebesgue measure zero). At points with a tangent plane, $\Sigma$ possesses a Riemannian 3-metric. The 3-metric defines a volume measure $\mu_\Sigma$ on $(\Sigma,\sB(\Sigma))$. Note that the projection $\pi_* \mu_\Sigma = \mu_\Sigma \circ \pi^{-1}$ to $\RRR^3$ has the same sets of measure zero as the Lebesgue measure.

\subsection{Configuration Space}
\label{sec:config}

The set $\Gamma(R)$ inherits a topology and a measure from $R$ as follows. 
Consider the space of \textit{ordered configurations} on $R$, i.e.,
\be
\Gamma_\mathrm{o}(R) = \bigcup_{n=0}^\infty R^n_{\neq}
\ee
with $R^n_{\neq} = \bigl\{ (q_1,\ldots,q_n)\in R^n: q_i\neq q_j \: \forall i\neq j \bigr\}$ the ordered configurations of $n$ particles without collisions. We naturally have the product topology and product measure on $R^n$, their restrictions on $R^n_{\neq}$, and the appropriate combination on the union, i.e., on $\Gamma_{\mathrm{o}}(R)$.

Let $\tau : \Gamma_\mathrm{o}(R) \rightarrow \Gamma(R)$ be the ``projection'' to $\Gamma(R)$ that forgets the ordering, i.e., for any $q \in \Gamma_\mathrm{o}(R)$, 
\be\label{taudef}
q = (x_1,\ldots,x_n)\,, \quad 
\tau(q) = \{ x_1,\ldots,x_n\}\,.
\ee
The topology we consider for $\Gamma(R)$ then is the weakest topology which makes $\tau$ continuous, i.e., the open subsets of $\Gamma(\Sigma)$ are those that have an open pre-image under $\tau$. For the Borel $\sigma$-algebras, it follows that $S\in \sB(\Gamma(R)) \Leftrightarrow \tau^{-1}(S) \in \sB(\Gamma_{\rm o}(R))$.

The measure $\mu_{\Gamma(R)}$ on $\sB(\Gamma(R))$ can be obtained as follows. For $S \in \sB(\Gamma(R))$, let $S_\mathrm{o} := \tau^{-1}(S) \in \sB(\Gamma_\mathrm{o}(R))$ and $S_\mathrm{o}^{(n)} = S_\mathrm{o} \cap R^n_{\neq}$. Then $S_{\rm o} = \bigcup_{n=0}^\infty S_{\rm o}^{(n)}$. Let furthermore $\mu_{R^n}$ denote the product measure on $R^n$. We define:
\be\label{dqdef}
	\mu_{\Gamma(R)}(S) = \sum_{n=0}^\infty \frac{1}{n!} \, \mu_{R^n}(S_{\rm o}^{(n)}).
\ee
The factor $\frac{1}{n!}$ compensates for the fact that $\tau^{-1}(q)$ contains $n!$ elements if $\# q = n$. When we write $\int dq\, f(q)$ for an integral over configuration space $\Gamma(\Sigma)$ or some subset of it, as in \eqref{eq:pvmdensity} or \eqref{eq:generalizedbornrule}, we actually mean integration relative to the measure $\mu_{\Gamma(\Sigma)}$, i.e., $dq$ is short for $\mu_{\Gamma(\Sigma)}(dq)$.

The following subsets of $\Gamma(\Sigma)$ will be of special interest. For $R \subseteq \Sigma$, we define: 
\begin{align}
	\emptyset(R) &= \{ q \in \Gamma(\Sigma) : q \cap R = \emptyset \},\nonumber\\
	\exists(R) &= \{ q \in \Gamma(\Sigma) : q \cap R \neq \emptyset \}, \nonumber\\
	\forall(R) &= \{ q \in \Gamma(\Sigma) : q \cap R = q \}.
\end{align}
$\emptyset(R)$ is the set of configurations on $\Sigma$ with no point in $R$, $\exists(R)$ the one with at least one point in $R$ and $\forall(R)$ with points exclusively in $R$.
Denoting complements by a superscript $c$ ($R^c=\Sigma \backslash R$ and $S^c=\Gamma(\Sigma)\setminus S$), we have that
\be
	\exists(R) = \emptyset(R)^c,~~~\forall(R) = \emptyset(R^c).
\ee
 Note that the notation $\exists(R)$, $\emptyset(R)$ and $\forall(R)$ does not make explicit to which $\Sigma$ these sets refer. If this is not clear from the context, we will indicate this with a subscript. Furthermore, for subsets $R\subseteq A \subseteq \Sigma$ of a Cauchy surface we introduce
\be
	\emptyset_A(R) = \{ q \in \Gamma(A) : q \cap R = \emptyset \}
\ee
and similarly for $\exists_A(R), \forall_A(R)$.

\subsection{Hilbert Spaces and Hypersurface Evolution}
\label{sec:hypersurfaceevolution}

Let us summarize the definition from the introduction:
\begin{defn}
A \textit{hypersurface evolution} $\mathscr{E} = ( \Hilbert_\circ, P_\circ, U_{\circ}^{\circ})$  is a collection of
\begin{enumerate}
	\item Hilbert spaces $\Hilbert_A$ for every subset $A\subseteq \Sigma$ of every Cauchy surface $\Sigma$, equipped with
	\item a PVM $P_A : \sigma(A) \rightarrow {\rm Proj}(\Hilbert_A)$, where $\sigma(A)$, the $\sigma$-algebra associated with $A$, is the restriction of $\mathscr{B}(\Gamma(\Sigma))$ to $\forall(A)$, $\sigma(A)=\{S\cap \forall(A):S\in \sB(\Gamma(\Sigma))\}$, and ${\rm Proj}(\Hilbert)$ denotes the set of projections on $\Hilbert$, and
	\item unitary isomorphisms $U_{\Sigma}^{\Sigma'}:\Hilbert_\Sigma \to \Hilbert_{\Sigma'}$ for every pair of Cauchy surfaces $\Sigma, \Sigma'$
\end{enumerate}
with the following properties:
\begin{enumerate}
	\item[(0)] $U_\Sigma^\Sigma = I_\Sigma$ and $U_{\Sigma'}^{\Sigma''} U_{\Sigma}^{\Sigma'} = U_{\Sigma}^{\Sigma''}$ for all Cauchy surfaces $\Sigma, \Sigma', \Sigma''$.
	\item[(i)] For every $S\subset \Gamma(\Sigma)$ with $\mu_{\Gamma(\Sigma)}(S) = 0$, also $P_\Sigma(S) = 0$.
	\item[(ii)] For every $\Sigma$, the range of $P_\Sigma(\emptyset(\Sigma)) \Hilbert_\Sigma$ is 1-dimensional. That is, up to a phase, there is a unique \textit{vacuum state} $| \emptyset(\Sigma) \rangle \in \range P_\Sigma(\emptyset(\Sigma))$ with $\bigl\|  |\emptyset(\Sigma) \rangle \bigr\| = 1$.
	\item[(iii)] Factorization of the PVM: $\Hilbert_{A\cup B}=\Hilbert_A \otimes \Hilbert_B$ and $P_{A\cup B} = P_A \otimes P_B$ [where the tensor product of two PVMs is the PVM characterized by \eqref{PAB}] for any mutually disjoint $A,B\subseteq \Sigma$.
\end{enumerate}
\end{defn}

\begin{defn}
	Given a hypersurface evolution $\mathscr{E}$, we call the elements of a collection $\Sigma \mapsto \psi_\Sigma$ with $\psi_\Sigma \in \Hilbert_\Sigma$ for every $\Sigma$ such that $\psi_{\Sigma'} = U_{\Sigma}^{\Sigma'} \psi_\Sigma$ for every $\Sigma, \Sigma'$ \textit{hypersurface wave functions}.
	Similarly, the elements of a collection $\Sigma \mapsto \rho_\Sigma$ with $\rho_\Sigma$ a density matrix on $\Hilbert_\Sigma$ such that $\rho_{\Sigma'} = U_{\Sigma}^{\Sigma'} \rho_\Sigma \, U_{\Sigma'}^\Sigma$ are called \textit{hypersurface density matrices}.
\end{defn}

\paragraph{Remarks.} 
\begin{enumerate}
\setcounter{enumi}{\theremarks}
	\item\label{rem:species} The configurations in $\Gamma(\Sigma)$ do not distinguish between several species of particles. For $m$ species, one may consider replacing it by $\Gamma(\Sigma)^m$. However, here we shall only consider detectors which do not distiguish between species, and then it suffices to use the PVM $\tilde P_\Sigma$ obtained from the product PVM $P_\Sigma\otimes \cdots \otimes P_\Sigma$ via the mapping $\tau_m:\Gamma(\Sigma)^m\to \Gamma(\Sigma)$ that discards the information about the particle species,
	\be
	\tau_m(q_1,\ldots,q_m) := q_1 \cup \ldots \cup q_m\,,
	\ee 
	$\tilde P_{\Sigma}(S) = [P_\Sigma \otimes \cdots \otimes P_{\Sigma}](\tau^{-1}_m(S))$. Accordingly, we consider only $\Gamma(\Sigma)$ as the configuration space, and not $\Gamma(\Sigma)^m$.
	\item Property (i) implies that the probability measure $\Prob^{\psi_\Sigma}(\cdot) = \| P_\Sigma(\cdot) \psi_\Sigma \|^2$ is absolutely continuous with respect to the measure $\mu_{\Gamma(\Sigma)}(\cdot)$ and hence possesses a density, $\rho^{\psi_\Sigma}$. That is, for any $S \subseteq \Gamma(\Sigma)$, $\Prob^{\psi_\Sigma}(S) = \int_S dq\, \rho^{\psi_\Sigma}(q)$. In many cases, $\rho^{\psi_\Sigma}$ can simply be represented as $|\psi_\Sigma|^2$, as in \eqref{eq:pvmdensity}.
	\item If $A,B\subseteq \Sigma$ differ only by a set of measure 0, $\mu_\Sigma(A\setminus B)=0=\mu_\Sigma(B\setminus A)$, then for $\psi\in\Hilbert_\Sigma$ being concentrated in $A$ is equivalent to being concentrated in $B$. That is because then also $\forall(A)$ and $\forall(B)$ differ only by a set of measure 0, $\mu_{\Gamma(\Sigma)}(\forall(A)\setminus \forall(B))=0=\mu_{\Gamma(\Sigma)}(\forall(B)\setminus \forall(A))$, and because by Property (i) every set of $\mu_{\Gamma(\Sigma)}$-measure 0 also has $P_\Sigma$-measure zero.
	\item By iterated application of axiom (iii), factorization of the PVM, one obtains the corresponding statement for more than two sets: Given a partition of $R\subseteq\Sigma$ into $n$ sets $R_1,\ldots,R_n$, we have that $\Hilbert_R=\Hilbert_{R_1}\otimes \cdots \otimes \Hilbert_{R_n}$ and $P_R= P_{R_1}\otimes \cdots\otimes P_{R_n}$ [i.e., $P_R(S_1\times \cdots \times S_n) = P_{R_1}(S_1) \otimes \cdots \otimes P_{R_n}(S_n)$ for all $S_i\subseteq \Gamma(R_i)$].
	\item\label{rem:vacuum} Note that the existence of the vacuum state $| \emptyset(\Sigma) \rangle$ together with (iii) implies the existence of vacuum states $| \emptyset(R_i) \rangle \in \Hilbert_{R_i}$ with $\| | \emptyset(R_i) \rangle \| = 1$ which are unique up to a phase, so $| \emptyset(\Sigma) \rangle = e^{i \theta} \, | \emptyset(R_1) \rangle \otimes \cdots \otimes | \emptyset(R_n) \rangle$ for some $\theta \in (-\pi, \pi ]$.
	\item\label{rem:unitaryequivalence} We regard two hypersurface evolutions, $(\Hilbert_\circ, P_\circ, U_\circ^\circ)$ and $(\tilde\Hilbert_\circ,\tilde{P}_\circ, \tilde{U}_\circ^\circ)$, as equivalent (i.e., as representing the same physical evolution) if there are unitary isomorphisms $V_\Sigma: \Hilbert_\Sigma \to \tilde\Hilbert_{\Sigma}$ such that $\tilde{P}_\Sigma(\cdot) = V_\Sigma\, P_\Sigma(\cdot) \, V_\Sigma^{-1}$ and $\tilde{U}_{\Sigma}^{\Sigma'} = V_{\Sigma'} \, U_\Sigma^{\Sigma'}\, V_\Sigma^{-1}$. Note that equivalent hypersurface evolutions lead to equal curved Born distributions for all $\Sigma$ and $\psi_0$. Correspondingly, also the $\Hilbert_A$ and $P_A$ for $A\subset \Sigma$ are only defined up to unitary isomorphism. In fact, if we want to formulate axiom (iii), the factorization of the PVM, very carefully, we should say that for every $A,B\subseteq \Sigma$ with $A \cap B = \emptyset$, there is a unitary isomorphism $U_{A,B}:\Hilbert_{A\cup B} \to \Hilbert_A \otimes \Hilbert_B$ such that $P_{A\cup B}(S_A\times S_B) = U_{A,B}\,[P_A(S_A) \otimes P_B(S_B)]\,U_{A,B}^{-1}$. (We hinted at this formulation already between \eqref{HilbertAB} and \eqref{PAB}.)
	
	This unitary freedom implies, of course, that there is no fact about which Hilbert space is ``really the right'' Hilbert space for $\Sigma$. So, while the first Hilbert space we introduced in Section~\ref{sec:rule1} was an $L^2$ space of functions $\Sigma\to\CCC^4$ (along with its Fock space), we are also allowed to take always the same Hilbert space, say $\tilde\Hilbert_{\Sigma}=\Hilbert_{\Sigma_0}$; we can then even take all $\tilde{U}_{\Sigma}^{\Sigma'}$ to be the identity, as would correspond to $V_{\Sigma} = U_{\Sigma}^{\Sigma_0}$. Then, of course, $\tilde{P}_{\Sigma}$ will be very different PVMs for different $\Sigma$, in fact the family $\tilde{P}_\circ$ then encodes the whole time evolution, so this choice of $(\tilde\Hilbert_\circ, \tilde{P}_\circ, \tilde{U}_\circ^\circ)$ is analogous to the Heisenberg picture in non-relativistic quantum mechanics.
	\item\label{rem:Fock} \textit{Relation to Fock space construction.} The familiar Fock space construction is closely related to the framework of hypersurface evolutions. Let us describe this in some detail. Suppose that the 1-particle Hilbert space is $\Hilbert_\Sigma^{(1)}=L^2(\Sigma, \CCC^k, \nu)$, which means the space of square-integrable functions relative to the inner product
	\be
	\langle \psi | \chi \rangle 
	= \int_\Sigma d^3x \: \psi(x)^\dagger \, \nu(x)\, \chi(x),
	\ee
	for some function $\nu$ from $\Sigma$ to the positive definite $k\times k$ matrices. For the Dirac equation, as described in Section~\ref{sec:rule1}, $k=4$ and $\nu(x) = \gamma^0\, \gamma^\mu\, n_\mu(x)$. Being a space of functions on $\Sigma$, $\Hilbert_{\Sigma}^{(1)}$ is automatically equipped with a PVM $P_\Sigma^{(1)}$ on $\Sigma$ acting on $\Hilbert_\Sigma^{(1)}$, the ``natural PVM,'' viz., $P_\Sigma^{(1)}(A)$ is the multiplication operator by the characteristic function of $A$. Now the bosonic (or fermionic) Fock space is
\be
	\sF_{\Sigma,\pm} = \Gamma_{\pm} \bigl( \Hilbert_\Sigma^{(1)} \bigr) 
	= \bigoplus_{n=0}^\infty S_\pm (\Hilbert_\Sigma^{(1)})^{\otimes n} \,,
\ee
where $\Gamma_{\pm}$ denotes the ``second quantization functor'' and $S_\pm$ the (anti-)symmetrization operator. It inherits a PVM $P_\Sigma$ from $\Hilbert_\Sigma^{(1)}$,
\begin{align}
\bigl( P_\Sigma(S) \, \psi \bigr)^{(n)} 
&= P_\Sigma^{(n)} \bigl( S\cap \Gamma_n(\Sigma) \bigr) \, \psi^{(n)} \\
P_\Sigma^{(n)}(S) 
&= P_{\Sigma,\mathrm{o}}^{(n)} \bigl( \tau^{-1}(S) \bigr) \\
P_{\Sigma,\mathrm{o}}^{(n)}(A_1\times \cdots \times A_n) 
&= P_\Sigma^{(1)}(A_1) \otimes \cdots \otimes P_\Sigma^{(1)}(A_n)\,,
\end{align}
where $P_{\Sigma,\mathrm{o}}^{(n)}$ is the corresponding PVM on \emph{ordered} $n$-particle configurations (i.e., on $\Sigma^n$) acting on $(\Hilbert_\Sigma^{(1)})^{\otimes n}$, and $\tau$ is the ``unordering map'' as in \eqref{taudef}. Equivalently, $P_\Sigma(S)$ is multiplication by the characteristic function of $\tau^{-1}(S)\subseteq \Gamma_\mathrm{o}(\Sigma)$ if we regard $\sF_{\Sigma,\pm}$ as the space of square-integrable functions on $\Gamma_\mathrm{o}(\Sigma)$ whose $n$-particle sector takes values in $(\CCC^k)^{\otimes n}$ and is (anti-)symmetric against permutation of particles (along with their spin indices), with inner product
\begin{multline}
	\langle \psi | \chi \rangle 
	= \sum_{n=0}^\infty \frac{1}{n!} \int_{\Sigma^n} d^3x_1\cdots d^3x_n \, \psi^{(n)}(x_1,\ldots, x_n)^\dagger \: \times\\
	\times \: \bigl[\nu(x_1) \otimes \cdots \otimes \nu(x_n)\bigr]\, \chi^{(n)}(x_1, \ldots, x_n)\,. 
	\label{scpFock}
\end{multline}

Let us verify the axioms (i), (ii), (iii): (i)~If $S\subseteq \Gamma(\Sigma)$ is a set of measure 0, then so is $\tau^{-1}(S)$, so $P_\Sigma(S)$ is multiplication by a function that is 0 almost everywhere, so $P_\Sigma(S)$ is the 0 operator. (ii)~$P_\Sigma(\{\emptyset\})$ is the projection onto the $n=0$ sector of Fock space, which is 1-dimensional. (iii)~Suppose $\Sigma=A \cup B$ with $A\cap B=\emptyset$. To see how $\sF_{\Sigma,\pm}$ can be identified with $\sF_{A,\pm} \otimes \sF_{B,\pm}$, pick an element $\psi$, regard it as a function on $\Gamma_\mathrm{o}(\Sigma)$, permute the particles (and their spin indices) so that all locations in $A$ are listed before any locations in $B$. That provides the mapping $U_{A,B}$;  since $\Gamma(A\cup B) = \Gamma(A) \times \Gamma(B)$, and since $P_\Sigma$ is multiplication by characteristic functions, $P_\Sigma$ also factorizes accordingly.

We turn to the time evolution. The free Dirac evolution, or the Dirac evolution in an external electromagnetic field, defines unitary 1-particle hypersurface mappings $U=U_\Sigma^{(1)\Sigma'}: \Hilbert_\Sigma^{(1)}\to \Hilbert_{\Sigma'}^{(1)}$, so $U^{\otimes n}$ maps the $n$-th tensor powers to each other, and $(U_{\Sigma}^{\Sigma'}\psi)^{(n)} = U^{\otimes n} \psi^{(n)}$ is the non-interacting evolution in Fock space; in other words, $U_{\Sigma}^{\Sigma'}$ is obtained from $U_{\Sigma}^{(1)\Sigma'}$ by applying the ``second quantization functor'' $\Gamma_{\pm}$. (The functor can be applied either to Hilbert spaces or to their unitary isomorphisms.) The relations $U_{\Sigma'}^{\Sigma''} U_{\Sigma}^{\Sigma'} = U_{\Sigma}^{\Sigma''}$ and $U_\Sigma^\Sigma=I_\Sigma$ are inherited from the 1-particle evolution.
	\item\label{rem:unordered} \textit{Direct construction of wave functions on unordered configurations.} Instead of using (anti-)symmetric wave functions on \emph{ordered} configurations, which seems unphysical, one can also directly construct Hilbert spaces of wave functions on \emph{unordered} configurations, so that $\psi_\Sigma$ is a function on $\Gamma(\Sigma)$. In the case of fermions, such a direct construction requires the use of a Hermitian vector bundle called the fermionic line bundle \cite{LM:1977,GTTZ:2014}.
	\item\label{rem:naturalization} {\it Regarding elements of $\Hilbert_\Sigma$ as functions.} In Remarks~\ref{rem:Fock} and \ref{rem:unordered}, we have pointed to examples in which the elements of $\Hilbert_{\Sigma}$ are actually \emph{functions} on $\Gamma_\mathrm{o}(\Sigma)$ or $\Gamma(\Sigma)$. It turns out that this situation is more than just an example, and one can always regard elements of $\Hilbert_\Sigma$ as functions on $\Gamma(\Sigma)$, as every Hilbert space with a PVM $P$ on a set $\Omega$ is unitarily equivalent to an $L^2$ space over $\Omega$ with the natural PVM; this can be called the ``naturalization'' of $P$. Here is the relevant statement \cite{Dixmier,DGTZ:2005}, which is closely related to the Hahn--Hellinger theorem \cite{Partha}: If $P$ is a PVM on the standard Borel space\footnote{A \emph{standard Borel space} is a
   measurable space isomorphic to a complete separable metric space
   with its Borel $\sigma$-algebra.} $\Omega$ acting on the Hilbert space $\Hilbert$, then there is a measurable field of Hilbert spaces $\Hilbert_q$ over $\Omega$, a $\sigma$-finite measure $\mu$ on $\Omega$, and a unitary isomorphism $U : \Hilbert \rightarrow \int^\oplus \mu(dq)\, \Hilbert_q$ to the direct integral of $\Hilbert_q$ that carries $P$ to the natural PVM on $\Omega$ acting on $\int^\oplus \mu(dq)\,\Hilbert_q$. 
	The naturalization is unique in the sense that if $\{ \Hilbert_q' \}, \mu', U'$ is another such triple, then there is a measurable function $f : \Omega \rightarrow (0,\infty)$ such that $\mu'(dq) = f(q) \, \mu(dq)$ and a measurable field of unitary isomorphisms $U_q : \Hilbert_q \rightarrow \Hilbert_q'$ such that $U' \psi(q) = f(q)^{-1/2} U_q U \psi(q)$.
	
	In our case, $\Omega=\Gamma(\Sigma)$ is a standard Borel space, so we can identify $\Hilbert_\Sigma$ with $\int^\oplus \mu(dq)\, \Hilbert_q$. The elements of the latter are functions on $\Gamma(\Sigma)$, viz., cross-sections of the bundle $\Hilbert_q$. Thus, if $\psi\in\Hilbert_\Sigma$ then
	\be
	\scp{\psi}{P_\Sigma(S)|\psi} = \int_{S} \mu(dq)\, |U\psi(q)|^2\,,
	\ee
	where $U\psi(q)\in\Hilbert_q$, and $|\cdot|$ is the norm of $\Hilbert_q$. Since $P_\Sigma$ is absolutely continuous relative to $\mu_{\Gamma(\Sigma)}$, so is $\mu$; thus (allowing $\Hilbert_q=\{0\}$ for some $q$ if necessary), by choosing $f=d\mu/d\mu_{\Gamma(\Sigma)}$, we can replace $\mu$ by $\mu_{\Gamma(\Sigma)}$, so that, finally,
	\be
	\scp{\psi}{P_\Sigma(S)|\psi} = \int_{S} dq\, |U\psi(q)|^2\,.
	\ee
	That is, the distribution $\scp{\psi}{P_\Sigma(\cdot)|\psi}$ can always be regarded as ``the $|\psi|^2$ distribution.''
\end{enumerate}
\setcounter{remarks}{\theenumi}

\subsection{Consequences of the Properties (IL) and (PL)}
\label{sec:FS}

\subsubsection{Vacuum Stays Vacuum}

The following property of a hypersurface evolution is a trivial consequence of (PL):

\begin{enumerate}
	\item[(NCFV)] \textit{\underline{N}o particle \underline{c}reation \underline{f}rom the \underline{v}acuum.} 
	For any two Cauchy surfaces $\Sigma,\Sigma'$, the vacuum space evolves to the vacuum space,
	\be\label{eq:NCFV1}
	U^{\Sigma'}_\Sigma \: P_\Sigma(\{\emptyset\})\: U^\Sigma_{\Sigma'} = P_{\Sigma'}(\{\emptyset\})\,.
	\ee
\end{enumerate}

Indeed, technically speaking, the vacuum state $\psi=|\emptyset\rangle_\Sigma$ is concentrated in the empty set $A=\emptyset\subseteq \Sigma$, and the grown set of the empty set is again empty, $\Gr(\emptyset,\Sigma')=\emptyset$. But the only state $\psi'\in \Hilbert_{\Sigma'}$ concentrated in the empty set is the vacuum, $\psi' \in \CCC |\emptyset\rangle_{\Sigma'}$, so $U^{\Sigma'}_\Sigma$ must map $|\emptyset\rangle_\Sigma$ to $|\emptyset\rangle_\Sigma$ up to a phase, quod erat demonstrandum.

We conjecture that, conversely, (IL) and (NCFV) together imply (PL). We now deduce from (IL) and (NCFV) that (NCFV) holds also locally.

\begin{prop}\label{prop:locNCFV}
	Suppose that a hypersurface evolution satisfies (IL) and (NCFV). Then it also satisfies (NCFV) locally; that is, for $\Sigma\cap \Sigma'=A$ (see Figure~\ref{fig:interaction_locality}),
	\be
		U_{\Sigma\setminus A}^{\Sigma'\setminus A} \: P_{\Sigma\setminus A}(\{\emptyset\}) \: U^{\Sigma\setminus A}_{\Sigma' \setminus A} = P_{\Sigma'\setminus A}(\{\emptyset\})\,.
		\label{eq:NCFV2}
	\ee
\end{prop}

\begin{proof}
We write $\emptyset_R$ for the 0-particle configuration in $\Gamma(R)$, $R=\Sigma,\Sigma',A$ etc. It lies in the nature of $\Gamma(\Sigma)$ that $\{\emptyset_\Sigma\} \cong \{\emptyset_A\} \times \{\emptyset_{\Sigma\setminus A}\}$ in the sense of \eqref{GammaAB}, and thus, by factorization of the PVM \eqref{PAB}, 
\be\label{PemptysetAAc}
P_\Sigma(\{\emptyset_\Sigma\}) = P_A(\{\emptyset_A\})\otimes P_{\Sigma\setminus A}(\{\emptyset_{\Sigma\setminus A}\})\,.
\ee
By \eqref{eq:il} of (IL),
\begin{align}
U_{\Sigma}^{\Sigma'} \: P_\Sigma(\{\emptyset_\Sigma\}) \: U^{\Sigma}_{\Sigma'}
&= \Bigl[I_A \otimes U_{\Sigma\setminus A}^{\Sigma'\setminus A} \Bigr]\: P_\Sigma(\{\emptyset_\Sigma\}) \:  \Bigl[I_A \otimes U^{\Sigma\setminus A}_{\Sigma'\setminus A} \Bigr] \nonumber\\
&= \Bigl[I_A \otimes U_{\Sigma\setminus A}^{\Sigma'\setminus A} \Bigr] \: \Bigl[ P_A(\{\emptyset_A\})\otimes P_{\Sigma\setminus A}(\{\emptyset_{\Sigma\setminus A}\}) \Bigr] \: \Bigl[ I_A \otimes U^{\Sigma\setminus A}_{\Sigma'\setminus A} \Bigr] \nonumber\\
&= P_A(\{\emptyset_A\})\otimes \biggl( U_{\Sigma\setminus A}^{\Sigma'\setminus A} P_{\Sigma\setminus A}(\{\emptyset_{\Sigma\setminus A}\}) U^{\Sigma\setminus A}_{\Sigma'\setminus A} \biggr)\,.
\end{align}
By (NCFV), the left-hand side equals $P_{\Sigma'}(\{\emptyset_{\Sigma'}\})$, which, by the analog of \eqref{PemptysetAAc} for $\Sigma'$, equals $P_A(\{\emptyset_A\})\otimes P_{\Sigma'\setminus A}(\{\emptyset_{\Sigma'\setminus A}\})$. Now \eqref{eq:NCFV2} follows.
\end{proof}

\subsubsection{Reduced Time Evolution Operators} 

An important consequence of propagation locality is to allow for a definition of reduced time evolution operators which map a state on $\Hilbert_A$ to $\Hilbert_{\Gr(A,\Sigma')}$. These operators will be crucial for the proof of the main theorem.

\begin{prop}[and definition]\label{prop:reducedevol}
	Let $\sE$ be a hypersurface evolution with (IL) and (PL). Let $\Sigma,\Sigma'$ be Cauchy surfaces and $A\subseteq \Sigma$.
	Then there exists an isometry $W_{A}^{\Gr(A,\Sigma')} : \Hilbert_A \rightarrow \Hilbert_{\Gr(A,\Sigma')}$ such that for every $\psi_\Sigma$ concentrated in $A$ (so that $\psi_\Sigma = \psi_A \otimes | \emptyset(A^c) \rangle$ for some $\psi_A \in \Hilbert_A$),
	\be
		U_\Sigma^{\Sigma'} \psi_\Sigma = \left( W_A^{\Gr(A,\Sigma')} \psi_A \right) \otimes \bigl| \emptyset(\Sr(A^c,\Sigma')) \bigr\rangle.
		\label{eq:woperator}
	\ee
	We call $W_A^{\Gr(A,\Sigma')}$ the \textnormal{reduced time evolution operator} from $A$ to $\Gr(A,\Sigma')$.	
	Its left inverse, given by its adjoint operator,
	 is denoted as $W_{\Gr(A,\Sigma')}^A$.
\end{prop}

\begin{proof}
	We define $W_A^{\Gr(A,\Sigma')}$ by the partial scalar product
	\be
		W_A^{\Gr(A,\Sigma')} \psi_A = \langle \emptyset(\Sr(A^c,\Sigma')) | \, U_\Sigma^{\Sigma'} |\psi_A \otimes  \emptyset(A^c) \rangle.
		\label{eq:defw}
	\ee
	Let $\psi_\Sigma$ be concentrated in $A \subseteq \Sigma$. (PL) implies that there exists a $\psi_{\Gr(A,\Sigma')} \in \Hilbert_{\Gr(A,\Sigma')}$ such that $U_\Sigma^{\Sigma'} \psi_\Sigma = \psi_{\Gr(A,\Sigma')} \otimes | \emptyset(\Sr(A^c,\Sigma')) \rangle$. Then $W_A^{\Gr(A,\Sigma')} \psi_A = \psi_{\Gr(A,\Sigma')}$, and \eqref{eq:woperator} follows.
	The fact that $W_A^{\Gr(A,\Sigma')}$ is an isometry is implied by \eqref{eq:defw} and the unitarity of $U_\Sigma^{\Sigma'}$.	
\end{proof}

\paragraph{Remarks.}
\begin{enumerate}
\setcounter{enumi}{\theremarks}
\item Since $W$ depends on the choice of the vacuum vectors $|\emptyset(\cdot)\rangle$, a different choice of these vectors (different by a phase factor) will change $W$ by a phase factor. At the end of the day, this ambiguity in the definition of $W$ will not affect our results because the $W$'s will appear symmetrically with their adjoints in the expression for the detection probabilities.

\item Note that $W$ is generally not unitary because it is not surjective. That is because for an abritrary state $\chi\in\Hilbert_{\Gr(A,\Sigma')}$, $U_{\Sigma'}^\Sigma \, \bigl[ \chi\otimes |\emptyset(\Sr(A^c,\Sigma') \rangle \bigr]$ will generally not be concentrated in $A$. 
\item (PL) can equivalently be formulated in terms of density matrices. A density matrix $\rho_\Sigma$ is said to be concentrated in $A \subseteq \Sigma$ iff
\be
	P_{\Sigma}(\forall(A)) \rho_\Sigma P_{\Sigma}(\forall(A)) = \rho_\Sigma.
\ee
Then $\sE$ satisfies (PL) iff for every $\rho_\Sigma$ concentrated in $A$, $\rho_{\Sigma'}$ is concentrated in $\Gr(A,\Sigma')$.

Accordingly, we can use the reduced time evolution operator $W_A^{\Gr(A,\Sigma)}$ to describe the evolution of a density matrix concentrated in $A$. Let $\rho_\Sigma$ be concentrated in $A$. Then there is a density matrix $\rho_A$ on $\Hilbert_A$ such that $\rho_\Sigma = \rho_A \otimes P_{\Sigma}\bigl( \emptyset(A^c) \bigr)$. Hence,
\be
	U_{\Sigma}^{\Sigma'} \rho_\Sigma \, U_{\Sigma'}^\Sigma = \left( W_A^{\Gr(A,\Sigma)} \rho_A \, W_{\Gr(A,\Sigma)}^A \right) \otimes P_{\Sigma'}\bigl( \emptyset(\Sr(A^c,\Sigma')) \bigr) .
\ee
\end{enumerate}
\setcounter{remarks}{\theenumi}

\section{Examples}
\label{sec:examples}

In this section, we give an overview of some approaches to defining a hypersurface evolution with the properties (IL) and (PL).

\subsection{Free Dirac Evolution}
\label{sec:free}

In Remark~\ref{rem:Fock} in Section~\ref{sec:hypersurfaceevolution} above, we have given the definition of the hypersurface evolution for non-interacting Dirac particles, possibly in an external electromagnetic field, without any distinction between positive and negative energies. The hypersurface evolution with distinction between positive and negative energies has been discussed in \cite{DM:2016}. 

Let us return to the situation of Remark~\ref{rem:Fock} and verify (IL) and (PL). We first collect some observations about the 1-particle Dirac equation. It is known (e.g., \cite[Sec.~1.5]{Thaller} or \cite[Sec.~7.1]{pt:2013a}) that the wave function propagates no faster than light; thus, it is determined on $B'\subseteq \Sigma'$ by initial values in $\Gr(B',\Sigma)$. Now let $\Sigma\cap \Sigma' = A$, $B=\Sigma\setminus A$, $B'=\Sigma' \setminus A$. Then $\Hilbert_\Sigma^{(1)}=\Hilbert_A^{(1)} \oplus \Hilbert_B^{(1)}$ and $\Hilbert_{\Sigma'}^{(1)}=\Hilbert_A^{(1)} \oplus \Hilbert_{B'}^{(1)}$. Since $B=\Gr(B',\Sigma)$ and $B'=\Gr(B,\Sigma')$, and since $\psi|_A$ is the same in $\psi_\Sigma$ and $\psi_{\Sigma'}$, we have that
\be\label{IL(1)}
U_{\Sigma}^{(1)\Sigma'} = I^{(1)}_A \oplus U_B^{(1)B'}\,,
\ee
where $U_B^{(1)B'}:\Hilbert_B^{(1)} \to \Hilbert_{B'}^{(1)}$ is a unitary isomorphism.
 
Now we turn to the Fock spaces to confirm (IL). Since, as explained in Remark~\ref{rem:Fock}, $U_{\Sigma}^{\Sigma'} = \Gamma_{\pm}(U_\Sigma^{(1)\Sigma'})$ (the second quantization of unitary isomorphisms), and since the second quantization functor $\Gamma_{\pm}$ turns $\oplus$ into $\otimes$, it follows from \eqref{IL(1)} that $U_{\Sigma}^{\Sigma'} = I_A \otimes U_B^{B'}$ with $U_B^{B'} = \Gamma_{\pm}(U_B^{(1)B'})$, which is (IL).

We now verify (PL). The fact that the 1-particle Dirac wave function $\psi^{(1)}$ propagates no faster than light means that if $\psi^{(1)}_\Sigma$ is concentrated in $B\subseteq \Sigma$ (i.e., $\psi^{(1)}_\Sigma(x) =0$ for $x\in\Sigma\setminus B$) then $\psi^{(1)}_{\Sigma'}= U_\Sigma^{(1)\Sigma'}\psi^{(1)}_\Sigma$ is concentrated in $\Gr(B,\Sigma')$. As a consequence for $n$ non-interacting Dirac particles, if $\psi^{(n)}_\Sigma$ is concentrated in $B^n$, then $\psi^{(n)}_{\Sigma'}$ is concentrated in $\Gr(B,\Sigma')^n$. As a consequence in Fock space, if $\psi_\Sigma$ is concentrated in $\forall(B)$, then $\psi_{\Sigma'}$ is concentrated in $\forall(\Gr(B,\Sigma'))$; this is (PL).

\subsection{Tomonaga-Schwinger Picture}
\label{sec:TS}

The Tomonaga-Schwinger picture is closely related to the axiomatic framework of hypersurface evolutions. It associates a wave function $\widetilde{\psi}_\Sigma$ in a fixed Hilbert space $\widetilde{\Hilbert}$ with every Cauchy surface $\Sigma$.
The evolution of $\widetilde\psi_\Sigma$ is defined by the \textit{Tomonaga-Schwinger equation}, which relates $\widetilde{\psi}_\Sigma$ and $\widetilde{\psi}_{\Sigma'}$ for two infinitesimally neighboring Cauchy surfaces $\Sigma, \Sigma'$:
\be
	i \left( \widetilde{\psi}_{\Sigma'} - \widetilde{\psi}_{\Sigma} \right) = \left( \int_\Sigma^{\Sigma'} \!\!\! d^4 x \: \mathcal{H}_I(x) \right) \widetilde{\psi}_{\Sigma}\,.
	\label{eq:ts}
\ee
Here, $ \int_\Sigma^{\Sigma'} d^4 x$ denotes the integral over the $4$-volume enclosed by $\Sigma$ and $\Sigma'$, and $\mathcal{H}_I(x)$ denotes the interaction Hamiltonian density in the interaction picture, i.e., a function on $\M$ mapping space-time points $x$ to Hermitian operators on $\widetilde\Hilbert$.
The Tomonaga-Schwinger equation \eqref{eq:ts} is consistent iff the Hamiltonian density satisfies the condition
\be\label{TSconsistent}
	\bigl[ \mathcal{H}_I(x), \mathcal{H}_I(x') \bigr] = 0~~~{\rm if}~(x-x')^2 < 0.
\ee
In this way, the Tomonaga-Schwinger equation defines a unitary isomorphism $\widetilde{U}_\Sigma^{\Sigma'} : \widetilde{\Hilbert} \rightarrow \widetilde{\Hilbert}$ for every pair of Cauchy surfaces with 
\be\label{tildeUcompose}
\widetilde{U}_\Sigma^{\Sigma} = I\text{ and }
\widetilde{U}_{\Sigma'}^{\Sigma''} \widetilde{U}_\Sigma^{\Sigma'} = \widetilde{U}_\Sigma^{\Sigma''}\,.
\ee
These isomorphisms are related to the $U_{\Sigma}^{\Sigma'}$ of the hypersurface evolution via the free evolution. To this end, let $F_\Sigma^{\Sigma'}$ denote the free hypersurface evolution as defined, e.g., in Remark~\ref{rem:Fock}, and let, for simplicity, $\widetilde{\Hilbert}=\Hilbert_{\Sigma_0}$. Then the full time evolution (i.e., converted back from the interaction picture) is given by
\be\label{UtildeU}
	U_\Sigma^{\Sigma'}  = F_{\Sigma_0}^{\Sigma'}\, \widetilde{U}_{\Sigma}^{\Sigma'} \, F_\Sigma^{\Sigma_0}\,.
\ee
Correspondingly,
\be
	 \psi_\Sigma = F^\Sigma_{\Sigma_0} \, \widetilde{\psi}_\Sigma \,.
\ee
The Tomonaga-Schwinger picture can thus be regarded as an interaction picture version of the framework of hypersurface evolutions.

The Hilbert spaces $\Hilbert_\Sigma$ and PVMs $P_\Sigma$ can be taken to be the same ones as for the free evolution. If the free evolution satisfies the axioms of a hypersurface evolution, then so does the full time evolution (because \eqref{tildeUcompose} and \eqref{UtildeU} imply property (0)).

Let us turn to the properties (IL) and (PL).
As mentioned already in Remark~\ref{rem:oplus}, it is plausible that axiom (iii), factorization of the PVM, entails that
\be
\Hilbert_\Sigma=\int^\oplus_{\Gamma(\Sigma)} dq\, \bigotimes_{x\in q} \Hilbert_x
\ee
(and an analogous decomposition for the PVM). 
On a non-rigorous, heuristic level, this can be rewritten as a continuous tensor product
\be
\Hilbert_\Sigma = \bigotimes_{x\in \Sigma} \Gamma_{\pm}(\Hilbert_x)
\ee
relative to the measure $\mu_\Sigma$ and
\be
P_\Sigma = \bigotimes_{x\in\Sigma} P_x
\ee
with $P_x$ a PVM on $\Gamma(\{x\}) = \{\emptyset,\{x\}\}$ acting on $\Gamma_{\pm}(\Hilbert_x)$ with $P_x(\{\emptyset\})$ the projection to the 0-particle sector in the Fock space over $\Hilbert_x$ and $P_x(\{\{x\}\})$ the projection to the sum of all other sectors. (The PVMs on the bosonic and fermionic Fock spaces are unitarily equivalent.) 

In this notation, (IL) amounts to the condition that $\mathcal{H}_I(x)$, when transported to $\Sigma$ with the free evolution, acts non-trivially only on $\Gamma_{\pm}(\Hilbert_x)$ (which then ensures that the consistency condition \eqref{TSconsistent} holds), and (NCFV) to the condition that $\mathcal{H}_I(x)$ is block diagonal relative to the decomposition of the Fock space $\Gamma_{\pm}(\Hilbert_x)$ into the orthogonal sum of the 0-particle sector and the sum of all other sectors. As mentioned already before Proposition~\ref{prop:locNCFV}, it seems plausible that (PL) follows from (IL) and (NCFV).

A simple, explicit (but non-rigorous) example of $\mathcal{H}_I(x)$ satisfying these conditions is the emission-absorption model of \cite{pt:2013c}, a toy quantum field theory of two particle species called $x$-particles and $y$-particles, both Dirac particles, where $x$-particles can emit and absorb $y$-particles. The Hilbert space $\Hilbert_{\Sigma_0}$ is $\Gamma_-(L^2(\RRR^3,\CCC^4)) \otimes \Gamma_+(L^2(\RRR^3,\CCC^4))$, and
\be\label{HdIdef2}
\mathcal{H}_I(t,\vx)
= e^{iH_\free t} \biggl( \Bigl[ \sum_{r=1}^4 a_r^\dagger(\vx) \, a_r(\vx)\Bigr] \otimes \sum_{s=1}^4 \Bigl[ g_s^*\, b_s(\vx)+g_s \, b^\dagger_s(\vx) \Bigr] \biggr) e^{-iH_\free t}
\ee
with $H_\free$ the free Dirac Hamiltonian, $g\in\CCC^4$ a fixed spinor, and $a(\vx)$ and $b(\vx)$ the annihilation operators of the $x$- and $y$-particles, respectively. From \cite[Sec~4.3]{pt:2013c} one can conclude that the model satisfies the axioms of a hypersurface evolution, as well as (IL) and (NCFV).

\subsection{Multi-Time Wave Functions}
\label{sec:MT}

The central idea of the theory of multi-time wave function is to provide a covariant notion of wave function in the particle-position representation using a time coordinate for each particle in addition to its space coordinates. This idea dates back to the very beginnings of relativistic quantum theory \cite{dirac:1932,dfp:1932,bloch:1934,schwinger:1948,tomonaga:1946} and has over the years been considered again and again \cite{CVA:1983,DV82b,DV85,schweber:1961}. Recently, consistent multi-time equations for quantum field theories \cite{pt:2013c,pt:2013d} and other natural interactions \cite{lienert:2015a,lienert:2015b,lienert:2015c,LN:2015} have been described, and the conditions for consistency have been studied more closely \cite{pt:2013a,ND:2016}; see also the review \cite{LPT:2017}. 

As mentioned in the introduction, a multi-time wave function is
a spinor-valued function $\phi$ on the spacelike configurations.
It will be convenient in the following to use ordered configurations; the set of ordered spacelike configurations is
\begin{align}
\sS_\mathrm{o} &= \bigcup_{n=0}^\infty \sS^{(n)}_\mathrm{o} \quad \text{with}\\
\sS^{(n)}_\mathrm{o}&=\Bigl\{ (x_1,\ldots, x_n) \in \M^n: (x_i-x_j)^2<0 \: \forall i\neq j \Bigr\}\,.
\end{align}
Equivalently, $\phi=(\phi^{(0)},\phi^{(1)},\phi^{(2)},\ldots)$ with $\phi^{(n)}:\sS_\mathrm{o}^{(n)} \to \CCC^{k(n)}$. 
For example, for Dirac particles, the dimension of spin space is $k(n) = 4^n$.

\subsubsection{Multi-Time Evolution}

The evolution of multi-time wave functions is usually determined by a set of PDEs with one equation for each time coordinate:
\be
	i \frac{\partial}{\partial x^0_k} \phi^{(n)} (x_1,\ldots,x_n) = (H_k^{(n)} \phi) (x_1,\ldots,x_n ),~~~n\in \NNN.
	\label{eq:multitime}
\ee
Here, $(H_k^{(n)} \phi) (x_1,\ldots,x_n)$ is a differential expression in $\phi$, potentially involving all $\phi^{(m)}$. If $\phi^{(m)}$ for $m > n$ appears in $H^{(n)} \phi$, it must be evaluated at $m$ points constructed out of $x_1,\ldots,x_n$. 
The zero-particle amplitude $\phi^{(0)}$ is automatically time-independent.

Since $\phi$ needs to satisfy several equations simultaneously, the system \eqref{eq:multitime} is consistent iff 
\be\label{consistent}
\bigl[ H_j -i \partial_{x^0_j}, H_k - i \partial_{x^0_k} \bigr] =0
\ee
on $\sS_\mathrm{o}$; this condition is the analog of \eqref{TSconsistent}.

As mentioned already in \eqref{psiSigmaphi}, a hypersurface wave function can be defined by
$\psi_\Sigma = \phi|_{\Gamma(\Sigma)}$. In many cases, the inner product between wave functions on $\Gamma(\Sigma)$ is naturally given by \eqref{scpFock} or a similar expression, defining $\Hilbert_\Sigma$; then $P_\Sigma$ is multiplication by characteristic functions.
One expects that a multi-time wave function $\phi$ on $\sS_\mathrm{o}$ is determined by initial data on $\Gamma(\Sigma)$ for any Cauchy surface $\Sigma$. In this case, the multi-time evolution defines evolution operators $U_{\Sigma}^{\Sigma'}$ via $\phi|_{\Gamma(\Sigma)} \mapsto \phi|_{\Gamma(\Sigma')}$. If these operators are unitary, then they define a hypersurface evolution.

Conversely, the $\psi_\Sigma$ fit together to form a $\phi$ iff $\psi_\Sigma(q) = \psi_{\Sigma'}(q)$ for all $q$ that are ordered subsets of $\Sigma \cap \Sigma'$. 
In \cite[Sec.~4.3]{pt:2013c}, it is shown by means of the Tomonaga-Schwinger picture (though not rigorously) that this is indeed the case if the hypersurface evolution satisfies (IL) and (NCFV). Thus, it seems very likely that, assuming (IL) and (PL), the pictures provided by the multi-time wave function, the hypersurface evolution, and the Tomonaga-Schwinger equation are mutually equivalent.

It lies in the nature of multi-time theories that (NCFV) is automatically satisfied (whenever the $U_\Sigma^{\Sigma'}$ are unitary). That is because the number of time variables equals the number of particles, so the 0-particle sector is time independent. Put differently, when inserting points from $\Sigma$ into the space-time variables $x_k$ of $\phi$, there is nothing to be inserted into $\phi^{(0)}$, so $\psi^{(0)}_\Sigma$ is a complex number independent of $\Sigma$. Keep in mind that the vacuum subspace $\range P_\Sigma(\{\emptyset\})$ contains those $\psi_\Sigma$ for which at most $\psi^{(0)}_\Sigma$ is non-zero. Hence, if $\psi_\Sigma$ lies in the vacuum subspace and has norm 1, then $\psi^{(0)}_\Sigma$ is a complex number of modulus 1. Also for every other $\Sigma'$, it follows that $\psi^{(0)}_{\Sigma'}$ has modulus 1, while by unitarity $\psi_{\Sigma'}$ has norm 1, so all other sectors of $\psi_{\Sigma'}$ must vanish, and $\psi_{\Sigma'}$ lies in the vacuum subspace. This proves (NCFV).

Although many multi-time theories satisfy (IL), not all do. An example that does not is provided in \cite{ND:2016}; in this model, which is not Lorentz invariant, particles interact at a distance through a special, spin-dependent potential that is carefully contrived so as to respect the consistency condition \eqref{consistent}, which almost all potentials violate \cite{pt:2013a,ND:2016}.

\subsubsection{Examples of Multi-Time Theories Satisfying (IL)}

\begin{enumerate}
	\item \textit{Free Dirac particles.} The $n$-th sector of $\phi$ is an anti-symmetric function $\phi^{(n)} : \M^n \rightarrow (\CCC^{4})^{\otimes n}$ obeying the multi-time equations	
	\be
		i \frac{\partial}{\partial x_k^0} \phi^{(n)}(x_1,\ldots,x_n) = H^{\rm Dirac}_k  \phi^{(n)}(x_1,\ldots,x_n),~~k=1,\ldots,n,~n\in \NNN,
	\ee
	where $H^{\rm Dirac}_k $ is the free Dirac Hamiltonian acting on the variable $x_k$ and the spin index of the $k$-th particle. (We can also allow external electromagnetic fields.) It is a special situation of non-interacting particles that the multi-time wave function $\phi$ is defined even for non-spacelike configurations, i.e., on all of $\cup_{n=0}^\infty \M^n$. We need $\phi$ only on spacelike configurations in order to define a hypersurface evolution, and this hypersurface evolution is the same one as described in Remark~\ref{rem:Fock} of Section~\ref{sec:hypersurfaceevolution} and in Section~\ref{sec:free}; as mentioned already, it obeys (IL).
	
	\item \textit{Dirac particles in 1+1 dimensions with zero-range ($\delta$-)interactions.} In this Lorentz-invariant multi-time model \cite{lienert:2015a,lienert:2015b,lienert:2015c,LN:2015} (see \cite[Sec.~5]{LPT:2017} for a brief summary), particles interact upon contact, while particle number is conserved. The interaction is mathematically characterized through a boundary condition on the wave function. The combination of propagation locality, no creation or annihilation of particles, and no interaction at a distance implies (IL). The model is rigorously defined and proven to be consistent.
	
	\item \textit{Toy quantum field theory.} 
	This model \cite{pt:2013c} (see \cite[Sec.~6]{LPT:2017} for a brief summary) is a multi-time version of a simple quantum field theory in which $x$-particles can emit and absorb $y$-particles, where both $x$ and $y$ are Dirac particles. The model is not rigorously defined as it is ultraviolet divergent. However, on the non-rigorous level it has been shown to be consistent and to correspond to the Tomonaga-Schwinger equation with Hamitonian density \eqref{HdIdef2}, which obeys (IL).
\end{enumerate}

\section{Definitions Used for Detection Probabilities} \label{sec:detectors}

Before we formulate Theorem~\ref{thm:prob} in Section~\ref{sec:thm} below, we give a precise definition of the detection distribution on $\Sigma$ according to the ``horizontal piece approach.'' We first define the detection distribution for a finite partition of $\Sigma$, and afterwards (Proposition~\ref{prop:cdc}) the continuous detection distribution on $\Gamma(\Sigma)$. The definition makes use of a particular choice of Lorentz frame, which allows us to identify Minkowski space-time with $\RRR^4$. Recall that $\mu_\Sigma$ denotes the volume measure on $\Sigma$.

\begin{defn}
An \emph{admissible partition} of $\Sigma$ consists of finitely many subsets $\patch_1,\ldots,\patch_r$ of $\Sigma$ that are mutually disjoint, $\patch_\ell\cap \patch_m = \emptyset$ for $\ell\neq m$, and such that each $\patch_\ell$ is bounded and has boundary of measure zero, $\mu_\Sigma(\partial \patch_\ell)=0$. We set $\patch_{r+1}=\Sigma \setminus (\patch_1 \cup \ldots \cup \patch_r)$ to make $(\patch_1,\ldots,\patch_{r+1})$ a partition of $\Sigma$.
\end{defn}

The idea is to approximate detectors that measure for each $\patch_1,\ldots,\patch_r$ whether it contains at least one particle or none. It seems physically reasonable that the detector regions $\patch_1,\ldots, \patch_r$ are bounded, so they do not reach to infinity and are extended, for any given time resolution $\varepsilon>0$, over only finitely many time steps---whereas the $x^0$ coordinate may reach arbitrarily large values on $\Sigma$. The set $\patch_{r+1}$ comprises that region of $\Sigma$ where we do not place detectors and make no attempt at observing particles.

\begin{figure}[tp]
\centering
 \includegraphics[width=0.7\textwidth]{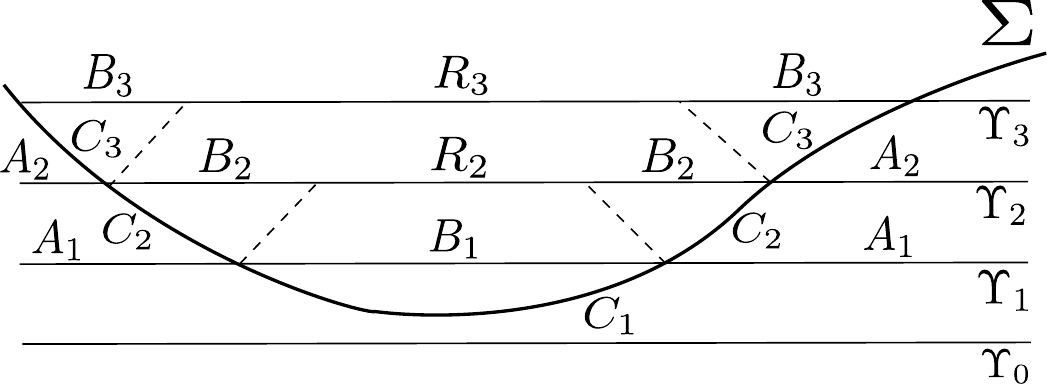}
 \caption{Illustration of the surfaces relevant to the detection process.}
 \label{fig:geom1d}
\end{figure}

\begin{defn}
	We introduce the following geometric notions (see Figure~\ref{fig:geom1d}).\\
	 For every $k\in\ZZZ$ let
\begin{align}
\Upsilon_k &:=\{x \in \M: x^0=k\varepsilon\}=\Sigma_{k\varepsilon},\\
A_k &:=\Upsilon_k \setminus \future(\Sigma),\\
C_k &:= \Bigr[ \Sigma \cap \past(\Upsilon_k) \Bigr]  \setminus \past(\Upsilon_{k-1}),\\
B_k &:=\Upsilon_k \cap \future(C_k) = \Gr(C_k,\Upsilon_k),\\
R_k &:=\Upsilon_k\setminus (A_k\cup B_k) \,. 
\end{align}
	Note that $\Upsilon_k= A_k\cup B_k \cup R_k$ is a partition of $\Upsilon_k$. 
	Furthermore, recall that $\pi$ denotes the ``vertical'' projection $\RRR^4\to\RRR^3$; for every $\ell=1,\ldots,r$ we set (see Figure~\ref{fig:geom})
\begin{align}
B_{k\ell} &:= \{x\in B_k: \pi(x) \in \pi(\patch_\ell)\},\\
C_{k\ell} &:= \{x\in C_k: \pi(x) \in \pi(\patch_\ell)\}\,.
\end{align}
\end{defn}

\begin{figure}[t]
\centering
 \includegraphics[width=0.9\textwidth]{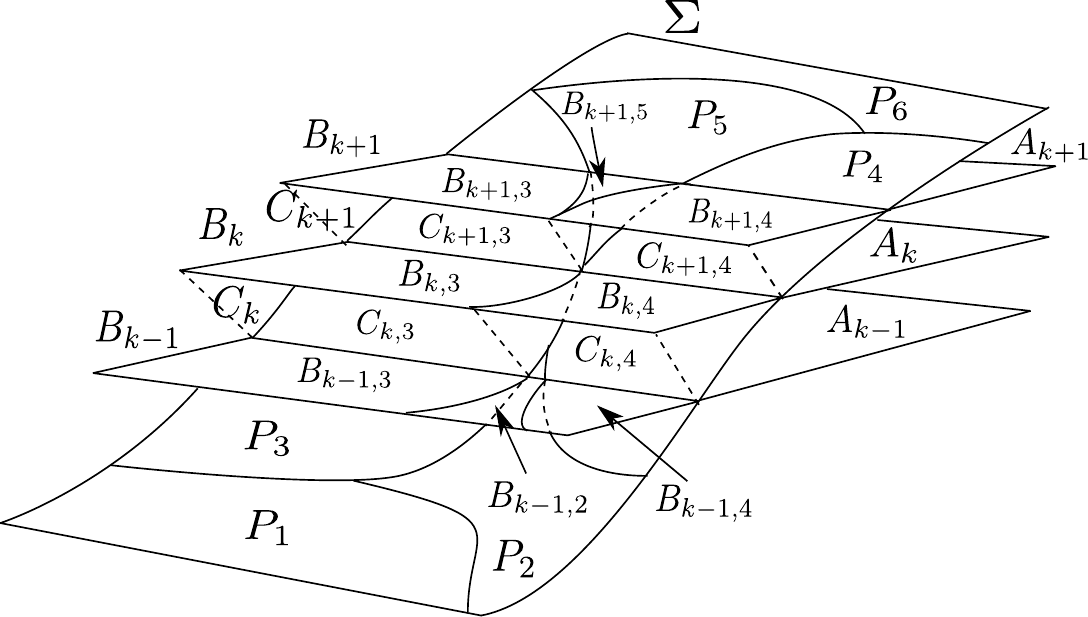}
 \caption{Illustration of the discretization scheme.}
 \label{fig:geom}
\end{figure}

Imagine that at every time $k\varepsilon$, an ideal quantum measurement is carried out for each $\ell$, detecting whether there is at least one particle in $B_{k\ell}$. The detection results are summarized by a 0-1-sequence $s = (s_{k\ell})$ where $s_{k\ell} = 0$ if no particle is detected in $B_{k\ell}$ and $s_{k\ell} = 1$ if at least one particle is detected in $B_{k\ell}$. After collapsing the quantum state respectively, we trace out all degrees of freedom outside $A_k$ and replace them by the vacuum state in order to avoid double detections. Finally, we are only interested in the coarse-grained results for the sets $\patch_\ell$ of the partition, given by a 0-1-sequence $L = (L_\ell)$, in the limit $\varepsilon \rightarrow 0$. These results for the $\patch_\ell$ are defined by $L_\ell = 0$ if $s_{k\ell} = 0$ for all $k$ and $L_\ell = 1$ if there is a $k$ such that $s_{k\ell} = 1$.
We now translate this intuitive thought into precise definitions.

\begin{defn}[detection map] Let $\sP = (\patch_1,...,\patch_{r+1})$ be an admissible partition of $\Sigma$, let $\varepsilon > 0$, and let 
\begin{align}
\kappa:= k_{\rm min} &= \min \bigl\{ k\in\ZZZ : \exists x\in \bigcup_{\ell=1}^r \patch_\ell: x^0\leq k\varepsilon \bigr\}\,,\\
K:= k_{\rm max} &= \min \bigl\{ k\in\ZZZ : \forall x\in\bigcup_{\ell=1}^r \patch_\ell: x^0 \leq k\varepsilon \bigr\}\,.
\end{align}
They exist because the $\patch_1,\ldots,\patch_r$ are bounded. 
Let $s_{k\ell}\in \{0,1\}$ for every $\kappa \leq k \leq K $ and every $\ell=1,\ldots,r$; we write $s = (s_{k\ell})_{\kappa \leq k \leq K ,\, 1 \leq \ell \leq r}$ as well as $s_k = (s_{k\ell})_{1 \leq \ell \leq r}$. Let $\sD(\Hilbert)$ denote the space of density matrices on the Hilbert space $\Hilbert$. 

For every $k=\kappa,\ldots,K$ we first define the \textit{detection map} $D_k(s_k) : \sD(\Hilbert_{\Upsilon_k}) \rightarrow \sD(\Hilbert_{\Upsilon_k})$ as follows. Let $\rho \in \sD(\Hilbert_{\Upsilon_k})$. Then
\be\label{Dkdef}
	D_k(s_k)  \rho := \frac{1}{\mathcal{N}_k(\rho)} \tr_{B_k\cup R_k} \Bigl( P_{\Upsilon_k}(M_B(s_k)) \rho \, P_{\Upsilon_k}(M_B(s_k)) \Bigr) \otimes P_{B_k\cup R_k}(\{\emptyset\}),
\ee
where $\tr_{B_k\cup R_k}$ denotes the partial trace over the tensor factor $\Hilbert_{B_k\cup R_k}$ in $\Hilbert_{\Upsilon_k} = \Hilbert_{A_k} \otimes \Hilbert_{B_k\cup R_k}$.\footnote{Should $B_k\cup R_k$ be empty, and in similar situations in the following, our formulas still apply: Then $\Hilbert_{B_k\cup R_k}=\CCC|\emptyset \rangle$ is 1-dimensional, so $\Hilbert_{\Upsilon_k}= \Hilbert_{A_k} \otimes \Hilbert_{B_k\cup R_k} \cong \Hilbert_{A_k}$, and $\Gamma(B_k\cup R_k)$ contains only one element, the empty configuration.} Furthermore,
\be
	\mbsk = \bigcap_{1 \leq \ell \leq r} \mbskl\,,
\ee
where
\be
	\mbskl = \left\{ \begin{array}{l} \emptyset(B_{k\ell})~{\rm if}~s_{k\ell} = 0\\ \exists(B_{k\ell})~{\rm if}~s_{k\ell} = 1 \end{array} \right.
\ee
stands for the part of the configuration space $\Gamma(\Upsilon_k)$ associated with the outcomes $s_k$. Finally, $\mathcal{N}_k(\rho) = \tr_{\Upsilon_k} \left( P_{\Upsilon_k}(\mbsk) \rho \, \right)$ is a normalization factor.

\end{defn}

Given the previous definition, we can now precisely state the measurement postulates we are using. Note that these are given by the Born rule at equal times (for density matrices) combined with a modified collapse postulate (tracing out particles which were detected and replacing them with the vacuum).

\begin{defn}(measurement postulates for equal times.)
	\begin{enumerate}
			\item \textit{Born rule.} Given the density matrix $\rho_{\Upsilon_k} \in \sD(\Hilbert_{\Upsilon_k})$, the conditional probability for the detection results $s_k$ given the previous detection results $s_\kappa,\ldots,s_{k-1}$ (relative to a partition $\sP$, a time discretization $\varepsilon$ and an initial density matrix $\rho_0 \in \sD(\Hilbert_{\Upsilon_0})$) is obtained by
		\be
			\Prob^{\rho_0,\varepsilon}_{\sP}(s_k | s_\kappa,\ldots,s_{k-1}) 
			= \Prob^{\varepsilon}_{\sP}(s_k | \rho_{\Upsilon_k}) 
			= \tr_{\Upsilon_k} \Bigl( P_{\Upsilon_k}(M_B(s_k))\,  \rho_{\Upsilon_k} \Bigr)
			\label{eq:condborn}
		\ee
		for any $k=\kappa,\ldots, K$. For $k=\kappa$, the left-hand side should be understood as $\Prob^{\rho_0,\varepsilon}_{\sP}(s_\kappa)$.
		\item \textit{Collapse rule.} If a result given by $s_k$ is detected on $\Upsilon_k$, the quantum state  $\rho_{\Upsilon_k} \in \sD(\Hilbert_{\Upsilon_k})$ collapses to $\rho_{\Upsilon_k}' = D_k(s_k) \rho_{\Upsilon_k}$. The new quantum state $\rho_{\Upsilon_{k+1}} \in \sD(\Hilbert_{\Upsilon_{k+1}})$ is given by
		\be
			\rho_{\Upsilon_{k+1}} = U_{\Upsilon_k}^{\Upsilon_{k+1}} \rho_{\Upsilon_k}' U_{\Upsilon_{k+1}}^{\Upsilon_k}.
			\label{eq:collapserule}
		\ee
		(We note that the detection mapping $D_k$ also removes particles in $B_k \setminus \bigcup_{\ell=1}^r B_{k\ell}$ that actually were not detected but did cross $B_k$ and therefore should not be counted again at a later time).
		\item \textit{First surface ($\Upsilon_\kappa$) rule.} In parallel to removing detected particles, we also remove after time $\kappa-1$ those particles that have crossed some $B_k$ with $k<\kappa$ (in places associated with the region $\patch_{r+1}$ without detectors). That is, we postulate that \eqref{eq:collapserule} also holds for $k=\kappa-1$ but with $\rho'_{\Upsilon_{\kappa-1}}$ given (not by a collapse via $D_k$ but instead) by
		\be\label{firstsurface}
			\rho'_{\Upsilon_{\kappa-1}} =  \tr_{B_{\kappa-1}\cup R_{\kappa-1}}\Bigl(U_{\Upsilon_0}^{\Upsilon_{\kappa-1}}\, \rho_0 \, U_{\Upsilon_{\kappa-1}}^{\Upsilon_0} \Bigr) \otimes P_{B_{\kappa-1}\cup R_{\kappa-1}} (\{\emptyset\}) \,.
		\ee
		In case $B_{\kappa-1}\cup R_{\kappa-1}=\emptyset$, this rule reduces to
		\be
		\rho'_{\Upsilon_{\kappa-1}} = U^{\Upsilon_0}_{\Upsilon_{\kappa-1}}\, \rho_0 \, U^{\Upsilon_{\kappa-1}}_{\Upsilon_0}\quad  \text{and} \quad 
		\rho_{\Upsilon_\kappa} = U^{\Upsilon_0}_{\Upsilon_{\kappa}}\, \rho_0 \, U^{\Upsilon_{\kappa}}_{\Upsilon_0}\,.
		\ee
		In case $\kappa=0$, notice that necessarily $B_{\kappa-1}=\emptyset = R_{\kappa-1}$. Note that \eqref{firstsurface} has trace 1.
	\end{enumerate}

\end{defn}

From these postulates for each time step it now follows, by the usual rule for conditional probabilities, that the probability distribution $\Prob^{\rho_0,\varepsilon}_{\sP}(s)$ on the set of detection results $\{0,1\}^{(K -\kappa+1)r}$ (relative to a partition $\sP$, an initial density matrix $\rho_0$ and a time discretization $\varepsilon$) is given by
	\begin{align}		
		\Prob^{\rho_0,\varepsilon}_{\sP}(s) 
		&= \Prob^{\rho_0,\varepsilon}_{\sP}(s_{K } | s_{\kappa},\ldots,s_{K -1}) \, 
		\Prob^{\rho_0,\varepsilon}_{\sP}(s_{K -1} | s_1,\ldots,s_{K -2}) \cdots  \Prob^{\varepsilon}_{\sP,\rho_0}(s_{\kappa}) \nonumber\\[3mm]
		&= \Prob^{\varepsilon}_{\sP}(s_{K } | \rho_{\Upsilon_K}) \, 
		\Prob^{\varepsilon}_{\sP}(s_{K -1} | \rho_{\Upsilon_{K -1}}) \cdots  \Prob^{\varepsilon}_{\sP}(s_{\kappa}|\rho_{\Upsilon_\kappa})\,.\label{totalprod}
	\end{align}

We are now ready to define the detection distribution that our main theorem is concerned with.

\begin{defn}
The \emph{detection distribution} $\Prob^{\psi_0}_{\det,\sP}$ relative to $\sP$ and $\psi=\psi_0\in\Hilbert_{\Sigma_0}$ with $\|\psi\|=1$ is the measure on $\{0,1\}^r$ defined by
\be\label{partitionprob}
\Prob^{\psi}_{\det,\sP}(L) = \lim_{\varepsilon\to 0}\Prob^{\psi,\varepsilon}_{\det,\sP}(L)\,,
\ee
where
\be
	\Prob^{\psi,\varepsilon}_{\det,\sP}(L) = \sum_{s:L} \Prob^{|\psi \rangle \langle \psi |,\varepsilon}_{\sP}(s)
	\label{eq:discretizeddistribution}
\ee
Here, the sum on the right-hand side runs over all $s$ compatible with $L$, that is, over all $s \in \{ 0,1\}^{(K -\kappa+1) r}$ such that for all $1 \leq \ell \leq r$: ($s_{k\ell} = 0 \, \forall k$ if $L_\ell = 0$) and ($\exists k : s_{k\ell} = 1$ if $L_\ell$ = 1). Put differently, the sum runs over all $s$ leading to the coarse-grained detection results $L$.
\end{defn}

\paragraph{Remarks.}
\begin{enumerate}
\setcounter{enumi}{\theremarks}
	\item One could also define the detection distribution for a mixed initial density matrix $\rho_0$, replacing $\Prob^{|\psi \rangle \langle \psi |,\varepsilon}_{\sP}(s)$ with $\Prob^{\rho_0,\varepsilon}_{\sP}(s)$ in \eqref{eq:discretizeddistribution}. 
	\item The question arises whether the limit $\varepsilon\to 0$ in \eqref{partitionprob} always exists. Theorem~\ref{thm:prob} will show, among other things, that it does. One could also set up a modified definition of $\Prob^{\psi}_{\det,\sP}$ that would allow for detecting particles in sets $\patch_\ell$ that are \emph{unbounded}, but that would come at the cost of increased complexity of the reasoning and, as already mentioned, it would seem physically unrealistic to have infinitely extended detectors and infinitely many time steps in which detection is attempted. Moreover, the probabilities for bounded sets already determine the probability distribution completely, as Proposition \ref{prop:cdc} shows.
\end{enumerate}
\setcounter{remarks}{\theenumi}

\section{Statement of the Main Results} \label{sec:thm}

\begin{prop}\label{prop:cdc}
(Continuum detection distribution) 
For every $\psi\in\Hilbert_0$ with $\|\psi\|=1$, there is at most one probability distribution $\Prob^\psi_{\det}$ on $\Gamma(\Sigma)$ that agrees with $\Prob^\psi_{\det,\sP}$ for every admissible partition $\sP=(\patch_1,\ldots,\patch_{r+1})$ in the sense that, for every $L \in \{0,1\}^r$,
\be
\Prob^\psi_{\det} \Bigl( \msl \Bigr) = \Prob^\psi_{\det,\sP} (L)\,.
\ee
Here,
\be
	\msl = \bigcap_{\ell=1}^r \msll
\ee
with
\be
	\msll = \left\{ \begin{array}{l} \emptyset(\patch_\ell)~{\rm if}~L_\ell = 0\\ \exists(\patch_\ell)~{\rm if}~L_\ell = 1 \end{array} \right. .
\ee
If it exists, that distribution $\Prob^\psi_{\det}$ will be called the \emph{continuum detection distribution} of $\psi$ on $\Gamma(\Sigma)$.
\end{prop}

We give the proof in Section~\ref{sec:proofcdc}. Informally, Proposition~\ref{prop:cdc} asserts that since we can choose the detector regions $\patch_1,\ldots,\patch_r$ arbitrarily (as long as they are bounded and have boundaries of measure zero), there is no more than one distribution on the continuous configuration space $\Gamma(\Sigma)$ that agrees with $\Prob^\psi_{\det,\sP}$ for every choice of $\sP$.

\setcounter{thm}{0}
\begin{thm}\label{thm:prob}
Suppose we are given a hypersurface evolution $(\Hilbert_\circ,P_\circ,U_\circ^\circ)$ in Minkowski space-time satisfying (IL) and (PL).
Let $\psi\in\Hilbert_{\Sigma_0}$ with $\|\psi\|=1$, let $\Sigma$ be a Cauchy hypersurface in the future of $\Sigma_0$, and let $\sP=(\patch_1,\ldots,\patch_{r+1})$ be an admissible partition of $\Sigma$. For every $L\in \{0,1\}^r$, the limit \eqref{partitionprob} exists and is equal to
\be
\Prob^\psi_{\det,\sP}(L) = \langle \psi_\Sigma | P_\Sigma(\msl) | \psi_\Sigma \rangle\,.
\ee
Put differently, for every $\psi\in\Hilbert_{\Sigma_0}$ with $\|\psi\|=1$, there is a continuum detection distribution $\Prob^\psi_{\det}$ on $\Gamma(\Sigma)$, and it is given by
\be
\Prob^\psi_{\det}(\cdot) = \scp{\psi_\Sigma}{P_\Sigma(\cdot)|\psi_\Sigma}\,.
\ee
\end{thm}

The proof is given in Section~\ref{sec:proofofmaintheorem}. Theorem~\ref{thm:prob} is our precise formulation of the curved Born rule. We have given an informal summary around \eqref{eq:generalizedbornrule}.

\section{Proofs} \label{sec:proofs}

\subsection{Proof of Proposition~\ref{prop:cdc}}
\label{sec:proofcdc}

\begin{proof}
Recall that $A\subseteq\Sigma$ is measurable iff $\pi(A)$ is, and $\pi$ is the projection $\RRR^4\to\RRR^3$. 
Let $\mathcal{B}_R(\vx)$ denote the open ball of radius $R>0$ around $\vx$ in $\RRR^3$. 
We will show that already those admissible partitions with $r=1$ and $\pi(\patch_1)$ the union of finitely many open balls (of finite radii) determine $\Prob^\psi_{\det}$ uniquely, in fact already the probabilities of $\emptyset(\patch_1)$ do.

It is a known fact that a probability measure $\Prob$ on a measurable space $(\Omega,\sF)$ is uniquely determined by its values on a $\cap$-stable generator of $\sF$. Thus, it suffices to show that the family
\be
\sA:=\Bigl\{\emptyset(A): \pi(A) \text{ is the union of finitely many open balls}  \Bigr\}\subset \sB(\Gamma(\Sigma))
\ee
is a $\cap$-stable generator. It is clear that $\sA$ is $\cap$-stable because
\be
\emptyset(A)\cap \emptyset(B)=\emptyset(A\cup B)\,,
\ee
and $\pi(A\cup B)=\pi(A)\cup \pi(B)$ is again a union of finitely many open balls.
So it remains to show that $\sigma(\sA)$, the $\sigma$-algebra generated by $\sA$, contains $\sB(\Gamma(\Sigma))$ and, therefore, coincides with it. In fact, we prove the slightly stronger statement that already for the smaller family
\be
\sC = \Bigl\{\emptyset\bigl(\mathcal{B}_R(\vx) \bigr): R>0, \vx\in \RRR^3 \Bigr\}\subset\pi(\sA) \subset \sB(\Gamma(\RRR^3))\,,
\ee
$\sigma(\sC)$ coincides with $\sB(\Gamma(\RRR^3))$.

To this end, let $\{\vx_1,\vx_2,\ldots\}$ be a countable dense set in $\RRR^3$ (e.g., $\QQQ^3$), let $f:\NNN\times \NNN\to \NNN$ be a bijection, and set
\be
A_{f(m,n)} = \mathcal{B}_{1/m}(\vx_n)\,.
\ee
This provides us with a point-separating sequence of balls $A_i$ in $\RRR^3$. Note that $\exists(A_i) = \emptyset(A_i)^c \in \sigma(\sC)$ and that, for every $n\in\NNN$,
\be\label{Mndef}
M_{n} :=  \hspace{-2mm} \bigcup_{\substack{i_1\ldots i_n=1\\A_{i_1}\ldots A_{i_n}\text{ pairw.~disjoint}}}^\infty
\hspace{-5mm} \exists(A_{i_1})\cap \ldots \cap \exists(A_{i_n}) \quad \in \:\: \sigma(\sC)\,.
\ee

\medskip

\noindent\underline{Claim.} $\displaystyle M_{n} = \bigl\{ q\in\Gamma(\RRR^3): \#q \geq n \bigr\}$

\medskip

\noindent\underline{Proof of the claim:}
Let $q=\{\vy_1,\ldots,\vy_m\}\in \Gamma(\RRR^3)$ with $m\geq n$. There are mutually disjoint $A_{i_1},\ldots,A_{i_m}$ with $\vy_j\in A_{i_j}$, so, ignoring those with $j>n$ (if any),
\be
q\in \exists(A_{i_1}) \cap \ldots \cap \exists(A_{i_n})\,,
\ee
so $q\in M_{n}$. On the other hand, if $q\in M_n$, then, by the definition \eqref{Mndef}, $q\in \exists(A_{i_1}) \cap \ldots \cap \exists(A_{i_n})$ for some pairwise disjoint $A_{i_1},\ldots,A_{i_n}$, so $\#q\geq n$. This proves the claim.\hfill$\square$

\bigskip

We continue the proof of Proposition~\ref{prop:cdc}. It follows from the claim that also
\be
\Gamma_n(\RRR^3):=\bigl\{ q\in\Gamma(\RRR^3): \#q=n\bigr\} = \{\#q\geq n\} \setminus \{\#q\geq n+1\} \in \sigma(\sC)\,.
\ee
It remains to show that $\sB(\Gamma_n(\RRR^3))\subset\sigma(\sC)$. Since $\Gamma_n(\RRR^3) = (\RRR^3)^n/\text{permutations}$ and $\sB((\RRR^3)^n)$ is generated by the Cartesian products of $n$ open 3-balls, we know that $\sB(\Gamma_n(\RRR^3))$ is generated by the unordered (symmetrized) Cartesian products of $n$ open 3-balls. Let us focus on $n=2$ (the generalization to arbitrary $n$ is straightforward). If two given 3-balls $\mathcal{B}_{R_1}(\vz_1),\mathcal{B}_{R_2}(\vz_2)$ are disjoint, then their unordered product can be written as $\Gamma_2(\RRR^3)\cap \exists(\mathcal{B}_{R_1}(\vz_1))\cap \exists(\mathcal{B}_{R_2}(\vz_2))$, and thus belongs to $\sigma(\sC)$. If the two balls are not disjoint, then their product in $\Gamma_2(\RRR^3)$ can be written as a countable union of products of two disjoint balls, and thus again belongs to $\sigma(\sC)$.
\end{proof}

\subsection{Proof of Theorem~\ref{thm:prob}}
\label{sec:proofofmaintheorem}

The proof is structured as follows. Firstly, we determine a closed expression for $\Prob^{|\psi\rangle \langle \psi|,\varepsilon}_{\sP}(s)$ (Section~\ref{sec:expression}). Secondly, this expression is used to obtain upper and lower bounds for $\Prob^{\psi,\varepsilon}_{\det,\sP}(L)$ that have the form $\langle \psi_\Sigma | P_\Sigma(M_i^\varepsilon) | \psi_\Sigma \rangle$ for suitable configuration space sets $M_i^\varepsilon \subset \Gamma(\Sigma),~i=1,2$ (Section~\ref{sec:bounds}). Thirdly, we show that these bounds converge to the same expression in the limit $\varepsilon \rightarrow 0$ (Section~\ref{sec:limit}). This demonstrates the existence of $\Prob^{\psi}_{\det,\sP}(L)$.

\subsubsection{Closed Expression for $\Prob^{|\psi \rangle \langle \psi |,\varepsilon}_{\sP}(s)$} \label{sec:expression}

\begin{prop} \label{prop:expression}
	\be
		\Prob^{|\psi \rangle \langle \psi |,\varepsilon}_{\sP}(s) = \langle \psi_\Sigma | P(s) | \psi_\Sigma \rangle\,,
		\label{eq:expression}
	\ee
	where (using the abbreviation $C = \bigcup_{k=\kappa}^{K } C_k$)
	\begin{align}
		&P(s) :\Hilbert_\Sigma \cong \left( \bigotimes_{k=\kappa}^{K } \Hilbert_{C_k} \right) \otimes \Hilbert_{\Sigma \backslash C} \rightarrow \Hilbert_\Sigma,\nonumber\\
		&P(s) = \left( \bigotimes_{k=\kappa}^{K }  \widetilde{P}_{C_k}\right) \otimes I_{\Sigma \backslash C}.
		\label{eq:ps}
	\end{align}
	Here, $\widetilde{P}_{C_k}$ is defined as follows. According to the factorization of the PVM \eqref{PAB}, we have that\footnote{Note that $\exists_\Sigma(B_{k\ell}) = \Gamma(\Sigma \backslash B_k) \times \exists_{B_k}(B_{k\ell})$ and $\emptyset_\Sigma(B_{k\ell}) = \Gamma(\Sigma \backslash B_{k}) \times \emptyset_{B_k}(B_{k\ell})$.}
	\be
		P_{\Upsilon_k}(\mbsk) = I_{A_k \cup R_k} \otimes P_{B_k} \bigl(\nbsk \bigr),
		\label{eq:pvmdecomp}
	\ee
	where $\nbsk$ is defined like $\mbsk$ but exchanging $\exists(\cdot)$ with $\exists_{B_k}(\cdot)$ and $\emptyset(\cdot)$ with $\emptyset_{B_k}(\cdot)$.
	Then	
	\be
		\widetilde{P}_{C_k} = W_{B_k}^{C_k} \: P_{B_k}\bigl(\nbsk \bigr)  \: W_{C_k}^{B_k}\,,
	\ee
	where $W_{C_k}^{B_k} : \Hilbert_{C_k}\to\Hilbert_{B_k}$ are the reduced evolution operators defined in \eqref{eq:defw}. 
\end{prop}
	
Note that $B_k = \Gr(C_k,\Upsilon_k)$. Furthermore, the phase ambiguity of $W_{C_k}^{B_k}$ corresponding to the choice of vacuum vectors disappears here because $W$ appears twice with opposite phases. Note also that $\widetilde{P}_{C_k}$, and thus $P(s)$, is in general not a projection because $W_{C_k}^{B_k}$ is in general not unitary; it is an isometry but in general not surjective, that is, it is unitary to a subspace of $\Hilbert_{B_k}$; correspondingly, $W_{B_k}^{C_k}W_{C_k}^{B_k}=I_{C_k}$, but $W_{C_k}^{B_k} W_{B_k}^{C_k}$ is the projection onto said subspace of $\Hilbert_{B_k}$. Finally, keep in mind that $P_{B_k}(\nbsk)$ \emph{is} a projection.

Before proving Proposition \ref{prop:expression}, we introduce an auxiliary density matrix that will facilitate the proof.

\begin{defn}[auxiliary density matrix]
	Let $\rho_0 \in \sD(\Hilbert_{\Upsilon_0})$. 
	For every $k = \kappa-1,\ldots,K $, we define the \textit{auxiliary density matrix} $\rho_{A_k} \in \sD(\Hilbert_{A_k})$ by
	\begin{align}
		\rho_{A_{\kappa-1}} &=  \tr_{B_{\kappa-1}\cup R_{\kappa-1}} \Bigl( U_{\Upsilon_0}^{\Upsilon_{\kappa-1}} \: \rho_0 \: U^{\Upsilon_0}_{\Upsilon_{\kappa-1}}  \Bigr),\label{eq:rhoAkappa}\\
		\rho_{A_{k}} &= \frac{1}{\mathcal{N}_{A_{k}}} \tr_{C_{k}} \left(  [I_{A_{k}} \otimes \widetilde{P}_{C_{k}} ] \, U_{A_{k-1}}^{A_{k} \cup C_{k}} \rho_{A_{k-1}} U_{A_{k} \cup C_{k}}^{A_{k-1}} \right),
					\label{eq:rhoak}
	\end{align}
	where $\mathcal{N}_{A_k}$ is the normalization factor yielding $\tr_{A_k}(\rho_{A_k}) = 1$. In case $\kappa=0$, $B_{\kappa-1}=\emptyset=R_{\kappa-1}$, so \eqref{eq:rhoAkappa} reduces to $\rho_{A_{-1}}= U_{\Upsilon_0}^{\Upsilon_{-1}} \: \rho_0 \: U^{\Upsilon_0}_{\Upsilon_{-1}}$, and $A_0\cup C_0 = \Upsilon_0$.
\end{defn}

	 Note that \eqref{eq:rhoAkappa} has trace 1. We also note that by the cyclic property of the partial trace, the expression inside $\tr_{C_k}$ in \eqref{eq:rhoak} can be replaced by
\be
	 \bigl[ I_{A_{k}} \otimes \widetilde{P}_{C_{k}}^{1/2} \bigr] \, U_{A_{k-1}}^{A_{k} \cup C_{k}} \rho_{A_{k-1}} U_{A_{k} \cup C_{k}}^{A_{k-1}}\, \bigl[ I_{A_{k}} \otimes \widetilde{P}_{C_{k}}^{1/2} \bigr]\,, 
\ee
which makes it evident that $\rho_{A_k}$ is self-adjoint.
	
	We have used (IL) to obtain the existence of the unitary operators $U^{A_{k+1} \cup C_{k+1}}_{A_k}$. This is possible as $A_k$ and $A_{k+1} \cup C_{k+1}$ are grown (and shrunk) sets of each other. More precisely,
	\be\label{Xidef}
	\Xi_{k,k+1}:=A_{k+1} \cup C_{k+1}\cup B_k \cup R_k
	\ee
is a Cauchy surface that has $B_k\cup R_k$ in common with the Cauchy surface $\Upsilon_k=A_k \cup B_k \cup R_k$. We will also need the following fact.

\begin{prop}\label{prop:WABC}
\be\label{eq:WABC}
W_{A_k}^{A_{k+1} \cup B_{k+1}} = \Bigl( I_{A_{k+1}} \otimes W_{C_{k+1}}^{B_{k+1}} \Bigr) \: U_{A_k}^{A_{k+1}\cup C_{k+1}}\,.
\ee
\end{prop}

\begin{proof}
The defining property \eqref{eq:woperator} of $W$ is that
	\be\label{eq:woperator2}
		U_\Sigma^{\Sigma'} \psi_A\otimes |\emptyset(A^c)\rangle 
		= \Bigl( W_A^{\Gr(A,\Sigma')} \psi_A \Bigr) \otimes \bigl| \emptyset(\Sr(A^c,\Sigma')) \bigr\rangle,
	\ee
we need to verify that the right-hand side of \eqref{eq:WABC} satisfies this property for $\Sigma=\Upsilon_k$, $\Sigma'=\Upsilon_{k+1}$, and $A=A_k$; that is, we need to show that
	\be\label{zuzeigen}
		U_{\Upsilon_k}^{\Upsilon_{k+1}} \psi_{A_k}\otimes |\emptyset(B_k\cup R_k)\rangle 
		= \Bigl[\Bigl( I_{A_{k+1}} \otimes W_{C_{k+1}}^{B_{k+1}} \Bigr) \: U_{A_k}^{A_{k+1}\cup C_{k+1}} \psi_{A_{k}} \Bigr] \otimes \bigl| \emptyset(R_{k+1}) \bigr\rangle.
	\ee

The $W_{C_{k+1}}^{B_{k+1}}$ on the right-hand side of \eqref{eq:WABC} is characterized by the property that 
	\be\label{Wright}
		U_{\Xi_{k,k+1}}^{\Upsilon_{k+1}} \psi_{C_{k+1}} \otimes \bigl|\emptyset(A_{k+1}\cup B_k \cup R_k)\bigr\rangle = \Bigl( W_{C_{k+1}}^{B_{k+1}} \psi_{C_{k+1}} \Bigr) \otimes \bigl| \emptyset( A_{k+1}\cup R_{k+1}) \bigr\rangle
	\ee
with $\Xi_{k,k+1}$ as in \eqref{Xidef}.
By (IL), $U_{\Xi_{k,k+1}}^{\Upsilon_{k+1}} = I_{A_{k+1}}\otimes U_{C_{k+1} \cup B_k \cup R_k}^{B_{k+1}\cup R_{k+1}}$, so a factor $|\emptyset(A_{k+1})\rangle$ splits off from \eqref{Wright}, i.e.,
	\be\label{Wright2}
		U_{C_{k+1} \cup B_k \cup R_k}^{B_{k+1}\cup R_{k+1}} \psi_{C_{k+1}} \otimes |\emptyset(B_k \cup R_k)\rangle = \Bigl( W_{C_{k+1}}^{B_{k+1}} \psi_{C_{k+1}} \Bigr) \otimes \bigl| \emptyset( R_{k+1}) \bigr\rangle,
	\ee
and thus
	\be\label{Wright3}
		U_{\Xi_{k,k+1}}^{\Upsilon_{k+1}} \psi_{A_{k+1}\cup C_{k+1}} \otimes |\emptyset(B_k \cup R_k)\rangle = \Bigl[ \Bigl( I_{A_{k+1}}\otimes W_{C_{k+1}}^{B_{k+1}} \Bigr) \psi_{A_{k+1} \cup C_{k+1}} \Bigr] \otimes \bigl| \emptyset( R_{k+1}) \bigr\rangle.
	\ee
Inserting $\psi_{A_{k+1}\cup C_{k+1}}=U^{A_{k+1}\cup C_{k+1}}_{A_k} \psi_{A_k}$ yields \eqref{zuzeigen}.
\end{proof}

\begin{prop}
	\label{prop:properties}
	For $\kappa-1\leq k \leq K$, $\rho_{A_k}$ has the following properties:
	\begin{enumerate}
		\item $\rho_{\Upsilon_k}' = \rho_{A_k} \otimes P_{B_k \cup R_k}(\{\emptyset\})$.
		\item $\rho_{\Upsilon_{k+1}} = \left(W_{A_k}^{A_{k+1} \cup B_{k+1}} \rho_{A_k} W_{A_{k+1} \cup B_{k+1}}^{A_k} \right) \otimes  P_{R_{k+1}}(\{ \emptyset \})$.
		\item The normalization constants of $\rho'_{\Upsilon_k}$ and $\rho_{A_k}$ coincide: $\mathcal{N}_k(\rho_{\Upsilon_k}) = \mathcal{N}_{A_k}$.
		\item Expression of conditional probabilities via $\rho_{A_k}$. For $\kappa\leq k \leq K$, 		\be
			\Prob_{\sP}^\varepsilon(s_k | \rho_{\Upsilon_k}) = \mathcal{N}_{A_k} \,.
		\ee
	\end{enumerate}
\end{prop}

\begin{proof}
	We prove statements 1., 2., and 3.\ together by induction along $k$, beginning with $k=\kappa-1$.  Statement 1.\ is immediate from the definitions \eqref{firstsurface} and \eqref{eq:rhoAkappa}, and 3.\ is true since the normalization constants of $\rho'_{\Upsilon_{\kappa-1}}$ and $\rho_{A_{\kappa-1}}$ are both 1. 
	For 2., we generally have that $\rho_{\Upsilon_{k+1}} = U_{\Upsilon_k}^{\Upsilon_{k+1}} \rho_{\Upsilon_k}' U_{\Upsilon_{k+1}}^{\Upsilon_k}$. As $\rho_{\Upsilon_{\kappa-1}}'$ is concentrated in $A_{\kappa-1}$ by definition \eqref{firstsurface}, we can use the reduced evolution operators (cf.\ Proposition~\ref{prop:reducedevol}) to write (note $R_{k+1} = \Sr(B_k \cup R_k,\Upsilon_{k+1})$): $\rho_{\Upsilon_{\kappa}} = (W_{A_{\kappa-1}}^{A_{\kappa}\cup B_{\kappa}} \: \rho_{A_{\kappa-1}} \: W_{A_{\kappa}\cup B_{\kappa}}^{A_{\kappa-1}}) \otimes P_{R_{\kappa}} (\{\emptyset\})$.
	
	For the induction step, let claims 1., 2., and 3.\ be true for some $k$. We have: $\rho_{\Upsilon_{k+1}}' = D_{k+1}(s_{k+1}) \rho_{\Upsilon_{k+1}}$. According to the induction assumption for 2.,
	\begin{align}
		\rho_{\Upsilon_{k+1}}' &= D_{k+1}(s_{k+1}) \left[ \left(W_{A_k}^{A_{k+1} \cup B_{k+1}} \rho_{A_k} W_{A_{k+1} \cup B_{k+1}}^{A_k} \right) \otimes P_{R_{k+1}} (\{\emptyset\}) \right] \nonumber\\
		&= \frac{1}{\mathcal{N}_{k+1}(\rho_{\Upsilon_k})} \tr_{B_{k+1} \cup R_{k+1}} \left[ P_{\Upsilon_{k+1}}(M_B(s_{k+1})) \left(W_{A_k}^{A_{k+1} \cup B_{k+1}} \rho_{A_k} W_{A_{k+1} \cup B_{k+1}}^{A_k} \right. \right. \nonumber\\ 
		& ~~~~~~~~~~~~~~~~~~~~\otimes P_{R_{k+1}}(\{\emptyset\})  \Big) P_{\Upsilon_{k+1}}(M_B(s_{k+1})) \Big] \otimes  P_{B_{k+1}\cup R_{k+1}}(\{\emptyset\})\nonumber\\
		&= \frac{1}{\mathcal{N}_{k+1}(\rho_{\Upsilon_k})} \tr_{B_{k+1}} \biggl[ \bigl[ I_{A_{k+1}} \otimes P_{B_{k+1}} (N_B(s_{k+1})) \bigr] \, \Bigl(W_{A_k}^{A_{k+1} \cup B_{k+1}} \rho_{A_k} W_{A_{k+1} \cup B_{k+1}}^{A_k} \Bigr) \biggr] \nonumber\\ 
		& ~~~~~~~~~~~~~~~~~~~~\otimes P_{B_{k+1}\cup R_{k+1}}(\{\emptyset\})
		\label{eq:rhocalc1}
	\end{align} 
using \eqref{eq:pvmdecomp} and uniqueness of the vacuum \eqref{UV}. 
By Proposition~\ref{prop:WABC}, and using the cyclic property of the partial trace, we obtain that
\begin{align}
	&\tr_{B_{k+1}} \biggl[ \bigl[ I_{A_{k+1}} \otimes P_{B_{k+1}}(N_B(s_{k+1})) \bigr] \left(W_{A_k}^{A_{k+1} \cup B_{k+1}} \rho_{A_k} W_{A_{k+1} \cup B_{k+1}}^{A_k} \right) \biggr]\nonumber\\
	&= \tr_{B_{k+1}} \biggl[ \bigl[ I_{A_{k+1}} \otimes P_{B_{k+1}}(N_B(s_{k+1}))\bigr]  \bigl[ I_{A_{k+1}} \otimes W_{C_{k+1}}^{B_{k+1}}\bigr]  \left(U_{A_k}^{A_{k+1} \cup C_{k+1}} \rho_{A_k} U_{A_{k+1} \cup C_{k+1}}^{A_k} \right) \bigl[ I_{A_{k+1}} \otimes W_{B_{k+1}}^{C_{k+1}} \bigr] \biggr]\nonumber\\
	&= \tr_{C_{k+1}} \biggl[  \Bigl[I_{A_{k+1}} \otimes \left( W_{B_{k+1}}^{C_{k+1}} P_{B_{k+1}}(N_B(s_{k+1})) W_{C_{k+1}}^{B_{k+1}} \right) \Bigr] 
	\left(U_{A_k}^{A_{k+1} \cup C_{k+1}} \rho_{A_k} U_{A_{k+1} \cup C_{k+1}}^{A_k} \right)
	 \biggr]
\end{align}
Comparing this expression with \eqref{eq:rhocalc1} and the definition \eqref{eq:rhoak} of $\rho_{A_{k+1}}$
yields that $\mathcal{N}_{k+1}(\rho_{\Upsilon_{k+1}})=\mathcal{N}_{A_{k+1}}$ and
\be\label{rhok+1rhoAk}
	\rho_{\Upsilon_{k+1}}' = \rho_{A_{k+1}} \otimes P_{B_{k+1}\cup R_{k+1}}(\{\emptyset\}).
\ee
We thus obtain 1.\ and 3.\ for $k+1$.

For 2., \eqref{rhok+1rhoAk} and \eqref{eq:woperator2} imply that
\begin{align}
 	\rho_{\Upsilon_{k+2}} 
	&= U_{\Upsilon_{k+1}}^{\Upsilon_{k+2}} \rho_{\Upsilon_{k+1}}' U_{\Upsilon_{k+2}}^{\Upsilon_{k+1}}\nonumber\\
	&= U_{\Upsilon_{k+1}}^{\Upsilon_{k+2}} \: \Bigl[\rho_{A_{k+1}} \otimes P_{B_{k+1}\cup R_{k+1}}(\{\emptyset\})\Bigr] \: U_{\Upsilon_{k+2}}^{\Upsilon_{k+1}}\nonumber\\
 	&=  \left( W_{A_{k+1}}^{A_{k+2} \cup B_{k+2}} \rho_{A_{k+1}} W_{A_{k+2} \cup B_{k+2}}^{A_{k+1}} \right) \otimes P_{R_{k+2}}(\{\emptyset\}),
\end{align}
which completes the induction step.
 
With regard to statement 4., the horizontal Born rule \eqref{eq:condborn}, 2., uniqueness of the vacuum \eqref{UV}, and Proposition~\ref{prop:WABC} together imply that
 \begin{align}
 	\Prob_{\sP}^\varepsilon(s_k | \rho_{\Upsilon_k}) &= \tr_{\Upsilon_k} \! \left( P_{\Upsilon_k}(\mbsk) \rho_{\Upsilon_k} \right)\nonumber\\
 	&= \tr_{A_k} \tr_{B_k} \left[ \left( I_{A_k} \otimes P_{B_k}(\nbsk) \right) \left( W_{A_{k-1}}^{A_k \cup B_k} \rho_{A_{k-1}} W_{A_k \cup B_k}^{A_{k-1}} \right) \right]\nonumber\\
 	&= \tr_{A_k} \tr_{C_k} \left[ \left( I_{A_k} \otimes \widetilde{P}_{C_k} \right) \left( U_{A_{k-1}}^{A_k \cup C_k} \rho_{A_{k-1}} U_{A_k \cup C_k}^{A_{k-1}} \right) \right]\nonumber\\
 	&= \mathcal{N}_{A_k} \,  \tr_{A_k}(\rho_{A_k}) = \mathcal{N}_{A_k}.
 \end{align}
\end{proof}

\begin{proof}[Proof of Proposition \ref{prop:expression}]
	By \eqref{totalprod} and statement 4.\ of Proposition~\ref{prop:properties},
	\be
		\Prob^{\rho_0,\varepsilon}_{\sP}(s) = \mathcal{N}_{A_\kappa} \, \mathcal{N}_{A_{\kappa+1}} \cdots \mathcal{N}_{A_K}.
	\ee
	We define the operators $\widetilde{\rho}_{A_k}$ in a similar way as $\rho_{A_k}$ but leave out the normalization factors, i.e.,
	\begin{align}
		\widetilde{\rho}_{A_{\kappa-1}} &=  
		\rho_{A_{\kappa-1}}, \nonumber\\
		\widetilde{\rho}_{A_{k+1}} 
		&= \tr_{C_{k+1}} \left(  [I_{A_{k+1}} \otimes \widetilde{P}_{C_{k+1}} ] \, 
		U_{A_k}^{A_{k+1} \cup C_{k+1}} \widetilde{\rho}_{A_k} 
		U_{A_{k+1} \cup C_{k+1}}^{A_k} \right).
		\label{tilderhoAkdef}
	\end{align}
	Since $\widetilde{\rho}_{A_{k+1}}$ contains just one factor $\widetilde{\rho}_{A_k}$, we obtain that
	\be
		\widetilde{\rho}_{A_k} = \mathcal{N}_{A_\kappa} \cdots \mathcal{N}_{A_k} \, \rho_{A_k}\,,
		\label{eq:rhorhotilde}
	\ee
	so that $\mathcal{N}_{A_\kappa} \, \mathcal{N}_{A_{\kappa+1}} \cdots \mathcal{N}_{A_k} = \tr_{A_k}(\widetilde{\rho}_{A_k})$ and
	\be
	\Prob^{\rho_0,\varepsilon}_{\sP}(s) = \tr_{A_K} (\widetilde{\rho}_{A_K})\,.
	\ee
	Now let
	\be
	S := \Sigma \setminus \past(\Upsilon_K) = \bigcup_{k=K+1}^\infty C_k
	\ee
	(a region without detectors).
	Then $S = \Gr(A_K,\Sigma)$ and $A_K = \Gr(S,\Upsilon_K)$ (see also Figure~\ref{fig:geom1d}), so $U_{A_K}^S$ exists, and
	\be
		\Prob^{\rho_0,\varepsilon}_{\sP}(s) = \tr_{S}(U_{A_K}^S\widetilde{\rho}_{A_K} U_{S}^{A_K}).
		\label{eq:traceformularho}
	\ee
	We now reorder the nested expressions in $\widetilde{\rho}_{A_k}$.

\bigskip

	\noindent\underline{Claim:} For all $k\geq \kappa$,
	\be
		\widetilde{\rho}_{A_k} =  \tr_{C_\kappa \cup \cdots \cup C_k} \biggl[  \bigl[ I_{A_k}\otimes \widetilde{P}_C^{(k)} \bigr] \,  \tr_{B_{\kappa-1}\cup R_{\kappa-1}} \Bigl(U_{\Upsilon_0}^{\Xi_{\kappa-1,k}} \, \rho_0 \, U_{\Xi_{\kappa-1,k}}^{\Upsilon_0} \Bigr) \biggr],
		\label{eq:resortingnestedtraces}
	\ee
	where
	\be
		\Xi_{\kappa-1,k} := B_{\kappa-1}\cup R_{\kappa-1} \cup C_\kappa \cup \cdots \cup C_k \cup A_k
	\ee
	is a Cauchy surface, and
	\be
		 \widetilde{P}_C^{(k)} := \bigotimes_{i=\kappa}^k \widetilde{P}_{C_i}.
	\ee
	
\bigskip
	
	\noindent \underline{Proof of the claim.} (Induction.) $k=\kappa:$ From \eqref{tilderhoAkdef} and \eqref{eq:rhoAkappa},
	\begin{align}
	\widetilde{\rho}_{A_\kappa} 
	&=  \tr_{C_{\kappa}} \biggl[  \bigl[ I_{A_{\kappa}} \otimes \widetilde{P}_{C_{\kappa}} \bigr] \, U_{A_{\kappa-1}}^{A_{\kappa} \cup C_{\kappa}} \:
	\tr_{B_{\kappa-1}\cup R_{\kappa-1}} \Bigl( U_{\Upsilon_0}^{\Upsilon_{\kappa-1}} \: \rho_0 \: U^{\Upsilon_0}_{\Upsilon_{\kappa-1}}  \Bigr)
	U_{A_{\kappa} \cup C_{\kappa}}^{A_{\kappa-1}} \biggr]  \nonumber\\
	&= \tr_{C_\kappa} \biggl[  \bigl[ I_{A_\kappa}\otimes \widetilde{P}_{C_\kappa} \bigr] \,  \tr_{B_{\kappa-1}\cup R_{\kappa-1}} \Bigl(U_{\Upsilon_0}^{\Xi_{\kappa-1,\kappa}} \, \rho_0 \, U_{\Xi_{\kappa-1,\kappa}}^{\Upsilon_0} \Bigr) \biggr]
	\end{align}
	as desired because $U_{\Upsilon_{\kappa-1}}^{\Xi_{\kappa-1,\kappa}}= I_{B_{\kappa-1}\cup R_{\kappa-1}} \otimes U_{A_{\kappa-1}}^{A_\kappa\cup C_\kappa}$ by (IL).

	For the induction step, let \eqref{eq:resortingnestedtraces} be true for some $k$. Then
	\begin{align}
		\widetilde{\rho}_{A_{k+1}} &= \tr_{C_{k+1}} \biggl[  \bigl[I_{A_{k+1}} \otimes \widetilde{P}_{C_{k+1}} \bigr] \, U_{A_k}^{A_{k+1} \cup C_{k+1}} \widetilde{\rho}_{A_k} U_{A_{k+1} \cup C_{k+1}}^{A_k} \biggr]\nonumber\\
 &= \tr_{C_{k+1}} \biggl[  \bigl[ I_{A_{k+1}} \otimes \widetilde{P}_{C_{k+1}} \bigr] \, U_{A_k}^{A_{k+1} \cup C_{k+1}} \:\times \nonumber\\
 &~~~\times \: \tr_{C_\kappa \cup \cdots \cup C_k} \biggl( \bigl[ I_{A_k}\otimes \widetilde{P}_C^{(k)} \bigr] \, \tr_{B_{\kappa-1}\cup R_{\kappa-1}} \Bigl(U_{\Upsilon_0}^{\Xi_{\kappa-1,k}} \, \rho_0 \, U_{\Xi_{\kappa-1,k}}^{\Upsilon_0} \Bigr) \biggr) U_{A_{k+1} \cup C_{k+1}}^{A_k} \biggr] \nonumber\\
 &= \tr_{C_\kappa \cup \cdots \cup C_{k+1}} \Biggl[  \Bigl[ I_{A_{k+1}} \otimes \widetilde{P}_{C_{k+1}} \otimes  I_{C_\kappa\cup \cdots \cup C_k} \Bigr] \, \Bigl[ U_{A_k}^{A_{k+1} \cup C_{k+1}} \otimes  I_{C_\kappa \cup \cdots \cup C_k} \Bigr] \, \Bigl[  I_{A_k}\otimes \widetilde{P}_C^{(k)} \Bigr] \:\times \nonumber\\
 &~~~\times \:  \tr_{B_{\kappa-1}\cup R_{\kappa-1}} \biggl( U_{\Upsilon_0}^{\Xi_{\kappa-1,k}} \, \rho_0 \, U_{\Xi_{\kappa-1,k}}^{\Upsilon_0} \biggr) \Bigl[  U_{A_{k+1} \cup C_{k+1}}^{A_k} \otimes I_{C_\kappa \cup \dots \cup C_k}  \Bigr] \Biggr] \nonumber\\
 &= \tr_{C_\kappa \cup \cdots \cup C_{k+1}} \Biggl[  \Bigl[ I_{A_{k+1}} \otimes \widetilde{P}_{C_{k+1}} \otimes  I_{C_\kappa\cup \cdots \cup C_k} \Bigr] \,  \Bigl[  I_{A_{k+1}\cup C_{k+1}}\otimes \widetilde{P}_C^{(k)} \Bigr]  \, \Bigl[ U_{A_k}^{A_{k+1} \cup C_{k+1}} \otimes  I_{C_\kappa \cup \cdots \cup C_k} \Bigr] \:\times \nonumber\\
 &~~~\times \:  \tr_{B_{\kappa-1}\cup R_{\kappa-1}} \biggl( U_{\Upsilon_0}^{\Xi_{\kappa-1,k}} \, \rho_0 \, U_{\Xi_{\kappa-1,k}}^{\Upsilon_0} \biggr) \Bigl[  U_{A_{k+1} \cup C_{k+1}}^{A_k} \otimes I_{C_\kappa \cup \dots \cup C_k}  \Bigr] \Biggr] \nonumber\\
 &= \tr_{C_\kappa \cup \cdots \cup C_{k+1}} \Biggl[  \Bigl[ I_{A_{k+1}} \otimes \widetilde{P}_{C_{k+1}} \otimes \widetilde{P}_C^{(k)} \Bigr] \, \tr_{B_{\kappa-1}\cup R_{\kappa-1}} \biggl( U_{\Upsilon_0}^{\Xi_{\kappa-1,k+1}} \, \rho_0 \, U_{\Xi_{\kappa-1,k+1}}^{\Upsilon_0} \biggr)  \Biggr]\,.
	\end{align}
 	This proves the claim.$\hfill\square$

\bigskip	
 	
 	We now turn to the proof of \eqref{eq:expression}. According to \eqref{eq:traceformularho} and \eqref{eq:resortingnestedtraces}, we have that
 	\begin{align}
 		\Prob^{\rho_0,\varepsilon}_{\sP}(s) &= \tr_{S}\left( U_{A_K}^S \tr_{C_\kappa \cup \cdots \cup C_K} \biggl[ \bigl[ I_{A_K}\otimes \widetilde{P}_C^{(K)} \bigr] \, \tr_{B_{\kappa-1}\cup R_{\kappa-1}} \Bigl( U_{\Upsilon_0}^{\Xi_{\kappa-1,K}} \, \rho_0 \, U_{\Xi_{\kappa-1,K}}^{\Upsilon_0} \Bigr) \biggr] U_{S}^{A_K} \right)\nonumber\\
 		&= \tr_{C_\kappa \cup \cdots \cup C_K \cup S} \biggl[ \bigl[ I_{S}\otimes \widetilde{P}_C^{(K)} \bigr] \,\tr_{B_{\kappa-1}\cup R_{\kappa-1}} \Bigl( U_{\Upsilon_0}^{\Xi_{\kappa-1,\infty}} \, \rho_0 \, U_{\Xi_{\kappa-1,\infty}}^{\Upsilon_0} \Bigr) \biggr] \nonumber\\[2mm]
 		&= \tr_{\Sigma} \left(P(s) \, U_{\Upsilon_0}^\Sigma \rho_0 U_\Sigma^{\Upsilon_0} \right)
 		\label{eq:expressioninitialdensitymatrix}
 	\end{align}
	with
	\be
	\Xi_{\kappa-1,\infty} := B_{\kappa-1} \cup R_{\kappa-1} \cup C_\kappa \cup \cdots \cup C_K \cup S\,.
	\ee
 	Taking $\rho_0 = |\psi_0\rangle \langle \psi_0|$, we finally obtain:
 	\begin{align}
 		\Prob^{\rho_0,\varepsilon}_{\sP}(s) =  \tr_{\Sigma} \left(P(s) \, U_{\Upsilon_0}^\Sigma |\psi_0\rangle \langle \psi_0| U_\Sigma^{\Upsilon_0} \right)
 		&= \scp{\psi_\Sigma}{P(s) |\psi_\Sigma}.
 	\end{align}
\end{proof}

\subsubsection{Bounds for $\Prob^{\psi,\varepsilon}_{\det,\sP}(L)$} \label{sec:bounds}

Since the collapses are represented by operators $P(s)$ that arise from projections to configuration sets on $B_k$ through transport to $C_k$, the idea is to bound them from below (and above) by projections to configuration sets on $C_k$ obtained by shrinking (or growing, respectively) $B_{k\ell}$. 

\begin{defn}[grown and shrunk configuration space sets]
We first introduce (see Figure~\ref{fig:def_c_grown_shrunk}):
\begin{align}
	\cklcheck &= \Sr(B_{k\ell},\Sigma) \subseteq C_k,\nonumber\\
	\cklhat &= \past(B_{k\ell}) \cap C_k.
\end{align}

Next, we define:
\begin{align}
	\mcskl &=  \left\{ \begin{array}{l} \emptyset(\ckl) ~{\rm if}~s_{k\ell} = 0,\\
							\exists(\ckl)~{\rm if}~s_{k\ell} =1, \end{array} \right. \nonumber\\
	\mcsklcheck &= \left\{ \begin{array}{l} \emptyset(\cklhat) ~{\rm if}~s_{k\ell} = 0,\\
							\exists(\cklcheck)~{\rm if}~s_{k\ell} =1, \end{array} \right.\nonumber\\
	\mcsklhat &= \left\{ \begin{array}{l} \emptyset(\cklcheck) ~{\rm if}~s_{k\ell} = 0,\\
							\exists(\cklhat)~{\rm if}~s_{k\ell} =1. \end{array} \right.	
	\label{eq:mcskldef}
\end{align}
Furthermore:
\begin{align}
	\mcs &= \bigcap_{k = \kappa}^K \bigcap_{\ell = 1}^r \mcskl,\nonumber\\
	\mcscheck &= \bigcap_{k = \kappa}^K \bigcap_{\ell = 1}^r \mcsklcheck,\nonumber\\
	\mcshat &= \bigcap_{k = \kappa}^K \bigcap_{\ell = 1}^r \mcsklhat.
	\label{eq:mscdef}
\end{align}
We then have that $\cklcheck \subseteq \ckl \subseteq \cklhat$. As $\exists(R_1) \subseteq \exists(R_2)$ and $\emptyset(R_2) \subseteq \emptyset(R_1)$ for $R_1 \subseteq R_2 \subseteq \Sigma$, it follows that $\mcsklcheck \subseteq \mcskl \subseteq \mcsklhat$ and $\mcscheck \subseteq \mcs \subseteq \mcshat$.

\begin{figure}[t]
\centering
 \includegraphics[width=0.6\textwidth]{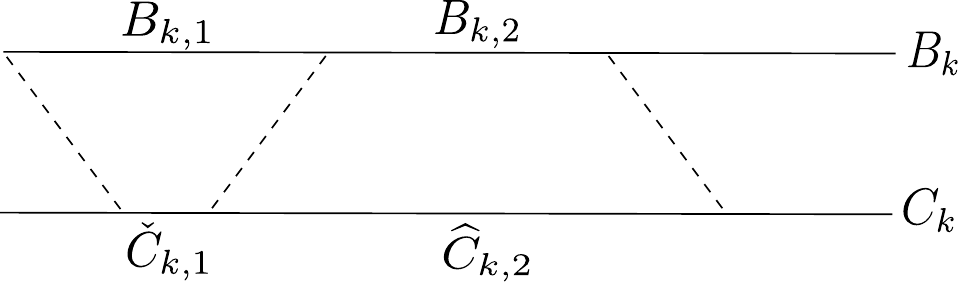}
 \caption{Definition of the sets $\cklcheck,\cklhat$. The surfaces $B_k$ and $C_k$ are viewed from the side.}
 \label{fig:def_c_grown_shrunk}
\end{figure}

Moreover, we define:
\begin{align}
	\mslepsilon &= \bigcup_{s:L} \mcs,\nonumber\\
	\mslcheck &= \bigcup_{s:L} \mcscheck,\nonumber\\
	\mslhat &= \bigcup_{s:L} \mcshat,
	\label{eq:msldef}
\end{align}
where the clause ``$s:L$'' means all $s \in \{ 0,1\}^{(K-\kappa+1) r}$ compatible with $L$, i.e., such that for all $1 \leq \ell \leq r$: $s_{k\ell} = 0$ for all $k$ if $L_\ell = 0$ and $\exists \, k$ with $s_{k\ell} = 1$ if  $L_\ell = 1$. The superscript $\varepsilon$ has been introduced to remind the reader that the introduced sets are $\varepsilon$-dependent. It follows that $\mslcheck \subseteq \mslepsilon \subseteq \mslhat$.
\end{defn}

\begin{prop} \label{prop:bounds}
	For all $\varepsilon > 0$ the following inequalities hold:
	\be
		\langle \psi_{\Sigma} |P_\Sigma(\mslcheck) | \psi_{\Sigma} \rangle 
		\leq \Prob^{\psi,\varepsilon}_{\det,\sP}(L) 
		\leq \langle \psi_{\Sigma} |P_\Sigma(\mslhat) | \psi_{\Sigma} \rangle.
		\label{eq:probbounds}
	\ee
\end{prop}

Before proving Proposition \ref{prop:bounds}, we establish some auxiliary statements.

\begin{prop}
	Let $D_{k\ell} = \Gr(B_{k\ell},\Sigma) \backslash C_k$, so that $\cklhat \cup D_{k\ell} = \Gr(B_{k\ell},\Sigma)$. Then the following operator inequalities\footnote{The operator inequalities are meant as follows. Let $P,Q : \Hilbert \rightarrow \Hilbert$. Then $P \leq Q  :\Leftrightarrow \langle \phi | P | \phi \rangle \leq \langle \phi | Q |\phi \rangle~\forall \phi \in \Hilbert$. Put differently, $Q-P$ is a positive operator. When $P,Q$ are projections, this is equivalent to $\range(P) \subseteq \range(Q)$.} hold:
	\begin{align}
		P_{\Sigma}(\emptyset(\cklhat \cup D_{k\ell})) &\leq U_{\Upsilon_k}^\Sigma P_{\Upsilon_k}(\emptyset_{\Upsilon_k}(B_{k\ell})) U_\Sigma^{\Upsilon_k} \leq P_{\Sigma}(\emptyset(\cklcheck)) \label{eq:projectorineq1},\\
		P_{\Sigma}(\exists(\cklcheck)) &\leq U_{\Upsilon_k}^\Sigma P_{\Upsilon_k}(\exists_{\Upsilon_k}(B_{k\ell})) U_\Sigma^{\Upsilon_k} \leq P_{\Sigma}(\exists(\cklhat \cup D_{k\ell})) \label{eq:projectorineq2}.
	\end{align}
\end{prop}

\begin{proof}
	We verify \eqref{eq:projectorineq1}; \eqref{eq:projectorineq2} then follows via $P \leq Q \Leftrightarrow I - Q \leq I - P$.
	
	To prove \eqref{eq:projectorineq1}, first note that (PL) can be equivalently formulated as follows (see also Figure~\ref{fig:grown_shrunk_sets}). For any $R \subseteq \Sigma$,
	\be
		U_\Sigma^{\Sigma'} P_{\Sigma}(\emptyset_\Sigma(R)) U_{\Sigma'}^\Sigma \leq P_{\Sigma'}(\emptyset_{\Sigma'}(\Sr(R,\Sigma'))). 
		\label{eq:fsreform}
	\ee
	Applying this for $R \leftrightarrow B_{k\ell}$, $\Sigma \leftrightarrow \Upsilon_k$ and $\Sigma' \leftrightarrow \Sigma$, we obtain the second inequality in \eqref{eq:projectorineq1} as $\cklcheck = \Sr(B_{k\ell},\Sigma)$. For the first inequality in \eqref{eq:projectorineq1}, we bring the evolution operators in \eqref{eq:fsreform} to the other side and apply the resulting inequality for $\Sigma$, $\Sigma' \leftrightarrow \Upsilon_k$, $R \leftrightarrow \cklhat \cup D_{k\ell}$ and make use of $\Sr(\cklhat \cup D_{k\ell},\Upsilon_k) = B_{k\ell}$.
\end{proof}

\begin{prop}\label{prop:minorante}
	\be
		P_\Sigma(\mcscheck) \leq P(s) \leq P_\Sigma(\mcshat).
		\label{eq:psinequality}
	\ee
\end{prop}

\begin{proof}
	Recall \eqref{eq:ps}, i.e.:
	\be
		P(s) = \left(\bigotimes_{k=\kappa}^K W_{B_k}^{C_k} \: P_{B_k}(\nbsk) \: W_{C_k}^{B_k} \right) \otimes I_{\Sigma \backslash C}.
	\ee
	Let
	\be
		\nbskl = \left\{ \begin{array}{l} \emptyset_{B_k}(B_{k\ell})~{\rm if}~s_{k\ell} = 0,\\ \exists_{B_k}(B_{k\ell})~{\rm if}~s_{k\ell} = 1. \end{array} \right.
	\ee
	Then $\displaystyle\nbsk = \bigcap_{\ell = 1}^r \nbskl$ and hence
	\be
		P_{B_k}(\nbsk) = \prod_{\ell=1}^r P_{B_k}(\nbskl).
	\ee
	
	\noindent \underline{Claim:}
	\be
		\prod_{\ell = 1}^r P_{C_k}(\ncsklcheck) \leq W_{B_k}^{C_k} \prod_{\ell=1}^r P_{B_k}(\nbskl) W_{C_k}^{B_k} \leq \prod_{\ell = 1}^r P_{C_k}(\ncsklhat),
		\label{eq:ineqclaim}
	\ee
	where
	\begin{align}
		\ncsklcheck &=  \left\{ \begin{array}{l} \emptyset_{C_k}(\cklhat)~{\rm if}~s_{k\ell} = 0,\\ \exists_{C_k}(\cklcheck)~{\rm if}~s_{k\ell} = 1, \end{array} \right.\nonumber\\
		\ncsklhat &=  \left\{ \begin{array}{l} \emptyset_{C_k}(\cklcheck)~{\rm if}~s_{k\ell} = 0,\\ \exists_{C_k}(\cklhat)~{\rm if}~s_{k\ell} = 1. \end{array} \right.
	\end{align}

\noindent \underline{Proof of the claim.} Focus on the first inequality in \eqref{eq:ineqclaim}, i.e.,
	\begin{align}
		\left\langle \chi_{C_k} \left\vert \prod_{\ell = 1}^r P_{C_k}(\ncsklcheck) \right\vert \chi_{C_k} \right\rangle   \leq   \left\langle \chi_{C_k} \left\vert W_{B_k}^{C_k} \prod_{\ell=1}^r P_{B_k}(\nbskl) W_{C_k}^{B_k} \right\vert \chi_{C_k} \right\rangle ~ \forall \chi_{C_k} \in \Hilbert_{C_k}.
	\end{align}
	The idea is to rewrite the right-hand side in terms of the $U$ instead of $W$ operators to be able to use the inequalities \eqref{eq:projectorineq1}, \eqref{eq:projectorineq2} after some further manipulations. Recall that
	\be
		U_{\Sigma}^{\Upsilon_k} | \chi_{C_k} \rangle \otimes \vacket{\Sigma \backslash C_k} = \left( W_{C_k}^{B_k}  | \chi_{C_k} \rangle \right) \otimes \vacket{\Upsilon_k \backslash B_k}.
	\ee
	Therefore, using $P_{\Upsilon_k}(\mbskl) = P_{B_k}(\nbskl) \otimes I_{\Upsilon_k \backslash B_k}$, we find that
	\begin{align}
		&\left\langle \chi_{C_k} \left\vert W_{B_k}^{C_k} \prod_{\ell=1}^r P_{B_k}(\nbskl) W_{C_k}^{B_k}  \right\vert \chi_{C_k} \right\rangle \nonumber\\
		= &\left\langle \chi_{C_k} \otimes \emptyset(\Sigma \backslash C_k) \left\vert U_{\Upsilon_k}^{\Sigma} \left( \prod_{\ell=1}^r P_{\Upsilon_k}(\mbskl) \right) U_{\Sigma}^{\Upsilon_k} \right\vert \chi_{C_k} \otimes \emptyset(\Sigma \backslash C_k) \right\rangle\nonumber\\
		= &\left\langle \chi_{C_k} \otimes \emptyset(\Sigma \backslash C_k) \left\vert  \prod_{\ell=1}^r \left( U_{\Upsilon_k}^{\Sigma}  P_{\Upsilon_k}(\mbskl) U_{\Sigma}^{\Upsilon_k} \right) \right\vert \chi_{C_k} \otimes \emptyset(\Sigma \backslash C_k) \right\rangle.
		\label{eq:ineqcalc1}
	\end{align}
	Now, for the expressions $U_{\Upsilon_k}^{\Sigma}  P_\Sigma(\mbskl) U_{\Sigma}^{\Upsilon_k}$ we use the previously obtained inequalities \eqref{eq:projectorineq1}, \eqref{eq:projectorineq2} (the respective lower bounds). This yields:
	\begin{align}
		&\eqref{eq:ineqcalc1} \geq \left\langle \chi_{C_k} \otimes \emptyset(\Sigma \backslash C_k) \left\vert  \prod_{\ell=1}^r P_\Sigma(\mcdsklcheck) \right\vert \chi_{C_k} \otimes \emptyset(\Sigma \backslash C_k) \right\rangle,
		\label{eq:ineqcalc2}
	\end{align}
	where
	\be
		\mcdsklcheck = \left\{ \begin{array}{ll} \emptyset(\cklhat \cup D_{k\ell})&{\rm if}~s_{k\ell} = 0,\\ \exists(\cklcheck)&{\rm if}~s_{k\ell} = 1. \end{array} \right.
	\ee
	Next, note that because of $D_{k\ell} \subseteq \Sigma \backslash C_k$,
	 \be
		P_\Sigma(\emptyset(\cklhat \cup D_{k\ell})) \, | \chi_{C_k} \! \otimes \emptyset(\Sigma \backslash C_k) \rangle = P_\Sigma(\emptyset(\cklhat)) \, | \chi_{C_k} \! \otimes \emptyset(\Sigma \backslash C_k) \rangle,
	\ee
	and therefore
	\be
		P_\Sigma(\mcdsklcheck) \, | \chi_{C_k} \! \otimes \emptyset(\Sigma \backslash C_k) \rangle = P_\Sigma(\mcsklcheck) \, | \chi_{C_k} \! \otimes \emptyset(\Sigma \backslash C_k) \rangle.
	\ee
	As a consequence, since $P_\Sigma(\mcsklcheck)= P_{C_k} (\ncsklcheck)\otimes I_{\Sigma\setminus C_k}$,
	\begin{align}
		&\left\langle \chi_{C_k} \left\vert W_{B_k}^{C_k} \prod_{\ell=1}^r P_{B_k}(\nbskl) W_{C_k}^{B_k} \right\vert \chi_{C_k} \right\rangle\nonumber\\
		&   \geq   \left\langle \chi_{C_k} \! \otimes \emptyset(\Sigma \backslash C_k) \left\vert  \prod_{\ell=1}^r P_\Sigma(\mcsklcheck) \right\vert \chi_{C_k} \! \otimes \emptyset(\Sigma \backslash C_k) \right\rangle\nonumber\\
		&= \left\langle \chi_{C_k} \left\vert  \prod_{l=1}^L P_{C_k}(\ncsklcheck) \right\vert \chi_{C_k} \right\rangle
	\end{align}
	for all $\chi_{C_k} \in \Hilbert_{C_k}$. 
	This yields the first inequality of the claim \eqref{eq:ineqclaim}. For the second one, one proceeds analogously.$\hfill\square$

\bigskip
	
We now continue the proof of Proposition~\ref{prop:minorante}. Eq.~\eqref{eq:ineqclaim} implies, by taking a tensor product over $k$, that
	\be
		\left[ \bigotimes_{k=\kappa}^K \left( \prod_{\ell = 1}^r P_{C_k}(\ncsklcheck) \right) \right] \otimes I_{\Sigma \backslash C} \leq P(s) \leq \left[ \bigotimes_{k=\kappa}^K \left( \prod_{\ell = 1}^r P_{C_k}(\ncsklhat) \right) \right] \otimes I_{\Sigma \backslash C}.
	\ee
    Now the operator on the left-hand side equals $P_\Sigma(\mcscheck)$ (as can be seen by comparing the action of these operators on a state in the tensor product) and the one on the right-hand side $P_\Sigma(\mcshat)$ so, in fact, we obtain \eqref{eq:psinequality}.
\end{proof}

\medskip

\begin{proof}[Proof of Proposition \ref{prop:bounds}.] We prove the lower bound of \eqref{eq:probbounds} first, then the upper.

\medskip

\noindent \textit{Lower bound.} Applying \eqref{eq:psinequality} to the individual summands in $\displaystyle \Prob^{\psi,\varepsilon}_{\det,\sP}(L) = \sum_{s:L}\langle \psi_\Sigma | P(s) | \psi_\Sigma \rangle$, we obtain that
\be
	 \sum_{s:L} \langle \psi_\Sigma | P_{\Sigma}(\mcscheck) | \psi_\Sigma \rangle \leq \Prob^{\psi,\varepsilon}_{\det,\sP}(L).
	 \label{eq:boundscalc1}
\ee
Now, by the definition of the sets $\mcscheck$, we have that $\mcscheck \cap \check{M}_C(s') = \emptyset$ for $s \neq s'$. Thus,
\be
	\sum_{s:L} \langle \psi_\Sigma | P_{\Sigma}(\mcscheck) | \psi_\Sigma \rangle = \langle \psi_\Sigma | P_\Sigma \Big( \bigcup_{s:L} \mcscheck \Big) | \psi_\Sigma \rangle =  \langle \psi_\Sigma | P_\Sigma(\mslcheck)| \psi_\Sigma \rangle,
\ee
which, together with \eqref{eq:boundscalc1}, yields the lower bound of \eqref{eq:probbounds}.\\

\noindent \textit{Upper bound.} Here, we cannot proceed analogously as the sets $\mcshat$ are not mutually disjoint. Instead, we use \eqref{eq:psinequality} to show that
\be
	\sum_{s:L} P(s) \leq P_{\Sigma}(\mslhat),
	\label{eq:upperboundproj}
\ee
which will then directly yield the upper bound of \eqref{eq:probbounds}.

 	For notational convenience, we write
 	\be
 		P(s) = P_C(s) \otimes I_{\Sigma \backslash C},
 	\ee
 	where
 	\begin{align}
 		&P_C(s) : \bigotimes_{k=\kappa}^K \Hilbert_{C_k} \rightarrow  \bigotimes_{k=\kappa}^K \Hilbert_{C_k},\nonumber\\[2mm]
 		&P_C(s) = \sW_B^C \: P_B(s) \: \sW_C^B
 	\end{align}
 	with $\displaystyle \sW_C^B = \bigotimes_{k=\kappa}^K W_{C_k}^{B_k}$, $\sW_B^C = (\sW_C^B)^\dagger$, and $\displaystyle P_B(s) = \bigotimes_{k=\kappa}^K P_{B_k}(\nbsk)$. Then
 	\be
 		\sum_{s:L} P_C(s) = \mathscr{W}_B^C \Big( \sum_{s:L} P_B(s) \Big) \mathscr{W}_C^B.
 	\ee
While the $P_C(s)$ are in general not projections, the $P_B(s)$ are, and they are mutually orthogonal for different $s$ because the $P_B(s)$ act on $\Hilbert_B := \bigotimes_{k=\kappa}^K \Hilbert_{B_k}$, and the $\nbsk$ are mutually disjoint sets in $\Gamma(B_k)$. As a consequence, $\sum_{s:L} P_B(s)$ is again a projection. For it we want to derive an upper bound from the upper bounds \eqref{eq:psinequality} that we have for each $s$. 

We make the following simple observation. Let $\Hilbert_i$ and $\Kilbert_i$ be closed subspaces of $\Hilbert$ for each $i=1,\ldots,n$ such that $\Hilbert_i\subseteq \Kilbert_i$, and the $\Hilbert_i$ are mutually orthogonal (while the $\Kilbert_i$ are not necessarily). Then $\bigoplus_i \Hilbert_i \subseteq \bigvee_i \Kilbert_i$, where $\bigvee_i \Kilbert_i$ denotes the span of the $\Kilbert_i$. The corresponding statement for projections reads as follows: If $P_i$ and $Q_i$, $i=1,\ldots,n$, are projections in $\Hilbert$, if the $P_i$ are mutually orthogonal, and if $P_i \leq Q_i$, then
\be
\sum_i P_i \leq \bigvee_i Q_i\,,
 	 	\label{eq:sumpleqq}
\ee
where $\bigvee_i Q_i$ means the projection onto $\bigvee_i \Kilbert_i$. 
Later, we will use that $\bigvee_i P_\Sigma(S_i)= P_\Sigma(\bigcup_i S_i)$ for any sets $S_i\subseteq \Gamma(\Sigma)$.

In our case, $i \leftrightarrow s$, $P_i \leftrightarrow P_B(s)$, and $Q_i$ still needs to be chosen. Since $P_B(s)$ is closely related to $P_C(s)$, and since from \eqref{eq:psinequality} we have that $P_C(s)\leq P_C(\ncshat)$ with $\ncshat \subseteq \Gamma(C)$ the set defined by $\mcshat = \ncshat \times \Gamma(\Sigma\backslash C)$, our $Q_i$ should be closely related to $P_C(\ncshat)$. The latter is indeed a projection, but lives on $\Hilbert_C$ instead of $\Hilbert_B$. So we need to transport it to $\Hilbert_B$ via $\sW_C^B$, which yields
\be\label{PhatBsdef}
\sW_C^B\, P_C(s) \, \sW^C_B ~\leq~ \sW_C^B\, P_C(\ncshat) \, \sW^C_B ~=:~ \widehat{P}_B(s)
\ee
or, since $P_C(s)= \sW^C_B \, P_B(s) \, \sW^B_C$,
\be
\sW_C^B\, \sW^C_B \, P_B(s) \, \sW^B_C \, \sW^C_B ~~\leq~~ \widehat{P}_B(s)\,.
\ee
Now $\sW_C^B$ is unitary only to a \emph{subspace} of $\Hilbert_B$ (that may not contain $\range P_B(s)$); let $Q$ denote the projection onto that subspace, so $\sW_B^C\, \sW_C^B=I_C$ and $\sW_C^B \, \sW_B^C = Q$. Thus,
\be
Q\, P_B(s) \, Q ~\leq~ \widehat{P}_B(s)\,,
\ee
and we want to obtain from that an upper bound $Q_i$ for $P_B(s)$.

\bigskip

\noindent\underline{Claim.} If $P,Q,\widehat{P}$ are projections in $\Hilbert$ such that 
\be\label{QPQPQ}
QPQ \leq \widehat{P} \leq Q\,,
\ee
then
\be\label{PPIQ}
P \leq \widehat{P} + (I-Q)\,.
\ee

\bigskip

\noindent\underline{Proof of the claim.}
Note first that $\widehat{P}\leq Q$ implies that $\widehat{P}$ and $I-Q$ are mutually orthogonal, so $\widehat{P}+(I-Q)$ is again a projection. Now suppose for a moment that $P$ has rank 1, $P=|\psi\rangle \langle \psi|$. Then \eqref{QPQPQ} means that $Q\psi \in \range \widehat{P} \subseteq \range Q$. So, $\psi$ can be split, $\psi=Q\psi + (I-Q)\psi$, into a part that lies in $\range \widehat{P}$ and a part orthogonal to $\range Q$. Hence, $\psi$ lies in the range of $\widehat{P}+(I-Q)$, and since the latter is a projection, \eqref{PPIQ} follows. 

For the general case, let $\{\psi_i\}$ be an orthonormal basis of $\range P$, so $P=\sum_i |\psi_i\rangle \langle \psi_i|$. Then for every $i$, $Q|\psi_i\rangle \langle \psi_i|Q \leq QPQ \leq \widehat{P}$, so $Q\psi_i\in \range \widehat{P}$ and, as before, $\psi_i \in \range(\widehat{P}+(I-Q))$. Using again that $\widehat{P}+(I-Q)$ is a projection, \eqref{PPIQ} follows.$\hfill\square$
	
\bigskip
 
We continue the proof of Proposition~\ref{prop:bounds}. In our case, $\widehat{P}_B(s)$ is indeed $\leq Q$ because by its definition \eqref{PhatBsdef} it is a projection onto a subspace of $\range Q=\range \sW_C^B$. Thus, \eqref{PPIQ} applies, 
\be
P_B(s) ~\leq~ \widehat{P}_B(s) + (I-Q) ~=:~ Q_i
\ee
with $I=I_B$, and \eqref{eq:sumpleqq} entails that
\begin{align}
\sum_{s:L} P_B(s) 
&\leq \bigvee_{s:L} \Bigl[ \widehat{P}_B(s) + (I-Q) \Bigr] \nonumber\\
&= \Biggl[ \bigvee_{s:L} \widehat{P}_B(s) \Biggr] + (I-Q) \nonumber\\
&= \Biggl[ \bigvee_{s:L} \sW_C^B\, P_C(\ncshat)\, \sW_B^C \Biggr] + (I-Q) \nonumber\\
&= \sW_C^B \Biggl[ \bigvee_{s:L}  P_C(\ncshat) \Biggr]  \sW_B^C + (I-Q) \nonumber\\
&= \sW_C^B \:  P_C\biggl(\bigcup_{s:L}\ncshat \biggr)\:  \sW_B^C + (I-Q)\,.
\end{align}
 
	As a consequence, using $Q_B = \mathscr{W}_C^B \mathscr{W}_B^C$, $I_C = \mathscr{W}_B^C \mathscr{W}_C^B$, and $\displaystyle \mslhat = \big(\bigcup_{s:L} \ncshat \big) \times \Gamma(\Sigma \backslash C)$,
	\begin{align}
		\sum_{s : L} P(s) &=  \Big[ \mathscr{W}_B^C \Big( \sum_{s : L} P_B(s) \Big) \mathscr{W}_C^B \Big] \otimes I_{\Sigma \backslash C} \nonumber\\
		&\leq  \Big[ \mathscr{W}_B^C \Big( \mathscr{W}_C^B P_C \Big(\bigcup_{s:L} \ncshat \Big) \mathscr{W}_B^C  + (I - Q) \Big) \mathscr{W}_C^B \Big] \otimes I_{\Sigma \backslash C}  \nonumber\\
		&= P_C \Big(\bigcup_{s:L} \ncshat \Big) \otimes I_{\Sigma \backslash C} \nonumber\\
		&= P_\Sigma(\mslhat)\,.
	\end{align}
	This ends the proof.
\end{proof}

\subsubsection{Convergence of the Bounds in the Limit $\varepsilon \rightarrow 0$} \label{sec:limit}

In order to take the limit $\varepsilon \rightarrow 0$, we first introduce new sets in configuration space which lead to looser bounds for $\Prob^{\psi,\varepsilon}_{\det,\sP}(L)$ that are, however, still tight enough for our purposes and easier to deal with. These sets are defined in a simpler way than the previous sets $\mslcheck, \mslhat$ and are conveniently nested within each other for two different values of $\varepsilon$. 
	Let $\mathcal{B}_\varepsilon(\vx) \subset \RRR^3$ denote the ball with radius $\varepsilon$ around $\vx \in \RRR^3$, recall that $\pi(x^0,\vx) = \vx$, and let $\pi_\Sigma$ be the restriction of $\pi$ to $\Sigma$ (which is a homeomorphism).

\begin{defn}
	For all $\varepsilon > 0$, we define the following shrunk and grown versions of the patches $\patch_\ell$ (see Figure~\ref{fig:grown_shrunk_patches}):
	\begin{align}
		\plcheck &= \bigl\{ x \in \patch_\ell : \mathcal{B}_\varepsilon(\pi(x)) \subseteq \pi(\patch_\ell) \bigr\},\nonumber\\
		\plhat &= \bigcup_{x \in \patch_\ell} \pi_\Sigma^{-1}\Bigl( \mathcal{B}_\varepsilon(\pi(x)) \Bigr), \nonumber\\
		\dpepsilon &= \bigcup_{\ell=1}^r \bigl( \plhat \setminus \plcheck \bigr).
	\end{align}
	For all $\varepsilon > 0$, we have that $\plcheck \subseteq \patch_\ell \subseteq \plhat$, as well as 
	\be \label{monotonepatch} 
	 \check{\patch}_\ell^{\varepsilon_2} \subseteq \check{\patch}_\ell^{\varepsilon_1}\,, \qquad
	 \widehat{\patch}_\ell^{\varepsilon_1} \subseteq \widehat{\patch}_\ell^{\varepsilon_2}\,, \qquad
	 \partial\patch^{\varepsilon_1}\subseteq \partial\patch^{\varepsilon_2}
	\ee
	for all $0 < \varepsilon_1 < \varepsilon_2$.
	Furthermore, let
	\begin{align}
		\mpllcheck &= \left\{ \begin{array}{l} \emptyset(\plhat) ~{\rm if}~L_\ell = 0,\\
							\exists(\plcheck)~{\rm if}~L_\ell =1, \end{array} \right. \nonumber\\
		\mpllhat &= \left\{ \begin{array}{l} \emptyset(\plcheck) ~{\rm if}~L_\ell = 0,\\
							\exists(\plhat)~{\rm if}~L_\ell =1. \end{array} \right.
	\end{align}
	In addition, we introduce
	\begin{align}
		\mplcheck &= \bigcap_{\ell = 1}^r \mpllcheck, \nonumber\\
		\mplhat &= \bigcap_{\ell = 1}^r \mpllhat.
	\end{align}
	Then, the relations $\emptyset(R_1) \supseteq \emptyset(R_2)$, $\exists(R_1) \subseteq \exists(R_2)$ for $R_1 \subseteq R_2 \subseteq \Sigma$ imply that, for all $\varepsilon >0$,
	\be
		\mplcheck \subseteq \msl \subseteq \mplhat,
		\label{eq:mpinclusions1}
	\ee
	and that, for all $0 < \varepsilon_1 < \varepsilon_2$,
	\be
		\check{M}_\patch^{\varepsilon_2}(L) \subseteq \check{M}_\patch^{\varepsilon_1}(L), ~~~~~		\widehat{M}_\patch^{\varepsilon_1}(L) \subseteq \widehat{M}_\patch^{\varepsilon_2}(L). \label{eq:mpinclusions2}
	\ee
\end{defn}

\begin{figure}[tp]
\centering
 \includegraphics[width=0.6\textwidth]{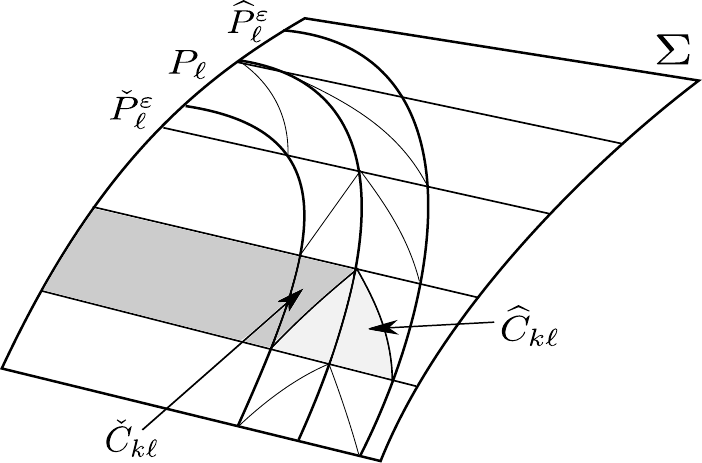}
 \caption{Illustration of the shrunk and grown patches $\plcheck, \plhat$ and their relation to the shrunk and grown sets $\check{C}_{k\ell}$ (dark grey region) and $\widehat{C}_{k\ell}$ (dark and light grey regions).}
 \label{fig:grown_shrunk_patches}
\end{figure}

\begin{prop}\label{prop:PlCkl}
	 $\displaystyle \plcheck \subseteq \bigcup_{k\in\ZZZ} \check{C}_{k\ell}$ and $\displaystyle \bigcup_{k\in \ZZZ} \widehat{C}_{k\ell} \subseteq \plhat$ for all $\ell$ (see Figure~\ref{fig:grown_shrunk_patches}).
\end{prop}

\begin{proof}
We prove the second relation, the first is analogous. The key is that $\widehat{C}_{k\ell} = C_k \cap \Gr(B_{k\ell},\Sigma)$, and since $C_k$ lies in the layer of thickness $\varepsilon$ between $\Upsilon_{k-1}$ and $\Upsilon_k$, any point of $B_{k\ell}$ can grow, in projection $\pi$, at most to a ball of radius $\varepsilon$. That is,
\begin{align}
\pi\Bigl(\bigcup_k \widehat{C}_{k\ell} \Bigr)
&=\bigcup_k \pi \bigl( C_k \cap \Gr(B_{k\ell},\Sigma) \bigr) 
= \bigcup_k\bigcup_{x\in B_{k\ell}} \pi\bigl( C_k \cap \past(x) \bigr) \nonumber\\
&\subseteq \bigcup_k \bigcup_{x\in B_{k\ell}} \mathcal{B}_\varepsilon(\pi(x)) 
= \bigcup_k \bigcup_{x\in C_{k\ell}} \mathcal{B}_\varepsilon(\pi(x)) \nonumber\\
&= \bigcup_{x\in \patch_{\ell}} \mathcal{B}_\varepsilon(\pi(x)) = \pi(\plhat)\,.
\end{align}
\end{proof}

\begin{prop}
\label{prop:mpproperties}
	For every $\varepsilon > 0$ the following inclusions hold:
	\begin{align}
		\mplcheck \cap \emptyset(\dpepsilon) &\subseteq \mslcheck, \label{eq:inclusion1}\\
		\mslhat &\subseteq \mplhat. \label{eq:inclusion2}		
	\end{align}
\end{prop}

\begin{proof}
Let us write $s_\ell=(s_{k\ell})_{\kappa\leq k \leq K}$ and $s_\ell : L_\ell$ if a sequence $s_\ell$ is compatible with $L_\ell$ for a particular $\ell$. So, $s:L$ is equivalent to $\forall \ell: s_\ell : L_\ell$, and $\bigcup_{s:L}$ can be replaced by $\bigcup_{s_1:L_1} \cdots \bigcup_{s_r:L_r}$. 

Now consider one $\ell$ and suppose $s_\ell: L_\ell$. If $L_\ell=0$, then $s_{k\ell}=0$ for all $k$, so 
\be
\bigcap_k \mcsklhat \Big|_{s_{k\ell}=0}= \bigcap_k \emptyset(\cklcheck) = \emptyset\Bigl( \bigcup_k \cklcheck \Bigr) \subseteq \emptyset(\plcheck) = \mpllhat
\ee
by Proposition~\ref{prop:PlCkl}, and thus,
\be
\bigcup_{s_\ell:L_\ell=0} \bigcap_k \mcsklhat \subseteq \mpllhat\,.
\ee
(The union is actually trivial, that is, there is just one term, as there is just one $s_\ell$ with all $s_{k\ell}=0$.)
Likewise, 
$\mcsklcheck=\emptyset(\cklhat)\supseteq \emptyset(\plhat)=\mpllcheck$ for all $k$ by Proposition~\ref{prop:PlCkl}, and thus,
\be
\bigcup_{s_\ell:L_\ell=0} \bigcap_k \mcsklcheck \supseteq \mpllcheck \supseteq \mpllcheck \cap \emptyset(\dpepsilon) \,.
\ee
On the other hand, if $L_\ell=1$, then $s_{k\ell}=1$ for some $k=k_0$, so, by Proposition~\ref{prop:PlCkl} again,
\be
\bigcap_k \mcsklhat \subseteq \mcsklhat \Big|_{k=k_0}
= \exists(\cklhat \big|_{k=k_0}) \subseteq 
\exists \Bigl( \bigcup_k \cklhat \Bigr) \subseteq \exists (\plhat)=\mpllhat
\ee
and thus,
\be
\bigcup_{s_\ell:L_\ell=1} \bigcap_k \mcsklhat \subseteq \mpllhat\,.
\ee
Likewise, $\mcsklcheck = \exists(\cklcheck)$ for some $k=k_0$, so, by Proposition~\ref{prop:PlCkl} again,
\begin{align}
\bigcup_{s_\ell:L_\ell=1} \bigcap_k \mcsklcheck 
&= \bigcup_{k_0} \bigcup_{\substack{s_{k\ell}=0,1\\ s_{k_0,\ell}=1}} \bigcap_k \mcsklcheck \nonumber\\
&= \bigcup_{k_0} \biggl[ \mcsklcheck\Big|_{k=k_0} \cap  \bigcap_{k\neq k_0}\Bigl( \emptyset(\cklhat)\cup \exists(\cklcheck) \Bigr) \biggr] \nonumber\\
&\supseteq \bigcup_{k_0} \biggl[ \mcsklcheck\Big|_{k=k_0} \cap  \emptyset(\dpepsilon)  \biggr] \nonumber\\
&= \bigcup_{k} \exists(\cklcheck) \cap \emptyset(\dpepsilon) \nonumber\\
&= \exists \Bigl( \bigcup_k \cklcheck \Bigr)  \cap \emptyset(\dpepsilon) \nonumber\\ 
&\supseteq \exists (\plcheck) \cap \emptyset(\dpepsilon) =\mpllcheck \cap \emptyset(\dpepsilon)\,.
\end{align}
Thus, the same relations hold in the cases $L_\ell=0$ and $L_\ell=1$. Taking the intersection over $\ell$, we obtain the Proposition.
\end{proof}

\begin{proof}[\textbf{Proof of Theorem \ref{thm:prob}}]
Let $\Prob^{\psi_\Sigma}(\cdot)$ denote the probability measure $\langle \psi_\Sigma | P_\Sigma(\cdot) | \psi_\Sigma \rangle$ on $\Gamma(\Sigma)$.
The relations \eqref{eq:inclusion1}, \eqref{eq:inclusion2} imply $\Prob^{\psi_\Sigma}(\mplcheck\cap \emptyset(\dpepsilon)) \leq  \Prob^{\psi_\Sigma}(\mslcheck)$ and $ \Prob^{\psi_\Sigma}(\mslhat) \leq \Prob^{\psi_\Sigma}(\mplhat)$ for all $\varepsilon > 0$. Hence, by \eqref{eq:probbounds},
	\be
		\Prob^{\psi_\Sigma}\Bigl(\mplcheck \cap \emptyset(\dpepsilon) \Bigr) 
		\leq \Prob^{\psi,\varepsilon}_{\det,\sP}(L) 
		\leq  \Prob^{\psi_\Sigma}\Bigl(\mplhat\Bigr).
		\label{eq:probbounds2}
	\ee
	At the same time, the relations \eqref{eq:mpinclusions1} yield:
	\be
		\Prob^{\psi_\Sigma}\Bigl(\mplcheck \cap \emptyset(\dpepsilon) \Bigr) 
		\leq \Prob^{\psi_\Sigma}(\msl) 
		\leq  \Prob^{\psi_\Sigma}\Bigl(\mplhat \Bigr).
		\label{eq:probbounds3}
	\ee
	To establish the claim of the theorem, $\displaystyle \lim_{\varepsilon \rightarrow 0} \Prob^{\psi,\varepsilon}_{\det,\sP}(L) = \Prob^{\psi_\Sigma}(\msl)$,
	it thus suffices to show that
	\be
		\lim_{\varepsilon \rightarrow 0} \Prob^{\psi_\Sigma}\Bigl(\mplcheck \cap \emptyset(\dpepsilon) \Bigr) 
		= \Prob^{\psi_\Sigma}(\msl) 
		= \lim_{\varepsilon \rightarrow 0} \Prob^{\psi_\Sigma}\Bigl(\mplhat \Bigr).
		\label{eq:problimits}
	\ee
By \eqref{monotonepatch} and \eqref{eq:mpinclusions2}, $\mplcheck \cap \emptyset(\dpepsilon)$ [respectively, $\mplhat$] is a set that depends on $\varepsilon>0$ as a decreasing [increasing] function, so its measure is also decreasing [increasing] function of $\varepsilon$. Now, any decreasing [increasing] bounded function $f(\varepsilon)$ of $\varepsilon>0$ possesses a limit as $\varepsilon\to 0$, and
\be
\lim_{\varepsilon\to 0} f(\varepsilon) = \sup_{\varepsilon>0} f(\varepsilon) = \lim_{n\to \infty} f(1/n)
\ee
[with sup replaced by inf if $f$ is increasing]. So the limits in \eqref{eq:problimits} exist, and 
we need to show that they are both equal to $\Prob^{\psi_\Sigma}(\msl)$.
By $\sigma$-continuity of measures and the monotonicity of the sets,
	\be
		\lim_{n \rightarrow \infty} \Prob^{\psi_\Sigma}\Bigl(\check{M}^{1/n}_{\patch}(L) \cap \emptyset(\partial \patch^{1/n}) \Bigr) 
		= \Prob^{\psi_\Sigma} \Bigl( \bigcup_{n\in\NNN} \bigl[\check{M}^{1/n}_\patch(L) \cap \emptyset(\partial \patch^{1/n}) \bigr] \Bigr)\,.
	\ee
Using again that the set is a decreasing function of $\varepsilon$,
	\be
		 \bigcup_{n\in\NNN} \bigl[\check{M}^{1/n}_\patch(L) \cap \emptyset(\partial \patch^{1/n}) \bigr]
		=  \bigcup_{\varepsilon>0} \bigl[\mplcheck \cap \emptyset(\dpepsilon) \bigr] \,.
	\ee
Thus,
	\begin{align}
		\lim_{\varepsilon \rightarrow 0} \Prob^{\psi_\Sigma}\Bigl(\mplcheck \cap \emptyset(\dpepsilon) \Bigr) 
		&= \Prob^{\psi_\Sigma} \Bigl( \bigcup_{\varepsilon>0} \bigl[\mplcheck \cap \emptyset(\dpepsilon) \bigr] \Bigr) \nonumber\\
		\lim_{\varepsilon \rightarrow 0} \Prob^{\psi_\Sigma}\Bigl(\mplhat \Bigr)
		&= \Prob^{\psi_\Sigma} \Bigl( \bigcap_{\varepsilon>0} \mplhat \Bigr)\,.
		\label{eq:sigmalimits}
	\end{align}
Let $q\in \bigcap_{\varepsilon>0} \mplhat \setminus \msl$. By the definition of the sets $\msl$ and $\mplhat$, there exist $x \in q$ and $1 \leq \ell \leq r$ such that $\pi(x) \in \partial \pi(\patch_\ell)$.
Thus, as $\pi_\Sigma : \Sigma \rightarrow \RRR^3$ is a homeomorphism, 
\be
	\bigcap_{\varepsilon>0} \mplhat \setminus \msl  \subseteq S_\partial,
\ee
where
\be
	S_\partial := \Bigl\{ q \in \Gamma(\Sigma) : \exists \ell : q \cap \partial \patch_\ell \neq \emptyset\Bigr\}.
\ee
Now, as $\mu_\Sigma(\partial \patch_\ell) = 0 \, \forall \ell$ by the definition of an admissible partition, $\mu_{\Gamma(\Sigma)}(S_\partial) = 0$. Hence, by definition of hypersurface evolution, $P_\Sigma(S_\partial) = 0$, so $\Prob^{\psi_\Sigma}(S_\partial) = 0$ and
\be
\Prob^{\psi_\Sigma} \Bigl( \bigcap_{\varepsilon>0} \mplhat \Bigr) = \Prob^{\psi_\Sigma} \Bigl(\msl \Bigr)\,.
\ee
The argument works in the same way for $\displaystyle \bigcup_{\varepsilon>0} \bigl[\mplcheck \cap \emptyset(\dpepsilon) \bigr]$.
\end{proof}

\section{Conclusions and Outlook} \label{sec:conclusions}

In this paper, we have shown that one can, starting only from measurement postulates for equal times (i.e., on horizontal surfaces), justify the \textit{curved Born rule} for arbitrary Cauchy surfaces $\Sigma$. To ``justify'' means to give a suitable precise definition of a ``detection process'' and a rigorous mathematical derivation of the curved Born rule using this definition as well as a set of assumptions about the underlying wave function dynamics. 
We use two essential assumptions: ``interaction locality'' (IL) and ``propagation locality'' (PL). 
Our main result (Theorem \ref{thm:prob}) demonstrates, among other things, the direct physical significance of the hypersurface wave function $\psi_\Sigma$ as a probability amplitude in detection experiments. 
The curved Born rule has been often taken for granted in the literature, but to the best of our knowledge a justification had not been given before. 

Our research gives rise to two further questions. Firstly, one can well imagine a different setup for the detection process than ours (which aims at using measurement postulates only for horizontal surfaces). For example, assuming Poincar\'e invariance of the theory,  one may approximate $\Sigma$ by a surface that is piecewise flat (but non-horizontal) as in Figure~\ref{fig:detectiondefs}(a). 
We conjecture that such an approximation scheme leads to the same curved Born rule while avoiding the need for excluding double detections.

Secondly, the fact that the detection probabilities along a curved Cauchy surface $\Sigma$ are determined already by the evolution (unitary and through collapses) of the wave function $\psi_{\Sigma_t}$ on \emph{horizontal} surfaces suggests 
another conjecture: that also the hypersurface wave function $\psi_\Sigma$ is already determined by the evolution of $\psi_{\Sigma_t}$ on the horizontal surfaces together with the postulates (IL) and (PL).

These two questions shall be addressed in subsequent work.

\bigskip

\noindent{\it Acknowledgments.}
We thank Detlev Buchholz, Carla Cederbaum, Eddy Keming Chen, Sheldon Goldstein, S\"oren Petrat, Nicola Pinamonti, Reiner Sch\"atzle, and Stefan Teufel for helpful discussions.\\[1mm] 
\begin{minipage}{15mm}
\includegraphics[width=13mm]{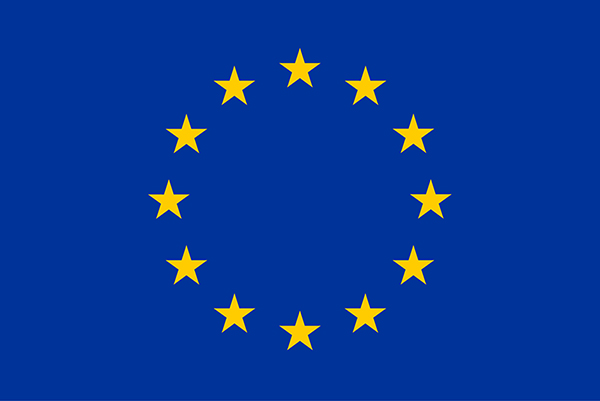}
\end{minipage}
\begin{minipage}{143mm}
This project has received funding from the European Union's Framework for Research and Innovation Horizon 2020 (2014--2020) under the Marie Sk{\l}odowska-
\end{minipage}\\[1mm]
Curie Grant Agreement No.~705295.

\end{document}